%% file: A_main.tex
\documentclass[11pt]{article}
\usepackage[margin=1in]{geometry}
\usepackage{amsmath}
\usepackage{amssymb}
\usepackage{amsthm}
\usepackage{enumitem}
\usepackage{bbold}
\usepackage[toc,page]{appendix}
\usepackage{xcolor}
\usepackage{hyperref}
\usepackage{mathtools}
\usepackage{graphicx}
\usepackage{float}
\usepackage{caption}

\usepackage[
  backend=biber,
  style=alphabetic,
  maxbibnames=99,
  maxalphanames=99
]{biblatex}

\newtheorem{theorem}{Theorem}[section]
\newtheorem{definition}[theorem]{Definition}
\newtheorem{claim}[theorem]{Claim}
\newtheorem{lemma}[theorem]{Lemma}

\newtheorem{corollary}[theorem]{Corollary}
\newtheorem{fact}[theorem]{Fact}

\newtheorem{algorithm}[theorem]{Algorithm}

\newenvironment{subproof}[1][\textup{\large\textsf{Proof:}}]{\begin{proof}[#1]}{\end{proof}}

\newtheoremstyle{sltheorem}
{}                
{}                
{\it}        
{}                
{\sffamily\large}       
{.}               
{ }               
{}                
\theoremstyle{sltheorem}
\newtheorem{subclaim}[theorem]{Claim}

\def\ov{\overline}
\newcommand{\ovbar}[1]{\mkern 1.5mu\overline{\mkern-1.5mu#1\mkern-1.5mu}\mkern 1.5mu}
\def\th{^{\text{th}}}
\def\eqdef{\stackrel{\text{def}}{=}}
\DeclareMathOperator{\E}{\mathbb{E}}

\begin{document}

\title{Near-Optimal Bounds for Testing Residual-String Equality \\ and Parenthesis Languages}
\author{Hadar Strauss \\ \textit{Weizmann Institute of Science}}
\maketitle

\input{B_abstract}
\newpage

\tableofcontents
\newpage

\input{C_intro}
\input{D_lb_overview}

\input{E_ub_overview}

\input{F_preliminaries}

\input{Ga_lb_proof_nonadaptive}
\input{Gb_lb_proof_adaptive}

\input{H_ub_proof}

\input{I_ub_dyck}

\input{J_acknowledgments}

\printbibliography[heading=bibintoc]{}

\newpage
\begin{appendices}

\input{K_equal_length_reduction}

\input{L_sd_fact}

\end{appendices}

\end{document}

%% file: B_abstract.tex
\begin{abstract}

Residual-String Equality, denoted $\texttt{ResStringEq}$, is the property consisting of all pairs of strings over $\{0,1,*\}$ that are equal after deleting all `$*$' symbols from them.
This property was first introduced by Fischer, Magniez, and Starikovskaya (SODA 2018), who used it to show a lower bound on testing the $\texttt{Dyck}$ languages, where $\texttt{Dyck}_m$ is the language consisting of balanced sequences of parentheses over $m$ parenthesis types.
They showed that testing $\texttt{ResStringEq}$ on inputs of length $n$ requires $\Omega(n^{1/5})$ queries, and presented a reduction from testing $\texttt{ResStringEq}$ to testing $\texttt{Dyck}_m$ for any $m \geq 2$.
Furthermore, they showed that $\texttt{Dyck}_m$ can be tested with $O(n^{2/5+\delta})$ queries for every constant proximity parameter, where $\delta>0$ is an arbitrarily small constant.

In this work, we nearly close the gap, by showing that testing $\texttt{ResStringEq}$, and hence $\texttt{Dyck}_m$ (for any $m\geq 2$), requires $\Omega(n^{2/5})$ queries.
We also show a stronger lower bound of $\Omega(\sqrt{n})$ for testers that make non-adaptive queries.
We establish that the $\Omega(\sqrt{n})$ bound is nearly tight, by presenting a non-adaptive tester for $\texttt{ResStringEq}$ that uses $O(n^{1/2+\delta})$ queries for an arbitrarily small constant $\delta>0$. 
Furthermore, we extend this non-adaptive tester to the $\texttt{Dyck}$ languages, with the same query complexity.
Finally, we improve the dependence on the proximity parameter $\epsilon$ in the tester of Fischer, Magniez, and Starikovskaya, reducing it from $O(1/\epsilon)^{\mathrm{poly}(1/\delta)}$ to $O(1/\epsilon)^{O(\log(1/\delta))}$.

\end{abstract}

%% file: C_intro.tex
\section{Introduction}

Property testing~\cite{GGR98,RS96} studies a relaxed notion of decision problems. 
Rather than deciding exact membership, a property tester is only required to distinguish (w.h.p.) between objects in the property and objects that are ``\emph{$\epsilon$-far}'' from the property, where $\epsilon>0$ is a given proximity parameter.
An object $x\in \Sigma^n$ is considered $\epsilon$-far from a property $\Pi_n \subseteq \Sigma^n$ if it differs from any object in the property on more than $\epsilon \cdot n$ positions.

Standard decision problems generally require reading the entire input, since flipping a single bit can change the decision.
In contrast, relaxed decision opens the possibility of algorithms that (probabilistically) read only a sub-linear portion of the input. 
Thus, in property testing, the input is viewed as a huge object to which the tester gets only oracle access, and the main focus is on minimizing the query complexity.

Perhaps the most basic property is the \texttt{Equality} property, defined as $\{ (s,s') \in \{0,1\}^n\times\{0,1\}^n : s = s' \}$. This property is easy to test by comparing the two input strings at $O(1/\epsilon)$ random positions.
The focus of this work is on the following natural \emph{variant} of $\texttt{Equality}$, introduced by~\cite{FMS18}, 
which is significantly harder to test.
In this variant, the tester is given two strings over $\{0,1,*\}$ and is required to decide whether the two strings are equal after deleting all `$*$' symbols from them.

\begin{definition}[Residual-String Equality (a.k.a.\ True-String Equality~\cite{FMS18})]\label{def:rse}
    For every string $s \in \{0,1,*\}^n$, let $\mathrm{Res}(s)$ denote the residual string obtained from $s$ after deleting all `$\,*$' symbols from it.
    The property \textup{\texttt{ResStringEq}} is the set of all pairs of strings $(s,s') \in \{0,1,*\}^n\times\{0,1,*\}^{n}$ such that $\mathrm{Res}(s)= \mathrm{Res}(s')$.\footnote{
        We could also allow $s$ and $s'$ to differ in length.
        Note that $\epsilon$-testing this more general definition reduces to $\epsilon/2$-testing the equal-length case, by padding the shorter string with `$*$' symbols. See Appendix~\ref{apdx:equal_length} for details. 
    }
\end{definition}

The property \texttt{ResStringEq}
was first introduced by~\cite{FMS18} as a means towards proving a lower bound on the complexity of testing the family of \texttt{Dyck} languages,
where $\texttt{Dyck}_m$ is the (context-free) language consisting of all balanced sequences of parentheses over $m$ parenthesis types.
The \texttt{Dyck} languages play an important role in the theory of context-free languages (see the Chomsky--Sch{\"{u}}tzenberger representation theorem~\cite{Chomsky63}).
They were first considered in the context of property testing by~\cite{AKNS01}, who initiated a general study of testing formal languages. 
In particular,~\cite{AKNS01} showed that
testing $\texttt{Dyck}_1$ can be done using $\widetilde{O}(1/\epsilon^2)$ queries, whereas testing $\texttt{Dyck}_2$ requires at least $\Omega(\log n)$ queries.\footnote{Note that a lower bound on testing $\texttt{Dyck}_2$ implies a lower bound on testing $\texttt{Dyck}_m$ for any $m \geq 2$.}
Testing the $\texttt{Dyck}$ languages then became the focus of~\cite{PRR03}, who showed that, for every $m \geq 2$, testing $\texttt{Dyck}_m$ requires 
$\Omega(n^{1/11})$ queries and can be done using $\widetilde{O}(n^{2/3})$ queries for every constant $\epsilon>0$.
These bounds were then improved by~\cite{FMS18}
to $\Omega(n^{1/5})$ and $O(n^{2/5+\delta})$, respectively, where $\delta > 0$ is an arbitrarily small constant. 
To prove their lower bound,~\cite{FMS18} first showed a simple reduction from testing \texttt{ResStringEq} to testing $\texttt{Dyck}_2$,
and then showed a lower bound of $\Omega(n^{1/5})$ on testing \texttt{ResStringEq}.\footnote{
    Specifically, given input $(s,s')$ to \texttt{ResStringEq}, the reduction runs the tester for $\texttt{Dyck}_2$ while emulating its queries to a parentheses sequence defined as follows.
    First, concatenate $s$ with the \emph{reverse} of $s'$. 
    Then, replace each `$0$' in $s$ by `$(($' and each `$1$' by `$[[$'. In the reverse of $s'$, replace each `$0$' by `$))$' and each `$1$' by `$]]$'. 
    Finally, replace all `$*$' symbols by `$()$'.
    (Note that each symbol in $(s,s')$ is replaced by two symbols in the parentheses-expression so that the correspondence between locations in the \texttt{Dyck} and \texttt{ResStringEq} problems is oblivious of the input strings.) 
}

\subsection{Our results}\label{subsec:intro:our_results}

Our main result is nearly closing the remaining gap, by showing:

\begin{theorem}\label{thm:main}
    Testing \textup{\texttt{ResStringEq}} requires $\Omega(n^{2/5})$ queries. 
\end{theorem}

As a corollary, we obtain:

\begin{corollary}\label{cor:dyck_lower_bound}
    For every $m\geq 2$, testing $\textup{\texttt{Dyck}}_m$ requires $\Omega(n^{2/5})$ queries. 
\end{corollary}

Both of these results are nearly tight, by the $O(n^{2/5+\delta})$-query tester given in~\cite{FMS18} for the $\texttt{Dyck}$ languages, and the reduction they showed from testing \texttt{ResStringEq} to testing $\textup{\texttt{Dyck}}_2$.
In particular, these results show that although \texttt{ResStringEq} is a seemingly easier problem than $\texttt{Dyck}_2$, both problems actually have similar complexity.  
We stress that both lower bounds hold for testers with two-sided error.\footnote{
    Note that~\cite{AKNS01} showed that one-sided error testers for the \texttt{Dyck} languages require $\Omega(n)$ queries, even for $\texttt{Dyck}_1$.
    Furthermore, the same linear lower bound holds for \texttt{ResStringEq}: Let $k\eqdef 2\cdot\epsilon\cdot n+1$, and consider the pair $(s,s')$ where $s=*^n$ and $s'=0^k*^{n-k}$. 
    Clearly $(s,s')$ is $\epsilon$-far from \texttt{ResStringEq}. 
    However, if the tester queries at most $n-k$ positions of $s$, then the query answers are consistent with the {\scriptsize YES}-instance obtained by changing $k$ of the unqueried positions in $s$ from `$*$' to `$0$'. 
}

\paragraph{The non-adaptive setting.}

An important type of testers are testers that make non-adaptive queries. 
A tester makes \textsf{non-adaptive} queries if its queries to the input do not depend on the answers to prior queries. 
Non-adaptive testers allow making all queries to the input in \emph{one batch}, which is useful in settings in which 
each round of accessing the input has a large overhead.

Both the earlier $\widetilde{O}(n^{2/3})$-query tester of~\cite{PRR03} and the later $O(n^{2/5+\delta})$-query tester of~\cite{FMS18} use \emph{adaptive} queries.
We show that non-adaptive testers for \texttt{ResStringEq} (and hence also for the $\texttt{Dyck}_m$ languages, for any $m\geq2$) require $\Omega(\sqrt{n})$ queries.

\begin{theorem}\label{thm:non_adaptive_lower_bound}
    Any \emph{non-adaptive} tester for \textup{\texttt{ResStringEq}} must make at least $\Omega(\sqrt{n})$ queries.  
\end{theorem}

Theorem~\ref{thm:non_adaptive_lower_bound} shows the necessity of adaptivity in the $O(n^{2/5+\delta})$-query tester of~\cite{FMS18}. 
We show that the $\Omega(\sqrt{n})$ bound is nearly tight, by designing a nearly matching non-adaptive tester for \texttt{ResStringEq}:

\begin{theorem}\label{thm:upper_bound}
    For any \textup{(}arbitrarily small\textup{)} constant $\delta>0$, and every constant $\epsilon>0$, there exists a \emph{non-adaptive} $\epsilon$-tester for \textup{\texttt{ResStringEq}} that makes $O(n^{1/2+\delta})$ queries.
\end{theorem}

Building on this tester, we obtain:

\begin{theorem}\label{thm:upper_bound:dyck}
    For every $m\in \mathbb{N}$, every \textup{(}arbitrarily small\textup{)} constant $\delta>0$, and every constant $\epsilon>0$, there exists a \emph{non-adaptive} $\epsilon$-tester for $\textup{\texttt{Dyck}}_m$ that makes $O(n^{1/2+\delta})$ queries.
\end{theorem}

\paragraph{The dependence of the query complexity on the proximity parameter.}
The $O(n^{2/5+\delta})$ query complexity achieved by~\cite{FMS18} comes at a high cost in terms of the dependence on $\epsilon$, namely, the dependence on $\epsilon$ in their tester is super-exponential in $1/\delta$. 
More precisely, their tester has query complexity
$O(\left(1/\epsilon\right)^{\textup{\textrm{poly}}(1/\delta)} \cdot n^{2/5 + \delta})$. 
We show that this super-exponential dependence can be avoided. 
Specifically, through a more careful analysis and parameter choice, we achieve query complexity $O(\left(1/\epsilon\right)^{O(\log(1/\delta))} \cdot n^{1/2 + \delta})$
for the \emph{non-adaptive} testers in Theorems~\ref{thm:upper_bound} and~\ref{thm:upper_bound:dyck}.
We show that this improved $\epsilon$ dependence can be carried over to the tester of~\cite{FMS18}, 
yielding \emph{adaptive} testers for both \texttt{ResStringEq} and the $\texttt{Dyck}$ languages with query complexity $O(\left(1/\epsilon\right)^{O(\log(1/\delta))} \cdot n^{2/5 + \delta})$ (see Theorems~\ref{thm:upper_bound:adaptive} and~\ref{thm:upper_bound_dyck:adaptive}, respectively).

\paragraph{Concurrent work.}

The $\Omega(n^{2/5})$ lower bound for testing $\texttt{Dyck}_m$ (for any $m\geq 2$), as well as the $\Omega(\sqrt{n})$ lower bound for non-adaptive testers, were independently obtained by~\cite{JPR27} in concurrent work.
Their lower-bound techniques are similar to ours.

\subsection{Organization}
A technical overview of our results is presented in Section~\ref{sec:technical_overview}, where Subsection~\ref{sec:overview:lb} contains an overview of the $\Omega(n^{2/5})$ lower bound (Theorem~\ref{thm:main}) and Subsection~\ref{sec:overview:up} contains an overview of the non-adaptive tester for \texttt{ResStringEq} (Theorem~\ref{thm:upper_bound}). 
Basic preliminaries are presented in Section~\ref{sec:preliminaries}.
The proof of the $\Omega(n^{2/5})$ lower bound is presented in Section~\ref{sec:lb_proof}, which also includes a proof of the $\Omega(\sqrt{n})$ lower bound for the non-adaptive case (see Subsection~\ref{subsec:non_adaptive_proof}).
The non-adaptive tester for \texttt{ResStringEq} and its analysis are presented in Section~\ref{sec:ub_proof}, where Subsection~\ref{subsec:ub_proof:warmup} begins with a warm-up towards this tester, and the actual tester is presented in Subsection~\ref{subsec:upper_bound:proof:actual}.
In Section~\ref{subsec:dyck_extension} we extend the non-adaptive tester to the \texttt{Dyck} languages.
The improved dependence on $\epsilon$ in the adaptive testers for \texttt{ResStringEq} and the \texttt{Dyck} languages is discussed in Subsections~\ref{subsec:adaptive_RSE} and~\ref{subsec:adaptive_dyck}, respectively.

%% file: D_lb_overview.tex
\section{Technical overview}\label{sec:technical_overview}

\subsection{The lower bound}\label{sec:overview:lb}

We prove the lower bound for \texttt{ResStringEq} by the standard indistinguishability method.
That is, we construct two distributions over input strings: a {\small YES} distribution, which is over inputs in \texttt{ResStringEq}, and a {\small NO} distribution, which is concentrated over inputs \emph{far} from \texttt{ResStringEq}. 
We then show that any tester making $o(n^{2/5})$ queries cannot distinguish between inputs drawn from these two distributions.

The proof builds on ideas developed in~\cite{FMS18} and~\cite{PRR03}.\footnote{
    Although the lower bound in~\cite{PRR03} is presented for the \texttt{Dyck} languages, their proof implicitly establishes the lower bound for \texttt{ResStringEq}.
} 
We give an overview of the proof in two parts. In the first part (presented in Section~\ref{subsec:overview:1/4}), we explain how to obtain a lower bound of $\Omega(n^{1/4})$.\footnote{Recall that the previous best known lower bound was~$\Omega(n^{1/5})$~\cite{FMS18}.} 
This improved lower bound is obtained by a tighter analysis of (essentially) the same construction of the {\small YES} and {\small NO} distributions used in~\cite{FMS18} and~\cite{PRR03}.
In the second part of the overview, presented in Section~\ref{subsec:overview:blocks}, we explain how we \emph{modify} the construction to further improve the bound to $\Omega(n^{2/5})$ (while using the improved analysis from the previous section).
We note that modifying the construction is necessary, as the distributions in the first construction can be distinguished using $\widetilde{O}(n^{1/4})$ queries (see details at the end of Section~\ref{subsec:overview:1/4}).

\subsubsection{An $\Omega(n^{1/4})$ lower bound}\label{subsec:overview:1/4}

\paragraph{The two distributions.}
We construct two random pairs of strings over $\{0,1,*\}^n$: a {\small YES} pair $(S,S')$ and a {\small NO} pair $(\widetilde{S},\widetilde{S}')$.
Roughly speaking, in both distributions, we choose in each string independently $n/2$ random positions for the star (i.e.,\ `$*$') symbols.
We then fill the remaining positions with uniform random bits, such that in the {\small YES} distribution we use the same values in both strings, 
whereas in the {\small NO} distribution, we use the same values only in the ``extremities'' of the two strings, and use independent values in the ``middle'' of the strings 
(specifically, in the middle third of the non-star positions).

More precisely, we first choose the locations of the star symbols identically in the two distributions, as follows.
We begin with the first string in the pair (i.e., $S$ in the {\small YES} pair and $\widetilde{S}$ in the {\small NO} pair).
We draw $n$ values $X = (X_1,\ldots,X_n) \in\{0,1\}^n$, and place a star symbol at every position $i\in[n]$ where $X_i=0$.
To ensure that exactly $n/2$ positions contain star symbols,\footnote{
    Note that in order for the {\scriptsize YES} pair to be in \texttt{ResStringEq}, we require that both strings in the pair contain exactly the same number of star positions.
} we draw the first $n/2$ values $X_1,\ldots,X_{n/2}$ independently and uniformly from $\{0,1\}$, and then set the other $n/2$ mirror positions to the complementary values; that is, we set $X_{n/2+i} = 1-X_{n/2+1-i}$ for every $i\in [n/2]$. 
Independently, we do the same for the second string (i.e., $S'$ and $\widetilde{S}'$), using fresh values $X' = (X'_1,\ldots,X'_n) \in \{0,1\}^n$. 
We stress that we use the same values $X$ and $X'$ to construct both the {\small YES} and {\small NO} pairs.

For every $i\in[n]$ such that $X_i = 1$ (i.e., $i$ is not a star position), we define its \textsf{rank} in the first string, denoted $R(i)$, as the number of non-star positions up to and including $i$; that is, $R(i) \eqdef \sum_{j\in[i]}X_j$. 
If $i$ is a star-position (i.e., $X_i = 0$), we define $R(i) = \bot$. 
Similarly, we let $R'(i)$ denote the rank of $i$ in the second string, defined analogously with respect to $X'$.

We now fill the remaining positions in the strings.
In the {\small YES} pair $(S, S')$, for every $r\in[n/2]$, we draw a uniform bit $V_r\in\{0,1\}$ and place it in the two positions with rank $r$; that is, we set $S_i=S'_{i'}=V_r$, where $i$ and $i'$ are the positions such that $R(i) = R'(i') = r$.
In the {\small NO} pair $(\widetilde{S}, \widetilde{S}')$, we do the same, except that on the middle block of ranks we use independent bits in the two strings. 
Specifically, let $M= \left(\frac{n}{6}, \frac{n}{3}\right]$ be the middle third of the ranks. Then, for ranks $r \in [n/2] \setminus M$ we use a shared uniform bit $V_r$ exactly as before (i.e., we set $\widetilde{S}_i=\widetilde{S}'_{i'}=V_r$, where $R(i) = R'(i')=r$), while for ranks $r \in M $ we draw two independent uniform bits $V_r, V_r' \in \{0,1\}$, and use $V_r$ in the first string and $V_r'$ in the second (i.e., we set $\widetilde{S}_i=V_r$ and $\widetilde{S}'_{i'}=V'_r$, where $R(i) = R'(i')=r$).

\paragraph{The indistinguishability claim.}
We say that $i\in[n]$ and $i'\in[n]$ are \textsf{matching positions} if $R(i) = R'(i') \neq \bot$.
We say that $i$ and $i'$ are \textsf{interior} matching positions if $R(i) = R'(i') \in M$.
First, we argue that as long as the tester does not query interior matching positions, it will not be able to distinguish between the two distributions. 
This holds because, if the tester did not query interior matching positions, then \emph{in both distributions} the answer to any query whose rank is in $M$ is a uniform bit, independent of all other answers.
On the other hand, we show that the probability of querying interior matching positions is at most $O\!\left(\frac{q^{2}}{\sqrt{n}}\right)$, where $q$ is the query complexity of the tester. 
This implies that as long as $q=o(n^{1/4})$, the two views are indistinguishable, establishing the lower bound. 
We proceed to explain how we bound the probability of querying interior matching positions.
We will rely on a standard anti-concentration bound that states that for any $j\in \mathbb{N}$, the probability of getting exactly $j$ heads in $t$ independent coin tosses is at most $O\Big(\frac{1}{\sqrt{t}}\Big)$.
In other words, the sum of $t$ independent unbiased Bernoulli random variables $X_1,\ldots,X_t$ satisfies 
$\max_{j\in\mathbb{N}}\left\{\Pr\!\big[ \sum_{i\in[t]}{X_i} = j \big]\right\} = O\Big(\frac{1}{\sqrt{t}}\Big)$.

\paragraph{The non-adaptive case.}
Let us first consider the case where the tester uses \emph{non-adaptive} queries.
First, note that if $R(i) \in M = \left(\frac{n}{6}, \frac{n}{3} \right]$, then it must hold that $i \in \left( \frac{n}{6}, \frac{5n}{6} \right]$ (because there must be at least $n/6$ positions before $i$ and at least $n/6$ positions after $i$ that have the remaining ranks). 
We claim that for every $i\in\left(\frac{n}{6},\frac{5n}{6}\right]$ and every fixed $r\in[n/2]$, it holds that $\Pr[R(i)=r]=O\!\left(\frac{1}{\sqrt{n}}\right)$. 
If $i\leq n/2$, then $R(i)=\sum_{j\in[i]}X_j$ is a sum of $i$ independent Bernoulli random variables, and hence by the anti-concentration bound stated above $\Pr[R(i)=r]=O\!\left(\frac{1}{\sqrt{i}}\right)=O\!\left(\frac{1}{\sqrt{n}}\right)$, where we use the fact that $i > n/6 = \Omega(n)$. 
The case $i>n/2$ is symmetric: since the total sum of $X_1,\ldots,X_n$ is $n/2$, we can write $R(i)=n/2-\sum_{j=i+1}^n{X_j}$, where $X_{i+1},\ldots,X_n$ are $n-i \geq n/6$ independent Bernoulli random variables.

Thus, for any pair of positions $(i,i')$, the probability that $R(i) = R'(i') \in M$ is at most $O\!\left(\frac{1}{\sqrt{n}}\right)$ (because $R(i)$ and $R'(i')$ are independent).
Since we assume the queries are non-adaptive, we can take a union bound over all pairs of queried positions, one from each string, and obtain that the probability that the tester queries interior matching positions is at most $O\!\left(\frac{q^2}{\sqrt{n}}\right)$, establishing the $\Omega(n^{1/4})$ bound for non-adaptive testers. 
(In Section~\ref{subsec:overview:blocks} we explain how we modify the construction of the {\small YES} and {\small NO} distributions to obtain a lower bound of $\Omega(\sqrt{n})$ for non-adaptive testers.)

\paragraph{The adaptive case.}

Handling \emph{adaptive} queries is more challenging.
When queries are adaptive, previously queried matching positions (that are not necessarily interior) can help the tester find further matching positions.\footnote{
    Note that, a priori, the tester may not know which of its previous queries match. 
    However, the observed bit-values can give information about such matches, because at (non-interior) matching positions the bits are shared, whereas at non-matching positions they are independent. 
    In particular, at the cost of an $O(\log n)$ factor to the query complexity, a distinguisher (of our {\scriptsize YES} and {\scriptsize NO} distributions) can determine w.h.p.\ which of its queries are non-interior matching positions, by querying $O(\log n)$ consecutive positions starting at each original query and checking whether the corresponding residual sequences match.
}
Specifically, if the tester previously queried matching positions $(i,i')$, then positions in the vicinity of $i$ and $i'$ are more likely to match as well (similarly to how there is a high chance of finding matching positions close to the extremities of the strings).
However, this probability decreases with the distance from $i$ and $i'$. 
To illustrate, suppose that we are given a pair of matching positions $(i, i')$,
and now consider a new pair of positions $(\hat{i}, \hat{i}')$ that lie on the same side of $(i,i')$; that is, $\hat{i} > i$ and $\hat{i}' > i'$, or $\hat{i} < i$ and $\hat{i}' < i'$.
Since $i$ and $i'$ are matching, the new positions $\hat{i},\hat{i}'$ are matching if and only if the number of non-star positions between $i$ and $\hat{i}$ is equal to the number of non-star positions between $i'$ and $\hat{i}'$.
By the anti-concentration bound referenced above, the probability of this event is at most $O\Big(\frac{1}{\sqrt{|i-\hat{i}|}}\Big)$.

For this reason, we will argue that in each query we can make only a limited amount of ``\emph{progress}'' from the extremities towards the interior of the string.
In order to make this more precise, we focus on the following simplified setting, which will convey the main idea in a cleaner form (we will briefly discuss how we handle the general setting at the end of this section). 
We will assume that the tester proceeds in two stages: 
In the first stage, it adaptively queries the second string, at positions denoted $i'_1,\ldots,i'_q \in [n]$. In the second stage, it adaptively queries the first half of the first string at monotonically increasing positions, denoted $i_1, \ldots,i_q \in [n/2]$ where $i_j < i_{j+1}$ for every $j\in[q-1]$.

Now, for each $k\in[q]$, we let $D_k$ denote the maximum first-string query position among $i_1,\ldots,i_k$ that matches one of the second-string query positions $i'_1,\ldots,i'_q$.   
If there is no such position, we define $D_k \eqdef 0$.
We define the \textsf{progress} of query $k$ as $P_k \eqdef D_k - D_{k-1} \geq 0$.

We first bound the probability of making a large amount of progress in a single query. More specifically, for each $p\in [n/2]$, we consider the probability that $P_k \geq p$.
Notice that if $P_k \geq p$, then the following two events must hold: 
\begin{enumerate}
    \item The $k\th$ first-string query position $i_k$ matches with some queried position in the second string. Let us denote this event by $E_1$.
    
    \item In the interval $(i_k - p, i_k]$ (i.e., the $p$ positions going back from $i_k$) there are \emph{no} previous first-string queries that match with a queried position in the second string. Let us denote this event by $E_2$.
\end{enumerate}
In order to analyze the probability that both events $E_1$ and $E_2$ hold,
we first condition on the answers to all queries made prior to $i_k$, which fixes the query positions $i'_1,\ldots,i'_q$ and $i_1,\ldots,i_{k}$.
Furthermore, we fix the star-positions in the first string up to position $i_k-p$
as well as \emph{all} star-positions in the second string 
(i.e., we fix the variables $X_i$ for all $i\leq i_k-p$, and all variables $X'_i$).

Suppose first that there are no previous queries that lie in the interval $(i_k - p, i_k]$ (i.e., including queries that do not match any queried position in the second string).
In this case, the variables $X_i$ in the interval $(i_k-p,i_k]$ are independent of the previous query answers that we conditioned on 
(because these answers concern only positions up to $i_k-p$ in the first string or positions in the second string).
In other words, fixing the previous query answers does not change the distribution of the variables in the interval $(i_k-p,i_k]$, and thus we still have $p$ independent Bernoulli variables in the sum defining the rank of $i_k$.
Therefore, the probability that $i_k$ matches any specific query position in the second string is at most $O\!\left(\frac{1}{\sqrt{p}}\right)$. 
By a union bound, the probability that there exists some second-string query position that matches with $i_k$ (i.e., event $E_1$ holds) is at most $O\!\left(\frac{q}{\sqrt{p}}\right)$. 

In general, there could be up to $q$ previous queries in the interval $(i_k-p,i_k]$.
First, note that conditioning on the previous query answers fixes the variables $X_i$ at queried positions in the interval.
More importantly, this conditioning fixes also the \emph{bit answers} to previous queries at non-star positions, which may affect the distribution of the remaining variables $X_j$ at unqueried positions in the interval.
To see why, recall that for every $r\notin M$, the two positions with rank $r$ contain a shared uniform bit $V_r$.
Suppose that the tester previously queried a non-star position $i$ in the interval $(i_k-p,i_k]$.
Then whether the observed bit $V_{R(i)}$ agrees with bits observed in the second string can give information about $R(i)$.
For example, if the observed bits differ, then the corresponding ranks cannot be equal.
Since $R(i)$ depends on the still-unfixed variables $X_j$ at positions $j<i$ in the interval, fixing the observed bit $V_{R(i)}$ 
may affect the distribution of these variables.

Nevertheless, we rely on the fact that the previous queries in the interval $(i_k-p,i_k]$ are all \emph{unmatched} (i.e., event $E_2$ holds), 
to show that fixing the previous bit answers has only a limited effect.
Specifically, we first show that after fixing all the relevant star-positions 
(i.e., all the variables $X_i$ at positions $i\leq i_k-p$ and at previously queried positions, as well as all second-string variables $X'_i$), we have:
\begin{align}\label{eq:overview:dependence}
    \Pr\left[\, E_1 \wedge E_2 \;|\; \ovbar{V} = \ovbar{v} \,\right] \leq \frac{\Pr[E_1]}{\Pr[E_2]}
\end{align}
where $\ovbar{V} = \ovbar{v}$ denotes the fixing of all previous bit answers (from both strings). 
The main observation needed to prove Eq.~\eqref{eq:overview:dependence} is that no matter how we fix the remaining variables $X_j$ such that all queries in the interval $(i_k-p,i_k]$ are unmatched, the query answers in the interval at non-star positions are always uniform bits, independent of all other answers. 
In other words, in any such fixing, the distribution of the previous bit answers \emph{remains the same}.\footnote{
    We stress that we do not \emph{actually} fix the remaining variables, we only consider all such fixings as part of proving Eq.~\eqref{eq:overview:dependence}.
}
This implies the conditional probability bound in Eq.~\eqref{eq:overview:dependence}; see Claim~\ref{claim:removing_answers_condition} for details. 

Thus, we will aim to show that $\Pr[E_2] = \Omega(1)$. 
Intuitively, we can use the same anti-concentration argument that we have used so far to upper bound the probability of the complement event $\ov{E}_2$
(i.e., the event that some query in the interval $(i_k-p,i_k]$ is matched with a second-string query).
However, to apply the anti-concentration argument, we need to have sufficiently many unfixed variables $X_i$ prior to each previous query in the interval.
We encounter a problem when many previous queries are ``crowded together'' near the start of the interval 
(recall that the variables $X_i$ are fixed at positions $i\leq i_k-p$, as well as at all previously queried positions).

To handle this, we use the following combinatorial observation. 
For any $\tau\in \mathbb{N}$, it is possible to partition the progress interval $(D_{k-1}, D_{k}]$ into two parts, $(D_{k-1}, C_k]$ and $(C_k, D_{k}]$, that satisfy the following two properties:

\begin{enumerate}
    \item Order the first-string query positions that lie in the second part $(C_k, D_{k}]$ from left to right, and denote the $\ell\th$ query position by $j_\ell$.
    Then, for every $\ell$, there are at least $\tau\cdot \ell$ unqueried positions strictly between $C_k$ and $j_\ell$.

    \item The size of the first part $(D_{k-1}, C_k]$ is less than $(\tau+1)\cdot N_k$, where $N_k$ is the number of first-string query positions that lie in $(D_{k-1}, D_k]$. 
\end{enumerate}

We can always obtain such a partition as follows. Start with $C_k=D_{k-1}$.
If the first property fails, move $C_k$ to an arbitrary query position for which it fails, and repeat until the property holds.
Notice that each time we move $C_k$, the interval that is skipped must be of length less than $(\tau+1)$ times the number of queried positions in that interval.
Since the skipped intervals are disjoint, the second property follows.

Keeping $\tau$ as a parameter to be set later, we denote the progress corresponding to the second part $(C_k,D_k]$ by $P'_k \eqdef D_k - C_k$, and continue our analysis with $P'_k \geq p$ instead of $P_k\geq p$.
Using the first property of the partition, we can now upper bound the probability of the event $\ov{E}_2$.
Specifically, we may assume that, for every $\ell$, there are at least $\tau\cdot\ell$ unqueried positions in $(i_k-p,i_k]$ before the $\ell\th$ query in this interval
(if this condition fails, we can increase $p$ until the condition holds; if no such $p$ exists, then the event $P'_k\geq p$ is impossible).
Hence, the probability that the $\ell\th$ query point in the interval $(i_k - p, i_k]$ matches with some second-string query point is at most $O\left(\frac{q}{\sqrt{\tau\cdot \ell}}\right)$. 
Taking a union bound over the at most $q$ queries in the interval, the probability that one of them matches some second-string query point (i.e., event $\ov{E}_2$ holds) is at most $\sum_{\ell \in [q]} O\left(\frac{q}{\sqrt{\tau\cdot \ell}}\right) = O\left(\frac{q}{\sqrt{\tau}}\cdot \sqrt{q}\right)$, where we use the fact that for every $t\in\mathbb{N}$ it holds that $\sum_{\ell\in[t]}{\frac{1}{\sqrt{\ell}}} = O(\sqrt{t})$.
Thus, choosing $\tau\geq C\cdot q^3$ for a sufficiently large constant $C$, 
we get $\Pr\left[\,\ov{E}_2\right] \leq 1/2$, implying $\Pr[E_2] = \Omega(1)$.

We are left to upper bound the probability of the event $E_1$ (i.e., the event that $i_k$ matches some second-string query point).
Notice that there are at least $\Omega(p)$ unqueried points in the interval $(i_k - p, i_k]$ (because for every queried point in the interval there are at least $\tau\geq 1$ unqueried points). 
Therefore, we have $\Pr[E_1] = O\left(\frac{q}{\sqrt{p}}\right)$. 
Using Eq.~\eqref{eq:overview:dependence}, we obtain $\Pr[P'_k \geq p] = O\left(\frac{q}{\sqrt{p}}\right)$.

We use the foregoing probability bound to bound the \emph{expected} amount of progress made across all queries.
First, for the second part of the progress, we have $\E\left[P'_k\right] = \sum_{p\in[n/2]}\Pr[P'_k \geq p] = \sum_{p\in[n/2]}O\!\left(\frac{q}{\sqrt{p}}\right) = O\left(q\cdot \sqrt{n}\right)$ (where we again use the bound $\sum_{\ell\in[t]}{\frac{1}{\sqrt{\ell}}} = O(\sqrt{t})$ for every $t\in\mathbb{N}$).
Summing across all queries, we get $\E\big[ \sum_{k\in[q]} P'_k \big] = O\left(q^2\cdot \sqrt{n}\right)$.

Next, consider the contribution to the progress coming from the first parts (i.e., $(D_{k-1}, C_k]$).
By the second property of the partition, this contribution across all queries is at most $(\tau+1) \cdot \sum_{k\in[q]} N_k \leq (\tau+1)\cdot q$.
Recall that we have required that $\tau\geq C\cdot q^3$.
Therefore, we set $\tau = C\cdot q^3$, making this contribution $O(q^4)$. Overall, we obtain
$
    \E\left[\sum_{k\in[q]}P_k\right] = O\left(q^4 + q^2 \cdot \sqrt{n}\right).
$

Now, recall that if the tester queries interior matching positions, then necessarily the total progress exceeds $n/6$.
By Markov's inequality, the probability of this event is at most $O\!\left(\frac{q^4}{n} + \frac{q^2}{\sqrt{n}} \right) = O\!\left(\frac{q^2}{\sqrt{n}}\right)$, establishing the desired $\Omega(n^{1/4})$ bound.

\paragraph{Handling the general setting.}

Recall that we assumed a simplified setting in which the tester first queries the second string, and then queries the first half of the first string at monotonically increasing positions.
To handle the general setting, we need to address three things.
First, queries can also be made to the second half of the first string.
To address this, we define the \textsf{depth} of each position to be its distance from the nearest end of the string, and then define $D_k$ as the maximum \emph{depth} of a matched first-string query position among the first $k$ queries.
Second, the first-string queries might not be made in increasing order of depth, and hence previous queries may lie deeper inside the first string than the current query $i_k$.
Therefore, decomposing the progress into two parts requires slightly more care, because we need to account also for queries that lie beyond the progress interval (see Claim~\ref{claim:progress_partition}).
Lastly, the queries to the first and second strings can be interspersed.
Therefore, a query to the second string can also result in progress, if it matches a previously queried first-string position whose depth exceeds $D_{k-1}$.
The analysis of this case is similar to the one we gave for the case where the progress is caused by a first-string query.

\paragraph{On the optimality of the $\Omega(n^{1/4})$ bound.}

As noted also by~\cite{JPR27}, for the current {\small YES} and {\small NO} distributions, the $\Omega(n^{1/4})$ bound is optimal up to logarithmic factors, even for non-adaptive distinguishers.
To see this, let $L \eqdef C\cdot\sqrt{n}$ for a sufficiently large constant $C$, and consider the interval $[n/2-L, n/2+L]$ in both strings.
With high probability, these intervals share $\Omega(L)$ ranks that all lie in $M$.
By the birthday paradox, if we take $O(\sqrt{L}) = O(n^{1/4})$ uniform samples from each interval, we will hit interior matching positions with high probability.
Thus, to distinguish the distributions it suffices to query $O(\log n)$ consecutive positions starting at each sample and check whether there is a match between any pair of the corresponding residual sequences (one from each string).

\subsubsection{Improving the lower bound to $\Omega(n^{2/5})$}\label{subsec:overview:blocks} 

To obtain the $\Omega(n^{2/5})$ lower bound, we modify the above construction of the {\small YES} and {\small NO} distributions as follows. We partition the strings into $m\eqdef n/b$ consecutive blocks of length $b$, where $b$ is a parameter that we will specify later.
Instead of choosing for each position individually whether it is a star-position or not, we choose for every block $h\in [m]$ a uniformly random value $X_h\in\{0,1,\ldots,b\}$, where $X_h$ will be the number of non-star positions in the $h\th$ block.
(As before, to ensure that there are exactly $n/2$ star-positions, we choose the first half of the block values $X_h$ independently, and then set complementary values for the mirror blocks.)
Within each block, the locations of the prescribed number of non-star positions can be chosen arbitrarily.

As before, we define the \textsf{rank} of each non-star position $i\in[n]$ to be the number of non-star positions up to and including position $i$. 
Under the new construction, the rank of position $i$ is the sum of all block values $X_h$ prior to the block of $i$, plus the number of non-star positions in the block of $i$ up to and including position $i$.

Intuitively, the gain of this new construction comes from making the number of star symbols in each block \emph{less concentrated}.
Specifically, whereas in the original construction the number of star symbols in each block was highly concentrated around $b/2$, here this number is spread evenly across all possible values in $\{0,1,\ldots,b\}$.
Consequently, in the middle of the two strings the rank of each position becomes less concentrated, making it harder to find interior matching positions.

More concretely, let us first reconsider the case of non-adaptive testers. 
In this case, we set $b \eqdef n/6$, so that the first $n/6$ positions form a single block.
Since the number of star-symbols in this block is uniform over $\{0,1,\ldots,b\}$, it holds that for any $i \in \big(\frac{n}{6}, \frac{5n}{6}\big]$, and any fixed $r$, the probability that $R(i) = r$ is at most $\frac{1}{b+1} = O\left(\frac{1}{n} \right)$. 
Thus, the probability of querying interior matching positions is at most $O\left(\frac{q^2}{n} \right)$, establishing the $\Omega(\sqrt{n})$ lower bound for non-adaptive testers. 

Returning to the adaptive case, note that each previous query can now affect the entire block to which it is made, because the star-positions inside the same block are now dependent. 
Therefore, we will use shorter blocks, 
and use the fact that for $t$ independent random variables $X_1,...\,,X_t$ that are uniform over $\{0,1,...\,,b\}$, it holds that $\max_{j\in\mathbb{N}}\left\{\Pr\left[ \sum_{i\in[t]}{X_i} = j \right] \right\} = O\left(\frac{1}{b\cdot\sqrt{t}}\right)$
(as opposed to $O\left(\frac{1}{\sqrt{t}}\right)$ that we had in the case of $t$ Bernoulli random variables).
As in the previous section, we analyze the tester's progress. 
However, we now measure the progress in \emph{blocks}; that is, the event $P_k \geq p$ now stands for making at least $p$ blocks of progress. 
We use a decomposition of the progress similar to the one we used in the previous section, with $\tau=1$.
The decomposition ensures that in the second part of the progress there are at least $\ell$ unqueried blocks prior to the $\ell\th$ query, 
whereas the first parts of the progress consist of up to $2q$ blocks in total.

We apply the stronger anti-concentration inequality to bound the contribution of the second part of the progress $P'_k$ to the total expected progress.
Specifically, we now have $\Pr\left[\,\ov{E}_2\right] = \sum_{\ell \in [q]}O\left(\frac{q}{b \cdot \sqrt{\ell}}\right) = O\left(\frac{q}{b}\cdot \sqrt{q}\right)$,
which is at most $1/2$ provided that $b\geq C\cdot q^{3/2}$ for a sufficiently large constant $C$.
Furthermore, we have $\Pr[E_1]=O\left(\frac{q}{b\cdot\sqrt{p}}\right)$.
This gives us $\E\left[\sum_{k\in[q]} P'_k\right] = O\left(q^2 \cdot \frac{\sqrt{n}}{b^{3/2}}\right)$.

On the other hand, we need to add the contribution of the first part of the progress, which can now be up to $2q$ blocks in total.
Thus, we get $ \E\left[\sum_{k\in[q]}P_k\right] = O\left(q + q^2 \cdot \frac{\sqrt{n}}{b^{3/2}}\right)$, which in units of string positions is $O\left(q\cdot b + q^2 \cdot \sqrt{\frac{n}{b}}\right)$.
We set $b \eqdef C\cdot q^{3/2}$, and by Markov's inequality, we get that the probability of reaching beyond position $n/6$ is at most $O\left(\frac{q^{5/2}}{n} + \frac{q^{5/4}}{\sqrt{n}}\right) = O\left(\frac{q^{5/4}}{\sqrt{n}}\right)$, giving us the desired $\Omega(n^{2/5})$ bound.

%% file: E_ub_overview.tex
\subsection{A non-adaptive tester}\label{sec:overview:up}

Our non-adaptive tester for \texttt{ResStringEq} is an adaptation of the tester of~\cite{FMS18}. 
Recall that~\cite{FMS18} showed a tester for the \texttt{Dyck} languages with query complexity $O(n^{2/5+\delta})$, where $\delta>0$ is an arbitrarily small constant.
As a sub-protocol used in their tester, 
they presented a procedure named ``Substring $\epsilon$-matching'' (see~\cite[Algorithm~3]{FMS18}). 
This procedure can essentially be viewed as a tester for \texttt{ResStringEq}, and it has our target query complexity of $O(n^{1/2+\delta})$.
However, it is \emph{adaptive}. 
We show that it is possible to modify their procedure to be non-adaptive, with essentially no additional overhead.

We next give a high-level overview of the tester. The overview will be structured as follows:
As a warm-up, we first present a basic tester, which is \emph{adaptive} and uses $\widetilde{O}(n^{2/3})$ queries.
We then explain how we can modify this basic tester to be non-adaptive, 
at the cost of increasing the query complexity to $\widetilde{O}(n^{4/5})$.
Finally, we show how we can improve the query complexity by using recursion, and obtain the desired $O(n^{1/2+\delta})$ bound.

\paragraph{The basic tester.}

For a string $x\in\{0,1,*\}^n$ and a position $i\in[n]$, we define the \textsf{rank} of $i$ in $x$, denoted $r_x(i)$, as the number of non-star symbols in $x$ up to and including position $i$; that is, $r_x(i)\eqdef |\{j\leq i : x_j\neq *\}|$.
The basic thing we would like to do is check that at (random) non-star positions of equal rank the two strings match.
However, the issue is that we do not know the ranks of positions we query.\footnote{
    Note that if each query told us, in addition to the value of the string at the queried position, also the rank of that position, then we could easily get the $O(\sqrt{n})$ bound: 
    We would query $O(\sqrt{n})$ random positions in each of the two input strings. 
    By the birthday paradox, w.h.p.\ the queried positions will contain two positions of matching ranks, and we can compare the corresponding values for consistency.
}

The first important observation is that by sampling random positions in the strings, we can obtain \emph{estimates} of the various ranks.
Specifically, for $\Delta \gg \sqrt{n}$, if we query each string at $O\big((\frac{n}{\Delta})^2\cdot \log(n)\big)$ random positions, we can obtain an estimate of the rank of each position $i\in[n]$ up to an additive deviation of $\Delta$, where the failure probability of each estimate individually is at most $o\!\left(\frac{1}{n}\right)$ (see Claim~\ref{claim:successful_estimates}). 
By a union bound, all the estimates are successful simultaneously with high probability.

This gives rise to the following basic tester: For a parameter $L\in [n]$ to be determined later,\footnote{
    We will set $L=n^{2/3}$ in the basic adaptive tester. 
}
we first obtain estimates of all ranks up to an additive deviation of $\Delta \eqdef 0.1 \cdot \epsilon \cdot L$.
Based on these estimates, we partition the strings into ``\textsf{rank-segments}'' that each contain approximately $L$ non-star symbols.
Specifically, the $k\th$ rank-segment of each string consists of the positions whose estimated ranks are between $(k-1)\cdot L$ and $k\cdot L$.

We now know that, for every $k$, 
the $k\th$ rank-segment of the first input string should ``approximately match'' with the $k\th$ rank-segment of the second input string.
Specifically, assuming the estimates are successful, each such pair of corresponding rank-segments contains the same range of ranks, up to an offset of at most $2\Delta$ at each endpoint.
Thus, in a {\small YES}-input, for each pair of corresponding rank-segments, the residual strings in the segments \textsf{match up to boundary slack $2\Delta$}, meaning that the two residual strings agree after deleting at most $2\Delta$ symbols from the beginning of one of them, and possibly leaving at most $2\Delta$ extra symbols at the end of one of them.
On the other hand, if the input is $\epsilon$-far from \texttt{ResStringEq}, then at least an $\Omega(\epsilon)$ fraction of the corresponding rank-segments do not satisfy this boundary-slack matching condition (see Claim~\ref{claim:many_bad_blocks}).

Therefore, after partitioning the strings into rank-segments, we uniformly select $O(1/\epsilon)$ pairs of corresponding rank-segments, query the entire segments, and check that their residual strings match up to boundary slack $2\Delta$.
Importantly, we can avoid querying any ``long'' rank-segments, where we call a rank-segment \textsf{long} if its length is greater than $L/(\alpha\cdot \epsilon)$, for a sufficiently small constant $\alpha > 0$.
Roughly speaking, this is because long rank-segments are necessarily \emph{sparse} 
(i.e., they each contain at most an $O(\alpha\cdot \epsilon)$ fraction of non-star symbols),
and hence they cannot contribute many mismatches between the strings.

Recall that we set $\Delta = \Theta(\epsilon\cdot L)$. 
Thus, for constant $\epsilon>0$, we query $\widetilde{O}\big(\frac{n^2}{L^2}\big)$ positions for obtaining the estimates, and $O(L)$ positions when querying the (non-long) rank-segments.
To balance these two terms, we set $L = n^{2/3}$ and obtain a tester of query complexity $\widetilde{O}(n^{2/3})$.\footnote{
Note that this reestablishes the query complexity bound obtained by combining the reduction of~\cite{FMS18} with the tester for the \texttt{Dyck} languages given in~\cite{PRR03}.
}
(In Section~\ref{subsec:upper_bound:proof:warm_up1} we provide a formal construction and analysis of this basic tester, as a warm-up towards our actual tester.)

\paragraph{Transforming the basic tester to be non-adaptive.}

Note that the basic tester is inherently adaptive: it must first obtain the rank estimates before it can query the corresponding rank-segments.
Hence, rather than directly querying the rank-segments (whose locations are unknown to us), we proceed as follows.
We consider a covering of the strings by \emph{fixed} \textsf{blocks} of equal length $2L/(\alpha\cdot \epsilon) = O(L)$ that are positioned at shifts of size $L/(\alpha\cdot \epsilon)$ from one another. 
Note that this choice of parameters ensures that any non-long rank-segment is necessarily contained in some block of the covering (recall that long rank-segments can be discarded).
In addition, note that there are $O(n/L)$ blocks. 

Rather than sampling corresponding rank-segments, we uniformly sample (in advance of estimating the ranks) $O\big(\sqrt{\frac{n}{L}}\big)$ blocks from the covering (and query the selected blocks in their entirety).
After obtaining the estimates, we search for pairs of corresponding rank-segments such that each segment is contained in one of the sampled blocks. For each such pair, we check that the residual strings of the two rank-segments match up to boundary slack $2\Delta$ (as described above).
Since we sample $O\big(\sqrt{\frac{n}{L}}\big)$ blocks, a birthday-paradox-type argument shows that, in the {\small NO}-case, with high probability, at least one of the pairs that do not match up to boundary slack $2\Delta$ will have both rank-segments contained in sampled blocks, and will be detected (see Claim~\ref{claim:birthday}).\footnote{
    We note that the ``Substring $\epsilon$-matching'' procedure of~\cite{FMS18} also utilizes the birthday paradox, although for a different reason that will become clearer when we present the final recursive tester.
    The procedure of~\cite{FMS18} directly samples a square-root number of (their analog of the) \emph{rank-segments}, whereas here we sample the fixed blocks.
} 

The total query complexity of the foregoing non-adaptive tester is $\widetilde{O}\big(\frac{n^2}{L^2} + \sqrt{\frac{n}{L}}\cdot L\big)$, where the first term accounts for the query complexity of obtaining the estimates (recall that $\Delta = \Theta(\epsilon\cdot L)$), and the second term accounts for the query complexity of querying $O(\sqrt{\frac{n}{L}})$ blocks, each of size $O(L)$.
Balancing the two terms, we set $L = \widetilde{O}(n^{3/5})$, which gives us query complexity $\widetilde{O}(n^{4/5})$. 
(In Section~\ref{subsec:upper_bound:proof:warm_up2} we present a formal construction and analysis of this basic non-adaptive tester, as a warm-up towards our actual tester.)

\paragraph{Improving the tester by using recursion.}
The basic tester reads the sampled blocks in their entirety, in order to compare the corresponding rank-segments and check if their residual strings match up to boundary slack.
Rather than reading the entire blocks, we would like to continue the comparison of the corresponding rank-segments by \emph{recursion}.
However, in contrast to the original input strings, whose residual strings need to match exactly, the residual strings of the corresponding rank-segments only need to match up to boundary slack $2\Delta$.
This means that there can be an offset of at most $2\Delta$ initial symbols that should be skipped in one of the two residual strings before starting the comparison.
The issue is that we do not know what the offset is and therefore do not know which positions should be compared.

The solution is to try \emph{all} offsets, and accept if and only if for one of them all recursive calls accepted. 
We take a union bound over the failure probability of each offset trial; note that we have only $O(n)$ possible offsets, whereas we can reduce the failure probability exponentially with repeated invocations. 
Crucially, we do not need to sample a separate collection of blocks for each offset. \emph{We sample $\widetilde{O}\big(\sqrt{\frac{n}{L}}\big)$ blocks once, independently of the offset, and use the same sampled blocks across all offset trials.}\footnote{
    In other words, pre-sampling a square-root number of the fixed blocks (rather than sampling corresponding rank-segments, as in the basic adaptive tester) now serves \emph{two purposes}: it allows the tester to be non-adaptive, and allows the same samples to be reused across the different offset trials.
}

More concretely, the final tester proceeds as follows. 
It is composed of two sub-procedures that are each recursive: a \textsf{Query} procedure, and a \textsf{Decision} procedure.
The Query procedure is run separately on each of the input strings and makes all queries. The Decision procedure receives the outputs of the Query procedure and decides whether to accept or reject the input pair.

On query access to a string $s \in \{0,1,*\}^n$, the Query procedure first queries $s$ at uniformly random positions, where these queries will later be used for obtaining the rank estimates.
It then samples $\widetilde{O}\big(\sqrt{\frac{n}{L}}\big)$ blocks of $s$, and recursively invokes the Query procedure on each sampled block.
In the final recursion round, it queries the entire input string.

We next give only a very high-level description of the Decision procedure; the full details appear in Section~\ref{subsec:upper_bound:proof:actual}.
When comparing strings $s$ and $s'$, the Decision procedure first estimates the ranks of all positions in the strings (using the estimation queries made by the Query procedure).
Then, \emph{at the first recursion level}, the procedure searches for pairs of corresponding rank-segments that are both contained in blocks that were sampled by the Query procedure.
For each such pair, it checks if the residual strings of the two rank-segments match up to boundary slack $2\Delta$, by recursively invoking the Decision procedure on these rank-segments.
The recursive call receives the outputs of the Query procedure on the two sampled blocks containing these rank-segments, together with the position of each rank-segment within its respective block.

In the next recursion levels, after estimating the ranks, the Decision procedure considers both possible ways of ordering the two input strings, and for each ordering it considers all relevant offsets $h$.
For each offset, it partitions the two strings into rank-segments as before, except that when partitioning the first string (w.r.t.\ the current ordering of the two strings), it discards the first $h$ (estimated) ranks. 
It then searches, as before, for pairs of corresponding rank-segments that are contained in sampled blocks, and recursively checks whether their residual strings match up to boundary slack.
The procedure accepts if and only if there was some inspected offset $h$ for which all recursive calls accepted.\footnote{
Actually, for convenience, in the actual implementation we do not treat the first recursive round separately. 
Instead, already in the first round, we use the more general procedure (that tries all offsets) described for the second round onward.
Hence, our tester accepts inputs even if their residual strings only match up to boundary slack $0.1\cdot \epsilon \cdot n$.
}

Denoting the query complexity of the tester after $r$ recursion rounds by $Q_r(n)$, we have $Q_r(n) = \widetilde{O}\left(\frac{n^2}{L^2} + \sqrt{\frac{n}{L}} \cdot Q_{r-1}\left(O(L)\right)\right)$, where $Q_0(n) = n$.
We write $Q_r(n) = \widetilde{O}(n^{\gamma_r})$, where $\gamma_r > 0$, and we choose $L$ so as to balance the terms $\frac{n^2}{L^2}$ and $\sqrt{\frac{n}{L}}\cdot L^{\gamma_{r-1}}$, which gives $L \eqdef n^{3/(3+2\gamma_{r-1})}$.
With this choice of $L$, we get $\gamma_r = 2-2\cdot 3/(3+2\gamma_{r-1})$. 
Together with $\gamma_0 = 1$, the solution to the recursion yields $\gamma_r = \frac{1}{2-(3/4)^r}$.
Thus, the final query complexity after $r$ recursion rounds is $\widetilde{O}\big(n^{\frac{1}{2-(3/4)^r}}\big)$.  
(Indeed, the basic non-adaptive $\widetilde{O}(n^{4/5})$-query tester discussed above is a special case corresponding to $r=1$.)
Setting $r=O(\log(1/\delta))$, we get the desired $O(n^{1/2+\delta})$ bound.

\paragraph{A remark on obtaining an
$O(n^{2/5+\delta})$-query \emph{adaptive} tester.}
Composing the basic adaptive tester with the foregoing recursive procedure directly yields an $O(n^{2/5+\delta})$-query \emph{adaptive} tester for \texttt{ResStringEq}, analogous to the tester of~\cite{FMS18} for the \texttt{Dyck} languages. 
Recall that the basic adaptive tester selects $O(1/\epsilon)$ pairs of corresponding rank-segments, and if both segments are short, it queries the segments in their entirety, and checks whether their residual strings match up to boundary slack $2\Delta$. 
Rather than reading the selected rank-segments entirely, we can instead perform each of the boundary-slack matching checks by invoking the recursive procedure described above. 
For constant $\epsilon>0$, this will require us to query only $O(L^{1/2+\delta})$ points in each of the sampled pairs of corresponding rank-segments.
To balance this term with the $\widetilde{O}(n^2/L^2)$ queries needed for the rank estimates,  we set $L=n^{4/5}$, giving us the claimed $O(n^{2/5+\delta})$ query complexity.

%% file: F_preliminaries.tex
\section{Preliminaries}\label{sec:preliminaries}

\paragraph{Property testing.}

For an alphabet $\Sigma$, a \textsf{property} is a collection of sets $\Pi = \bigcup_{n \in \mathbb{N}}\Pi_n$ such that $\Pi_n$ is a set of strings in $\Sigma^n$.
The relative hamming distance between two strings $x,x'\in\Sigma^n$ is the fraction of positions on which they differ. 
We say $x$ is $\epsilon$-\textsf{far} from $x'$ if the relative hamming distance between $x$ and $x'$ is greater than $\epsilon$, and otherwise we say they are $\epsilon$-\textsf{close}.
A string $x\in\Sigma^n$ is $\epsilon$-far from a property $\Pi = \bigcup_{n \in \mathbb{N}}\Pi_n$ if it is $\epsilon$-far from any $x'\in \Pi_n$; otherwise, it is $\epsilon$-close to $\Pi$.
A \textsf{tester} for a property $\Pi$ is a probabilistic algorithm that, on input parameters $n\in \mathbb{N}$, $\epsilon>0$ and oracle access to $x\in\Sigma^n$, outputs $1$ with probability at least $2/3$ if $x$ is in $\Pi$ (``completeness''), and outputs $0$ with probability at least $2/3$ if $x$ is $\epsilon$-far from $\Pi$ (``soundness'').
The \textsf{query complexity} of the tester is $q : \mathbb{N} \times [0, 1] \rightarrow \mathbb{N}$ if, on input $n$, $\epsilon$ and oracle access to any $x\in\Sigma^n$, the tester makes at most $q(n, \epsilon)$ queries to $x$.

\paragraph{Statistical distance and indistinguishability.}

The \textsf{statistical distance} between random variables $X$ and $X'$ is defined as
$$
    \Delta(X, X') \eqdef \frac{1}{2} \cdot \sum_{x}{|\Pr[X=x]-\Pr[X'=x]|} = \max_{f:\{0,1\}^*\to\{0,1\}}{\{\Pr[f(X)=1] - \Pr[f(X')=1]\}}
$$
We say that an algorithm $A$ \textsf{distinguishes} between two sequences of random variables $\{X_n\}_{n\in \mathbb{N}}$ and $\{X'_n\}_{n\in \mathbb{N}}$ if
$$
    |\Pr[A(X_n)=1]-\Pr[A(X'_n)=1]| = \Omega(1),
$$
which implies that $\Delta(X_n,X'_n) = \Omega(1)$.
We say that two sequences of random variables $\{X_n\}_{n\in \mathbb{N}}$ and $\{X'_n\}_{n\in \mathbb{N}}$ are \textsf{indistinguishable} if $\Delta(X_n,X'_n) \neq \Omega(1)$.

%% file: Ga_lb_proof_nonadaptive.tex
\section{Lower bound}\label{sec:lb_proof}

In this section we prove an $\Omega(n^{2/5})$ lower bound for testing \texttt{ResStringEq} (Theorem~\ref{thm:main}), as well as a stronger $\Omega(\sqrt{n})$ lower bound for testers that make non-adaptive queries (Theorem~\ref{thm:non_adaptive_lower_bound}). 
As described in Section~\ref{sec:overview:lb}, the proofs proceed by an indistinguishability argument.
We first present the two distributions that are used for both proofs.
We then prove the non-adaptive lower bound in Section~\ref{subsec:non_adaptive_proof}, before turning to the general adaptive case in Section~\ref{subsec:adaptive_proof}.

\subsection{The two distributions}\label{subsec:the_two_distributions}

We closely follow the construction described in the overview in Section~\ref{sec:overview:lb}.
We construct two random pairs of strings in $\{0,1,*\}^n$: a {\small YES} pair, denoted $(S,S')$, and a {\small NO} pair, denoted $\big(\widetilde{S},\widetilde{S}'\big)$.
We first choose the locations of the star symbols identically in the two distributions, as follows.
We begin with the first string in the pair (i.e., $S$ in the {\small YES} pair and $\widetilde{S}$ in the {\small NO} pair).
We partition the string into $m\eqdef n/b$ consecutive blocks of length $b$, for a parameter $b\in \mathbb{N}$ to be specified later.
We draw $m$ values $X = (X_1,\ldots,X_m)$, such that each $X_h$ is uniform over $\{0,1,\ldots,b\}$, 
where the value $X_h$ indicates that block $h$ in the constructed string will contain $X_h$ non-star symbols. 
To ensure that exactly $n/2$ positions contain star symbols, we draw the first $m/2$ values $X_1,\ldots,X_{m/2}$ independently, and then set the other $m/2$ mirror blocks to the complement value; that is, we set $X_{m/2+h} = b-X_{m/2+1-h}$ for every $h\in [m/2]$. 
Within each block, the locations of the prescribed number of non-star symbols can be chosen arbitrarily.
Independently, we repeat this process in the second string (i.e., $S'$ and $\widetilde{S}'$), but using fresh values $X' = (X'_1,\ldots,X'_m)$.\footnote{
    We note that the proof will not use the fact that $X'$ is drawn in the same way as $X$, only that $X'$ and $X$ are independent.
    We could have chosen arbitrary $n/2$ positions for the star symbols in the second string, as long as they are independent of $X$. 
}  
We stress that we use the same values $X$ and $X'$ to construct both the {\small YES} and {\small NO} pairs.

For every $i\in[n]$, we define the rank of $i$ in the first string, denoted $R(i)$, as follows. If $i$ is not a star position (according to $X$), we define $R(i)$ to be the number of non-star symbols up to and including $i$; that is, $R(i) \eqdef \sum_{h \leq \lfloor (i-1)/b \rfloor} X_h + c$, where $c$ is the number of non-star symbols in the block of $i$ up to and including position $i$.
Otherwise (i.e., $i$ is a star position), we define $R(i) = \bot$.
Similarly, we let $R'(i)$ denote the rank of position $i$ in the second string, defined analogously with respect to $X'$.

Before continuing with the construction, we note that $R(i)$ can be expressed using the sum over the block variables between $i$ and the nearest end of the string, and that these variables are independent: If $i\leq n/2$, then $R(i) = \sum_{h \leq \lfloor (i - 1)/b \rfloor}{X_h} + c$, where $X_1, \ldots, X_{\lfloor (i - 1)/b \rfloor}$ are independent. If $i > n/2$, then we use the fact that $\sum_{h\in[m]}{X_h} = n/2$ to write $R(i) = \sum_{h \leq \lfloor (i-1)/b \rfloor}{X_h} + c = n/2 - \sum_{h > \lfloor (i-1)/b \rfloor}{X_h} + c$, where $X_{\lfloor (i - 1)/b \rfloor + 1},\ldots,X_m$ are independent.
This fact will be useful later on. 

Returning to the construction, we now fill the remaining positions in the strings. 
In the {\small YES} pair $(S,S')$, for every $r\in[n/2]$, we draw a uniform bit $V_r\in\{0,1\}$ and place it in the two positions with rank $r$; that is, we set $S_i=S'_{i'}=V_r$, where $i$ and $i'$ are the positions such that $R(i) = R'(i') = r$.
In the {\small NO} pair $\big(\widetilde{S},\widetilde{S}'\big)$, we do the same, except that on the middle third of the ranks we use independent bits in the two strings. 
That is, let $M \eqdef \left(\frac{n}{6}, \frac{n}{3}\right]$ be the middle third of the ranks.
Then, for ranks $r \in [n/2] \setminus M$ we use a shared uniform bit $V_r$ exactly as before (i.e., we set $\widetilde{S}_i=\widetilde{S}'_{i'}=V_r$, where $R(i) = R'(i')=r$), while for ranks $r \in M $ we draw two independent uniform bits $V_r, V'_r \in \{0,1\}$, and use $V_r$ in the first string and $V'_r$ in the second (i.e., we set $\widetilde{S}_i=V_r$ and $\widetilde{S}'_{i'}=V'_r$, where $R(i) = R'(i')=r$).

\begin{claim}[implicit in the proof of{~\cite[Lem.~5.4]{FMS18}}; see also{~\cite[Lem.~14]{PRR03}}]
    For any sufficiently small $\epsilon>0$, the probability that $(\widetilde{S},\widetilde{S}')$ is $\epsilon$-far from \textup{\texttt{ResStringEq}} is at least $1-o(1)$.
\end{claim}

Thus, to establish Theorem~\ref{thm:main} (resp., Theorem~\ref{thm:non_adaptive_lower_bound}), we show that for any tester that makes $o(n^{2/5})$ queries (resp., $o(\sqrt{n})$ non-adaptive queries), the view of the tester on input $(S,S')$ is indistinguishable from its view on input $(\widetilde{S},\widetilde{S}')$. 

\subsection{The non-adaptive case}\label{subsec:non_adaptive_proof}

We begin with the non-adaptive case, proving the $\Omega(\sqrt{n})$ lower bound stated in Theorem~\ref{thm:non_adaptive_lower_bound}.
Consider an arbitrary non-adaptive tester for \texttt{ResStringEq} with query complexity $q$. 
Fix the randomness of the tester, and let $Q$ and $Q'$ be the sets of queries the tester makes to the first and second strings in the input pair, respectively.
Let $A \eqdef (S_Q,S'_{Q'})$ be the view of the tester on input $(S,S')$; that is, $A$ is the restriction of $(S,S')$ to the respective positions in $Q$ and $Q'$.
Similarly, let $\tilde{A} \eqdef \big( \widetilde{S}_Q, \widetilde{S}'_{Q'} \big)$ be the view of the tester on input $(\widetilde{S},\widetilde{S}')$.

To bound the statistical distance between $A$ and $\tilde{A}$ we will use the following fact:

\begin{fact}\label{fact:stat_dist}\!\!\!\footnote{See Appendix~\ref{apdx:proof_of_stat_dist_fact}, which provides a proof of a more general statement.}
    For any two random variables $X$ and $X'$, and any event $\mathcal{E}$ over their joint probability space, if $\left(X \,|\, \ovbar{\mathcal{E}}\,\right)$ and $\left(X' \,|\, \ovbar{\mathcal{E}}\,\right)$ are identically distributed,\footnote{
    The notation $\ovbar{\mathcal{E}}$ denotes the complement of the event $\mathcal{E}$.} then $\Delta(X, X') \leq \Pr[\mathcal{E}]$.
\end{fact}

We say that positions $(i, i') \in [n]\times[n]$ are \textsf{matching} positions if $R(i) = R'(i') \neq \bot$.
We say that $(i, i')$ are \textsf{interior} matching positions if $R(i) = R'(i') \in M$.
Let $\mathcal{E}$ be the event that there exist $i\in Q$ and $i'\in Q'$ such that $(i, i')$ are interior matching positions.

\begin{claim}\label{claim:identically_distributed}
    It holds that $(A \,|\, \ovbar{\mathcal{E}}\,)$ and $(\tilde{A} \,|\, \ovbar{\mathcal{E}}\,)$ are identically distributed. 
\end{claim}

\begin{proof}
    Fix the star positions arbitrarily (by fixing $X$ and $X'$) such that $\ovbar{\mathcal{E}}$ holds (note that the event $\mathcal{E}$ is determined by the star positions).
    We show that under this fixing, the conditional distributions of $A$ and $\tilde{A}$ are identical.
    First, note that we only need to consider the answers in $A$ and $\tilde{A}$ to the queries (among $Q$ and $Q'$) that are made to non-star positions, because the remaining answers are fixed star symbols.
    Clearly, the answers to queries whose ranks are outside $M$ are identically distributed in $A$ and $\tilde{A}$.
    Now, since $\ovbar{\mathcal{E}}$ holds, any query whose rank is in $M$ does not match any queried position in the other string. 
    Thus, in both $A$ and $\tilde{A}$, the answer to any query whose rank is in $M$ is a uniform bit independent of all other query answers.
    The claim follows.
\end{proof}

Thus, to bound $\Delta(A, \tilde{A})$, it suffices to bound $\Pr\left[\mathcal{E}\right]$.

\begin{claim}\label{claim:coll_prob}
    If we set $b = n/6$, it holds that $\Pr\left[\mathcal{E}\right] = O\left( \frac{q^2}{n} \right)$.
\end{claim}

It follows that unless $q = \Omega(\sqrt{n})$, the views $A$ and $\tilde{A}$ are indistinguishable, establishing the lower bound.

\begin{proof}[\textup{\textbf{Proof of Claim~\ref{claim:coll_prob}}}:]
    Our aim is to bound $\Pr\left[\mathcal{E}\right] = \Pr\left[\exists\, (i,i')\in Q\times Q' \textit{ s.t. } R(i) = R'(i') \in  M \right]$. 
    Note that if $R(i) \in M = \left(\frac{n}{6}, \frac{n}{3} \right]$ then it must hold that $i \in \left( \frac{n}{6}, \frac{5n}{6} \right]$ (because there must be at least $n/6$ positions before $i$ and at least $n/6$ positions after $i$ with the remaining ranks). 
    Now, for every $i \in \left( \frac{n}{6}, \frac{5n}{6} \right]$, if $R(i) \neq \bot$, then $R(i)$ equals $X_1$ plus an independent term.
    Since $X_1$ is uniform over $\{0,1,\ldots, b\}$, it follows that for every $r\in[n/2]$, it holds that $\Pr[R(i) = r] \leq \frac{1}{b+1} = O\left( \frac{1}{n} \right)$. 
    Thus, for each pair $(i,i')\in Q\times Q'$, the probability that $R(i) = R'(i') \in M$ is at most $O\left( \frac{1}{n} \right)$ (because $R(i)$ and $R'(i')$ are independent).
    Taking a union bound over all pairs $(i,i')\in Q\times Q'$, we get that $\Pr\left[ \mathcal{E} \right] = O\left(\frac{q^2}{n}\right)$, as desired.
\end{proof} 

%% file: Gb_lb_proof_adaptive.tex
\subsection{The adaptive case}\label{subsec:adaptive_proof}

We turn to the general adaptive case, proving the $\Omega(n^{2/5})$ lower bound stated in Theorem~\ref{thm:main}.
Consider an arbitrary (possibly adaptive) tester with query complexity $q$. 
As in the previous section, fix the randomness of the tester.  

\paragraph{Notation.}
As a convention, we write $x_{[k]}$ as a shorthand for $(x_1, \ldots, x_{k})$.
We next introduce notation for the adaptive queries.
For each query $k\in[q]$, let $Q_k(a_{[k-1]})$ be the position queried by the tester at the $k\th$ query, where $a_{[k-1]} = (a_1,\ldots, a_{k-1}) \in \{0,1,*\}^{k-1}$ denotes the first $k-1$ query answers. 
Let $Q_{<k}(a_{[k-2]})$ and $Q'_{<k}(a_{[k-2]})$ be the sets of positions queried in the first and second strings, respectively, prior to the $k\th$ query, and let $Q_{\leq k}(a_{[k-1]})$ and $Q'_{\leq k}(a_{[k-1]})$ be the sets of positions queried up to and {including} the $k\th$ query.
The \textsf{view} of the tester is the full sequence of answers $a = (a_1,\ldots, a_q)\in \{0,1,*\}^q$.
Let $Q(a)$ and $Q'(a)$ denote the full sequences of queries to the first and second strings, respectively, given the full sequence of answers $a\in \{0,1,*\}^q$.

 \medskip
Let $A = (A_1, \ldots, A_q)$ and $\tilde{A} = (\tilde{A}_1, \ldots, \tilde{A}_q)$ be the views of the tester when given $\big(S,S'\big)$ and $(\widetilde{S},\widetilde{S}')$ as inputs, respectively. 
For each $a \in \{0,1,*\}^{q}$, let $\mathcal{E}_a$ be the event that there exist $i\in Q(a)$ and $i'\in Q'(a)$ such that $(i,i')$ are interior matching positions (i.e., $R(i) = R'(i') \in M$).

The subsequent Claim~\ref{claim:identically_distributed_adaptive} together with Fact~\ref{fact:stat_dist_adaptive}
shows that $\Pr\left[\mathcal{E}_{A}\right] = \Pr\left[\mathcal{E}_{\tilde{A}}\right]$ and $\Delta(A,\tilde{A}) \leq \Pr\left[\mathcal{E}_{A}\right]$.
First, similarly to Claim~\ref{claim:identically_distributed}, we have:

\begin{claim}\label{claim:identically_distributed_adaptive}
    For each $a \in \{0,1,*\}^{q}$, it holds that $\Pr\big[ A=a, \ovbar{\mathcal{E}}_{A} \big] = \Pr\big[ \tilde{A} = a, \ovbar{\mathcal{E}}_{\tilde{A}} \big]$ \textup{(}where $ \ovbar{\mathcal{E}}_{A}$ and $\ovbar{\mathcal{E}}_{\tilde{A}}$ denote the complement events of $\mathcal{E}_{A}$ and $\mathcal{E}_{\tilde{A}}$, respectively\textup{)}.
\end{claim}

The proof is identical to that of Claim~\ref{claim:identically_distributed}, except that $\ovbar{\mathcal{E}}$ is replaced by $\ovbar{\mathcal{E}}_a$, and the query sets $Q,Q'$ are replaced by $Q(a),Q'(a)$.
Next, we use the following fact, which generalizes Fact~\ref{fact:stat_dist}:

\begin{fact}[see Appendix~\ref{apdx:proof_of_stat_dist_fact}]\label{fact:stat_dist_adaptive}
    Let $X$ and $X'$ be two discrete random variables supported over a domain $D$, and let $\mathcal{E}$ and $\mathcal{E}'$ be two events over the probability spaces of $X$ and $X'$, respectively.
    If for every $x\in D$ it holds that $\Pr[X=x, \ovbar{\mathcal{E}}] = \Pr[X' = x, \ovbar{\mathcal{E}}']$, 
    then $\Pr\left[\mathcal{E}\right] = \Pr\left[\mathcal{E}'\right]$ and $\Delta(X,X') \leq \Pr\left[\mathcal{E}\right]$.
\end{fact}

Thus, in order to upper bound $\Delta(A, \tilde{A})$, it suffices to upper bound $\Pr[\mathcal{E}_{A}]$. 

\begin{lemma}\label{claim:coll_prob_adaptive}
    If we set $b = C \cdot q^{3/2}$ for a sufficiently large constant $C$, then
    $\Pr[\mathcal{E}_{A}] =
    O\!\left(\frac{q^{5/4}}{\sqrt{n}}\right)
    $.
\end{lemma}

It follows that unless $q=\Omega(n^{2/5})$, the views $A$ and $\tilde{A}$ are indistinguishable, which establishes Theorem~\ref{thm:main}.  
It remains to prove Lemma~\ref{claim:coll_prob_adaptive}.

\subsubsection{Proof of Lemma~\ref{claim:coll_prob_adaptive}}\label{subsec:progress_proof}

We start by introducing a few definitions.
For a position $i$ in block $h\in[m]$, define its \textsf{depth} as $d(i)\eqdef\min\{h,m+1-h\}$, namely, 
the number of blocks from the nearest end of the string up to and including the block containing $i$.
For each query $k\in[q]$, let $D_k$ be the maximum depth of any query position in $Q_{\leq k}(A_{[k-1]})$ that matches some query position in $Q'_{\leq k}(A_{[k-1]})$. 
If no such position exists, then define $D_k \eqdef 0$. Furthermore, define $D_0 \eqdef 0$.
Define the \textup{\textsf{progress}} of the $k\th$ query as $P_k \eqdef D_k - D_{k-1}$.

To prove Lemma~\ref{claim:coll_prob_adaptive}, we will show the following:

\begin{lemma}\label{claim:total_progress_expectation}
    Suppose that $b \geq C \cdot q^{3/2}$, for a sufficiently large constant $C$. Then $\E\big[\sum_{k\in[q]} P_k \big] = O\left( q + q^2 \cdot \frac{\sqrt{n}}{b^{3/2}} \right)$.
\end{lemma}

From the above lemma we derive Lemma~\ref{claim:coll_prob_adaptive} as follows. 
Note that if event $\mathcal{E}_{A}$ holds (i.e., the tester queries interior matching positions when querying $(S, S')$), then the total progress must satisfy $\sum_{k\in[q]}P_k = D_q > \frac{n}{6\cdot b}$ (recall that the progress is measured in blocks). 
For $b \geq C \cdot q^{3/2}$ the above Lemma~\ref{claim:total_progress_expectation} gives $\E\!\big[\sum_{k\in[q]}P_k\big] = O\left( q + q^2 \cdot \frac{\sqrt{n}}{b^{3/2}} \right)$. 
By Markov's inequality, we get $\Pr[\mathcal{E}_{A}] \leq \Pr\!\big[\sum_{k\in[q]}P_k > \frac{n}{6\cdot b}\big] = O\Big(\frac{q\cdot b}{n} + \frac{q^{2}}{\sqrt{b\cdot n}}\Big)$. 
Setting $b = C \cdot q^{3/2}$, we get $\Pr[\mathcal{E}_{A}] = O\Big(\frac{q^{5/2}}{n}+\frac{q^{5/4}}{\sqrt{n}}\Big) = O\!\left(\frac{q^{5/4}}{\sqrt{n}}\right)$, establishing Lemma~\ref{claim:coll_prob_adaptive}.

Towards proving Lemma~\ref{claim:total_progress_expectation}, we establish a few preliminaries. 
We call a depth $d\in[m/2]$ \textsf{unqueried} at query $k$ if there are no first-string query positions in $Q_{\leq k}(A_{[k-1]})$ with depth $d$.
We call $c \in \{0,\ldots, m/2\}$ \textsf{good at query $k$} if the following holds.
Order the positions in $\{i\in Q_{\leq k}(A_{[k-1]}):d(i)>c\}$ by increasing depth, breaking ties arbitrarily, and denote the $\ell\th$ position by $j_\ell$.
Then, for every $\ell$, there are at least $\ell$ unqueried depths strictly between $c$ and $d(j_{\ell})$.

\begin{claim}\label{claim:progress_partition}

There exist cutting points $C_1,\ldots,C_q$ such that $D_{k-1}\leq C_k\leq D_k$ for every $k\in[q]$, satisfying the following two properties:
\begin{enumerate}
    \item For every $k \in [q]$, if $C_k \neq D_k$, then $C_k$ is good at query $k$.

    \item The total length of the initial parts $(D_{k-1},C_k]$ satisfies $\sum_{k\in[q]}(C_k-D_{k-1}) < 2q$. 
\end{enumerate}

\end{claim}

\begin{proof} 
Consider the cutting points $C_1,\ldots,C_q$ defined as follows.       
For each $k\in [q]$, let $g_k$ be the smallest depth larger than or equal to $D_{k-1}$ that is good at query $k$, and set $C_k=\min\{g_k, D_k\}$.
Clearly, the first property of the claim is satisfied.
Call a depth $c$ \textsf{bad} at query $k$ if it is not good at query $k$. 
To show that the second property is satisfied, we show that there are fewer than $2q$ bad depths at query $q$.  
Since all depths in $[D_{k-1},C_k)$ are bad at query $k$, and a depth that is bad at query $k$ remains bad at every subsequent query, this will establish the second property.

For any depth $c\in\{0,...\,, m/2\}$, we call a query $i\in Q_{\leq k}(A_{[k-1]})$ a \textsf{violating} query for $c$ at query $k$ if $d(i) > c$, and the number of queries in $Q_{\leq k}(A_{[k-1]})$ whose depth is in the range $(c, d(i)]$ is greater than the number of unqueried depths in the range $(c,d(i)]$.
Thus, $c$ is good at query $k$ if and only if it has no violating query at query $k$.

To bound the number of bad depths at query $q$, consider the following iterative process.
Let $b_1$ be the smallest depth in $\{0,...\,, m/2\}$ that is bad at query $q$, and let $v_1$ be the depth of an arbitrary query that is violating for $b_1$ at query $q$.
Now continue the same process from $v_1$: 
Let $b_2$ be the smallest depth from $v_1$ onwards that is bad at query $q$, and choose a violating query of depth $v_2$.
Repeat until there are no bad depths to choose.

For each $j$, let $t_j$ be the number of query positions among $Q_{\leq q}(A_{[q-1]})$ whose depths lie in $(b_j,v_j]$.
By the definition of a violating query, the interval $(b_j,v_j]$ contains fewer than $t_j$ unqueried depths, which implies that $|(b_j,v_j]| < 2t_j$.
Note that since the intervals $(b_j,v_j]$ are pair-wise disjoint, we have $\sum_j 2t_j \leq 2q$.
Since the intervals $[b_j,v_j)$ (across all $j$) cover all bad depths, and we have $\sum_j |[b_j,v_j)| < \sum_j 2t_j \leq 2q$, the second property follows.
\end{proof}

For every $k\in[q]$, denote the second part of the progress by $P'_k \eqdef D_k - C_k$.
The following lemma bounds the probability that $P'_k$ is large.
\begin{lemma}\label{claim:progress_prob}
    Suppose that $b \geq C \cdot q^{3/2}$, for a sufficiently large constant $C$.
    Then, for every query $k\in[q]$, and every $p>0$, it holds that $\Pr\left[ P'_k \geq p \right] = O\left(\frac{q}{b \cdot \sqrt{p}}\right)$. 
\end{lemma} 

Before proving Lemma~\ref{claim:progress_prob}, we use it to establish Lemma~\ref{claim:total_progress_expectation}.
We will also make use of the following fact: 

\begin{fact}\label{fact:sqrtsqrt}
    For every $t\in \mathbb{N}$ it holds that $\sum_{k\in[t]}{\frac{1}{\sqrt{k}}} = O(\sqrt{t})$.\footnote{
        This fact can be proved, e.g., by showing via a simple induction on $t$ that $\sum_{k\in[t]}{\frac{1}{\sqrt{k}}} \leq 2\cdot \sqrt{t} - 1$ for every $t\in \mathbb{N}$.
    }
\end{fact}

\begin{proof}[{\textnormal{\textbf{Proof of Lemma~\ref{claim:total_progress_expectation}:}}}]
    First, for every $k\in [q]$ we have:
    \begin{align*}
        \E\left[P'_k\right] 
        = \sum_{p\in [n/(2b)]}{\Pr\left[P'_k \geq p\right]} 
        = \sum_{p \in [n/(2b)]}{O\left(\frac{q}{b \cdot \sqrt{p}}\right)} 
        = O\left( \frac{q}{b} \cdot \sqrt{n/b} \right)
    \end{align*}
    where the second equality is by Lemma~\ref{claim:progress_prob}, and the last equality uses Fact~\ref{fact:sqrtsqrt}.
    Now, recall that $P_k = C_k - D_{k-1} + P'_k$. 
    By Claim~\ref{claim:progress_partition}, it holds that
    $
        \sum_{k\in [q]} (C_k - D_{k-1}) < 2q
    $.
    Therefore,
    \begin{align*}
        \E\Bigg[\sum_{k\in[q]} P_k \Bigg] 
        < 2q + \sum_{k\in[q]} \E\left[ P'_k \right] 
        = 2q + \sum_{k\in[q]} O\left(q \cdot \frac{\sqrt{n}}{b^{3/2}}\right)
        = O\left(q + q^2 \cdot \frac{\sqrt{n}}{b^{3/2}}\right) 
    \end{align*}
    concluding the proof.
\end{proof}

It remains to prove Lemma~\ref{claim:progress_prob}. 
To do so, we will use the following fact:

\begin{fact}[cf.{~\cite{MR08}}]\label{fact:bernoulli:generalized}
    Let $t,b\in\mathbb{N}$, and let $X_1,\ldots,X_t$ be $t$ independent random variables, each distributed uniformly over $\{0,1,\ldots,b\}$. 
    Then, for every $j\in \mathbb{N}$, it holds that $\Pr\left[\sum_{k\in[t]}{X_k} = j\right] = O\left( \frac{1}{b \cdot \sqrt{t}} \right)$.
\end{fact}

\begin{proof}[{\textnormal{\textbf{Proof of Lemma~\ref{claim:progress_prob}}}}]
Recall that our aim is to show that $\Pr\left[ P'_k \geq p \right] = O\left(\frac{q}{b \cdot \sqrt{p}}\right)$ for every $p>0$. 
We will prove the bound by showing it holds even when conditioning on any fixed value for the answers $A_{[k-1]}$.
Let $a_{[k-1]}$ be an arbitrary sequence in the support of $A_{[k-1]}$. 
Note that fixing the answers $A_{[k-1]}$ to $a_{[k-1]}$ fixes the previous query sets, as well as the current query position and the string to which it is made. 
Let us denote $i_k\eqdef Q_k(a_{[k-1]})$, $Q_{<k} \eqdef Q_{<k}(a_{[k-2]})$ and $Q'_{<k} \eqdef Q'_{<k}(a_{[k-2]})$.

Consider the event $P'_k \geq p$ under the condition $A_{[k-1]} = a_{[k-1]}$. 
Let us assume first that given the previous answers $a_{[k-1]}$, the $k\th$ query $i_k$ is made to the \emph{first string}.
The case where $i_k$ is made to the second string is analogous, and will be revisited at the end of the proof.
Notice that since the query $i_k$ is made to the first string, the event $P'_k \geq p$ implies that $D_k = d(i_k)$, and $D_{k-1} \leq d(i_k) - p$. 
That is, the following two events must hold:
\begin{enumerate}
    \item 
    The current query $i_k$ matches with some previous second-string query in $Q'_{<k}$.  
    Let us denote this event by $E_1$.

    \item 
    Denoting $I \eqdef \{ i\in[n]: d(i) > d(i_k) - p \}$, there are \emph{no} previous first-string queries in $Q_{<k} \cap I$ that match with some previous second-string query in $Q'_{<k}$.
    Let us denote this event by $E_2$.
\end{enumerate}
Thus,
\begin{align*}
    \Pr\Big[\,
        P'_k \geq p \;\,\Big|\; A_{[k-1]} = a_{[k-1]} 
    \,\Big] 
    \;\leq\; 
    \Pr\Big[\,
        E_1 \wedge E_2 \;\,\Big|\; A_{[k-1]} = a_{[k-1]} 
    \,\Big] 
\end{align*}

To bound the above probability, we will further condition on arbitrary values for the block variables $X_h$ of all blocks $h$ outside $I$, as well as all blocks $h$ inside $I$ that contain a previously queried position in $Q_{<k}$. 
Furthermore, we will condition on arbitrary values for \emph{all} second-string block variables in $X'$.

To define the foregoing condition more precisely, note that fixing any block variable $X_h$ fixes also the value of its mirror block $X_{m+1-h}$. 
Let us denote the set of all blocks whose variables will be fixed by $F$; that is, $F$ consists of all blocks $h$ that lie outside $I$ as well as all blocks $h$ where either block $h$ or $m + 1 - h$ contains a position in $Q_{<k}$.
The variables of the complement set of blocks $\ov{F}\eqdef[m]\setminus F$ will remain random.
Let $x_F = (x_h)_{h\in F}$ be an arbitrary sequence in the support of $X_F = (X_h)_{h\in F}$, and let $x'$ be an arbitrary sequence in the support of $X'$.
We will bound the following probability:
\begin{align*}
    \Pr\left[
        \;E_1 \wedge E_2 \;\middle| {\small \begin{array}{c} 
        A_{[k-1]} = a_{[k-1]}, \\[1ex] 
        X_F = x_F,\; X' = x' 
    \end{array}}
    \right] 
\end{align*}
See Figure~\ref{figure} for an illustration of the objects introduced so far.

\begin{figure}[h]
    \centering
    \includegraphics[width=0.96\textwidth]{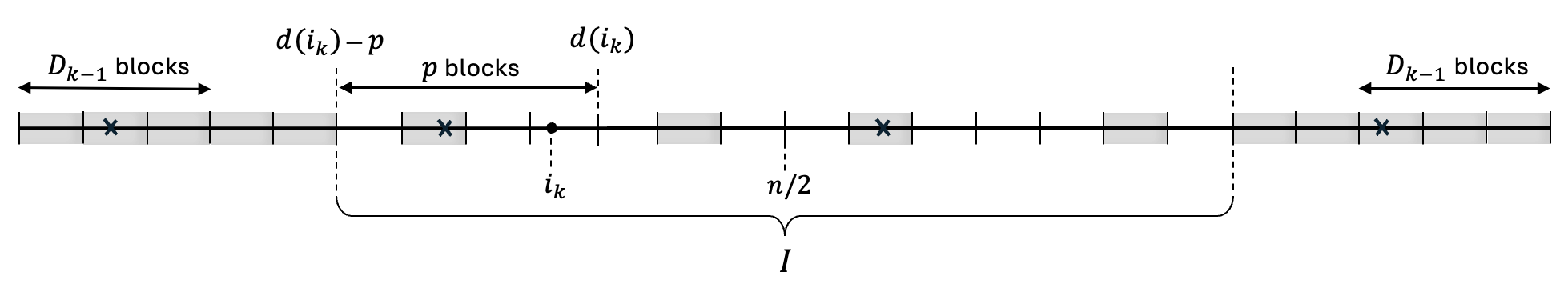}
    \vspace{-0.25em}
    \captionsetup{width=.92\linewidth,  font=small}
    \caption{
    An illustration of the first string in the input pair.
    The point labeled $i_k$ is the new query point, and the points marked by a cross (i.e.,\ ``{\small $\pmb{\times}$}'') represent the previous first-string query points (i.e., the points in $Q_{<k}$).
    The deepest previous first-string query point that matches with some previous second-string query is at depth $D_{k-1} \leq d(i_k)-p$.
    The range $I$ contains all points whose depth is greater than $d(i_k)-p$.
    In particular, any previous query point that lies inside $I$ necessarily does not match any of the previous query points in the second string. 
    The shaded blocks are the blocks in $F$; that is, the blocks whose variables we fix. 
    These include all blocks outside the range $I$, as well as all blocks containing previous query points or their mirror positions. 
    The remaining (i.e., unshaded) blocks are the blocks in $\ov{F}$. 
    The variables of these blocks remain random.
    }
    \label{figure}
\end{figure}

Consider the condition $A_{[k-1]} = a_{[k-1]}$.
Notice that it imposes two types of constraints.
First, it determines for every queried position whether it is a star-position or not. 
However, this is already determined by the block variables that we fixed (i.e., $X_F=x_F$ and $X'=x'$), since these blocks contain all previously queried positions.\footnote{
    We assume that the fixing of the block variables is compatible with the fixing of the previous answers; otherwise, the corresponding probability is trivially $0$.
}
Second, and more importantly, the condition $A_{[k-1]} = a_{[k-1]}$ determines the bit values at every non-star position that was queried. 
In more detail, recall that for every $r$, the two positions with rank $r$ contain a shared uniform bit $V_r$.
Let us denote by $J$ (resp., $J'$) the set of positions in $Q_{<k}$ (resp., $Q'_{<k}$) whose corresponding answer (according to $a_{[k-1]}$) is not a star symbol, and let us denote the values seen in $J$ and $J'$ by $V_{R(J)} \eqdef (V_{R(i)} )_{i\in J}$ and $V_{R'(J')} \eqdef ( V_{R'(i')} )_{i'\in J'}$, respectively. 
Then the condition $A_{[k-1]} = a_{[k-1]}$ fixes the values $V_{R(J)}$ and $V_{R'(J')}$.
Thus, we aim to bound the following probability:
\begin{align*}
    \Pr\left[
        \;E_1 \wedge E_2 \;\middle| {\small \begin{array}{c} 
        (V_{R(J)}, V_{R'(J')}) = \ovbar{v}, \\[1ex] 
        X_F = x_F,\; X' = x' 
    \end{array}}
    \right] 
\end{align*}
where $\ovbar{v}$ is a sequence of values determined from $a_{[k-1]}$ alone.

Now, consider the event $E_2$, which asserts that there are no previous first-string query positions in $I$ that match with a previously queried position in the second string.
We next show that under the conditions $X_{F} = x_F, X'=x'$, \emph{no matter how we fix the remaining variables $X_{\ov{F}} = (X_h)_{h\in \ov{F}}$ such that the event $E_2$ holds}, the joint distribution of $(V_{R(J)}, V_{R'(J')})$ remains the same. 

\begin{subclaim}\label{claim:same_dist}
    For every $x_{\ov{F}}\in \mathrm{supp}(X_{\ov{F}})$ such that the event $E_2$ holds when fixing $X_{\ov{F}} = x_{\ov{F}}$ and $X_{F} = x_F, X'=x'$, 
    it holds that the joint distribution of $(V_{R(J)}, V_{R'(J')})$ under the fixing $X_{\ov{F}} = x_{\ov{F}}, X_{F} = x_F, X'=x'$ \emph{is the same}.
\end{subclaim}

\begin{subproof}
    At a high level, the issue is that changing the value of $X_{\ov{F}}$ can change the ranks $R(J)$, which can affect whether there exist $i \in J$ and $i' \in J'$ with matching ranks.
    This can change the joint distribution of $(V_{R(J)},V_{R'(J')})$, because at equal ranks there are equal values, whereas values at distinct ranks are independent.
    However, the value of $X_{\ov{F}}$ can affect only the ranks of positions in $I$, whereas the event $E_2$ guarantees that there is no queried position in $I$ that matches with a queried position in the second string, meaning that all values seen in $I$ are uniform bits, independent of all other answers.

    To make the foregoing argument more precise, let us denote $V_{R(J \cap I)} = (V_{R(i)} )_{i\in J \cap I}$, and $V_{R(J \setminus I)} = (V_{R(i)} )_{i\in J \setminus I}$.
    Fix $X_F = x_F, X'=x'$. 
    Now, fixing $X_{\ov{F}}$ does not affect $V_{R(J \setminus I)}$ or $V_{R'(J')}$, because the corresponding ranks are already determined by $x_F$ and $x'$, respectively.
    Furthermore, no matter how we fix $X_{\ov{F}}$ such that the event $E_2$ holds, we have that $V_{R(J \cap I)}$ is a sequence of uniform bits, independent of $V_{R'(J')}$ (and certainly of $V_{R(J \setminus I)}$). 
    The claim follows.
\end{subproof}

We use the above claim to establish:

\begin{subclaim}\label{claim:removing_answers_condition}

It holds that
\begin{align}\label{eq:removing_answers_condition}
    \Pr\left[
        \;E_1 \wedge E_2 
        \;\middle| 
        {\small \begin{array}{c} 
            (V_{R(J)}, V_{R'(J')}) = \ovbar{v}, \\[1ex] 
            X_F = x_F,\; X' = x' 
        \end{array}}
    \right] 
    \;\leq\; 
    \frac{
        \Pr\left[
            \; E_1 
            \;\middle|\! 
            {\small \begin{array}{c} X_F = x_F,\, X' = x' \end{array}} \!
        \right]}{
        \Pr\left[ 
            \;E_2 
            \;\middle|\! 
            {\small \begin{array}{c} X_F = x_F,\, X' = x'  \end{array}} \!
        \right]}. 
\end{align}

\end{subclaim}

\begin{subproof}

For convenience, throughout this proof we omit the explicit condition $X_F = x_F,\, X' = x'$ from all probability expressions, with the understanding that this condition implicitly applies in all probability terms.
Let us denote $\ovbar{V} \eqdef (V_{R(J)}, V_{R'(J')})$. Under the foregoing convention, our aim is to show that $\Pr\left[ E_1 \wedge E_2 \,|\, \ovbar{V} = \ovbar{v} \right] \leq \Pr\left[ E_1 \right] / \Pr\left[ E_2 \right]$.

By Claim~\ref{claim:same_dist}, there exists a value $\alpha$, such that for every $x_{\ov{F}}$
for which event $E_2$ holds when fixing $X_{\ov{F}} = x_{\ov{F}}$ and $X_F = x_F, X' = x'$, it holds that $\Pr\left[\ovbar{V} = \ovbar{v} \,|\, X_{\ov{F}} = x_{\ov{F}}\right] = \alpha$.
It follows that $\Pr\left[ \ovbar{V} = \ovbar{v} \,|\, E_1 \wedge E_2 \right] = \alpha = \Pr\left[ \ovbar{V} = \ovbar{v} \,|\, E_2 \right] $.  
By Bayes' rule, we have
\begin{align*}
    \Pr\left[
            E_1 \wedge E_2 \;\middle|\; \ovbar{V} = \ovbar{v} 
        \right] 
    &= \Pr\left[
             \ovbar{V} = \ovbar{v} \;\middle|\, E_1 \wedge E_2 
        \right] 
        \cdot \frac{\Pr\left[  E_1 \wedge E_2  \right]}
                   {\Pr\left[\ovbar{V} = \ovbar{v}\right]} \\
    &= \Pr\left[
             \ovbar{V} = \ovbar{v} \;\middle|\, E_2 
        \right]
        \cdot 
        \frac{\Pr\left[  E_1 \wedge E_2  \right]}
             {\Pr\left[ \ovbar{V} = \ovbar{v} \right]} \\
    &\leq \frac{\Pr\left[ \ovbar{V} = \ovbar{v} \right]}
               {\Pr[E_2]}
          \cdot 
          \frac{\Pr\left[ E_1 \wedge E_2 \right]}
               {\Pr\left[ \ovbar{V} = \ovbar{v} \right]} \\
    &\leq \frac{\Pr\left[ E_1 \right]}
              {\Pr\left[ E_2 \right]}
\end{align*}
completing the proof.
\end{subproof}

We are thus left to bound Eq.~\eqref{eq:removing_answers_condition}.
At this point, we will use the first property guaranteed by Claim~\ref{claim:progress_partition}; that is, that $C_k = d(i_k) - P'_k$ is good at query $k$ 
(we use here the fact that $P'_k$ is positive, since $P'_k \geq p>0$).
Let us refer to a depth that is good at query $k$ as just ``good'', for short.
Note that we may assume without loss of generality that $d(i_k)-p$ is good.
Otherwise, let $p'\geq p$ be the smallest integer for which $d(i_k)-p'$ is good (if no such $p'$ exists, the event $P'_k\geq p$ is impossible).
Since $d(i_k) - P'_k$ is good, by the minimality of $p'$, whenever $P'_k\geq p$, we also have $P'_k\geq p'$.
Thus, we may replace $p$ by $p'$ throughout our analysis. 
Note that $p'$ depends only on $p$ and the fixed previous answers, and that increasing $p$ makes the bound $O\left(\frac{q}{b\cdot \sqrt{p}} \right)$ only stronger.

Using the fact that $d(i_k)-p$ is good, 
we proceed to bound the probabilities in Eq.~\eqref{eq:removing_answers_condition}. 
For the subsequent discussion we work under the fixing $X_F=x_F,X'=x'$, and we let $R_{x_F}(i)$ and $R'_{x'}(i)$ denote the ranks $R(i)$ and $R'(i)$, respectively, after this fixing.
We first consider the event $E_1$, which states that $i_k$ matches some second-string query in $Q'_{<k}$.
Note that since $d(i_k)-p$ is good, there are at least as many unqueried depths as queried depths in the range $(d(i_k)-p,d(i_k)]$, implying there are $\Omega(p)$ unqueried depths strictly between $d(i_k)-p$ and $d(i_k)$. 
It follows that there are at least $\Omega(p)$ unfixed block variables between the block containing $i_k$ and the nearest end of the string.
Recall that the rank of each position $i$ can be expressed using the sum over the block variables between $i$ and the nearest end of the string, and that these variables are independent.
Therefore, using Fact~\ref{fact:bernoulli:generalized} we have that for every second-string query $i'\in Q'_{<k}$ the probability that $R_{x_F}(i_k) = R'_{x'}(i') \neq \bot$ is at most $O\Big(\frac{1}{b\cdot \sqrt{p}}\Big)$.
By a union bound over the possible second-string queries $i'\in Q'_{<k}$ we obtain:
\begin{align}\label{eq:E1}
    \Pr\left[ 
        \;E_1
        \;\middle|\!\! 
        {\small \begin{array}{c} X_F = x_F,\, X' = x'  \end{array}} \!
    \right]
    = O\left( \frac{q}{b \cdot \sqrt{p}} \right)
\end{align}

Next, consider the event $E_2$ in the denominator in Eq.~\eqref{eq:removing_answers_condition}. 
We will upper bound the probability of the complement event 
$\ov{E}_2$; 
namely, that there \emph{exists} a previous first-string query in $Q_{<k} \cap I$ that matches with a previous second-string query in $Q'_{<k}$.
Let us order the previous first-string queries in $Q_{<k} \cap I$ by their depth in increasing order, breaking ties arbitrarily, and denote the ordered positions by $j_1,\ldots,j_{|Q_{<k} \cap I|}$.
Since $d(i_k)-p$ is good, for every $\ell \in [|Q_{<k} \cap I|]$ there are at least $\ell$ unqueried depths strictly between $d(i_k)-p$ and $d(j_\ell)$.
Therefore, the probability that there exists a second-string query $i'\in Q'_{<k}$ such that $R_{x_F}(j_\ell)=R'_{x'}(i')\neq\bot$ is at most $O\big(\frac{q}{b\cdot\sqrt{\ell}}\big)$.
Taking a union bound over the queries in $Q_{<k} \cap I$, we get:
\begin{align}\label{eq:E2}
    \Pr\left[ 
            \;\ov{E}_2
            \;\middle|\!\! 
            {\small \begin{array}{c} X_F = x_F,\, X' = x'  \end{array}} \!
        \right]
    \,\leq\, \sum_{\ell \leq |Q_{<k} \cap I|}{O\left(\frac{q}{b\cdot \sqrt{\ell}}\right)} 
    \,=\, O\left(\frac{q}{b}\cdot \sqrt{q}\right)
    \;\leq\; \frac{1}{2}
\end{align}
where the equality uses Fact~\ref{fact:sqrtsqrt} together with $|Q_{<k} \cap I| \leq q$, and the last inequality uses the lemma's hypothesis that $b\geq C\cdot q^{3/2}$ for a sufficiently large constant $C$.
Combining Eq.~\eqref{eq:E2},~\eqref{eq:E1}, and~\eqref{eq:removing_answers_condition},
we obtain the desired $O\left( \frac{q}{b \cdot \sqrt{p}} \right)$ bound.

Recall that until now we have focused on the case where, given the previous answers $a_{[k-1]}$, the $k\th$ query $i_k$ is a \emph{first-string} query.
Let us now consider the case where $i_k$ is a \emph{second-string} query.
Note that in this case, the event $P'_k \geq p$ implies that some previous first-string query $i \in Q_{<k}$ matches with $i_k$, and $D_{k} = d(i)$.
To handle this case, we repeat the same argument above, with the following modification: 
First, we begin by taking a union bound over the possible choices of $i\in Q_{<k}$, which will replace the union bound over the possible choices of $i'\in Q'_{<k}$ taken in Eq.~\eqref{eq:E1}. 
We change the definition of $E_1$ to be the event that the chosen $i$ matches with $i_k$.
We then repeat the same argument, except that the chosen position $i \in Q_{<k}$ takes the role of $i_k$, and $i_k$ takes the role of $i'$.
\end{proof}

%% file: H_ub_proof.tex
\section{A non-adaptive tester for \texttt{ResStringEq}}\label{sec:ub_proof}

In this section, we present a \emph{non-adaptive} $\epsilon$-tester for \texttt{ResStringEq} that makes $O(n^{1/2+\delta})$ queries, where $\delta>0$ is an arbitrarily small constant. More precisely, we show:

\begin{theorem}\label{thm:upper_bound:restatement}
    For every constant $\delta>0$, there exists a \emph{non-adaptive} tester for \textup{\texttt{ResStringEq}} with query complexity $O(\left(1/\epsilon\right)^{O(\log(1/\delta))} \cdot n^{1/2 + \delta})$.
\end{theorem}

\smallskip
We follow the structure of the overview presented in Section~\ref{sec:overview:up}: 
We begin with a warm-up (Section~\ref{subsec:ub_proof:warmup}), where in Section~\ref{subsec:upper_bound:proof:warm_up1} we present a basic tester that is adaptive and uses $\widetilde{O}(n^{2/3})$ queries for every constant $\epsilon>0$, and in Section~\ref{subsec:upper_bound:proof:warm_up2} we show how we can modify the basic tester to be \emph{non-adaptive}, at the cost of increasing the query complexity to $\widetilde{O}(n^{4/5})$.
In Section~\ref{subsec:upper_bound:proof:actual} we present the actual (recursive) non-adaptive tester that achieves the desired $O(n^{1/2+\delta})$ bound.

\subsection{A warm-up}\label{subsec:ub_proof:warmup}

\subsubsection{A basic adaptive tester}\label{subsec:upper_bound:proof:warm_up1}

In this section we present, as an initial warm-up, an \emph{adaptive} tester for \texttt{ResStringEq} that makes $\widetilde{O}(n^{2/3})$ queries.
For a string $x\in\{0,1,*\}^n$ and a position $i\in[n]$, we define the \textsf{rank} of $i$ in $x$ as $r_x(i)\eqdef |\{j\leq i : x_j\neq *\}|$.\footnote{
    Note that in contrast to Section~\ref{sec:lb_proof}, here we define the rank using the same formula across all positions, rather than setting the rank of star positions to $\bot$.
}

\begin{algorithm}[a basic $\widetilde{O}(n^{2/3})$-query adaptive tester for \texttt{ResStringEq}]\label{alg:basic_adaptive}
On input $n\in \mathbb{N}, \epsilon > 0$, and oracle access to $(s,s') \in \{0,1,*\}^n \times \{0,1,*\}^n$, proceed as follows.
\begin{enumerate}[leftmargin=0.8cm]
    \item Let $L \,\eqdef\, n^{2/3}$, and $\Delta \,\eqdef\, \alpha_1 \cdot \epsilon \cdot L$, where $\alpha_1>0$ is a sufficiently small constant \textup{(}to be set in the analysis\textup{)}.
    
    \item \textup{\textsf{(Estimate the ranks):}}
    For each string $x\in\{s,s'\}$, independently: 
    Take $\, T \,\eqdef\, O\!\left(\frac{n^2}{\Delta^2}\log n\right)$ independent uniform samples from $[n]$, and query $x$ at the sampled positions. 
    For every $i\in[n]$ and every $t\in[T]$, let $X^{x,i}_t$ equal $1$ if the $t\th$ sampled position, denoted $j_t$, satisfies that $j_t \leq i$ and $x_{j_t} \neq *$; otherwise, let $X^{x,i}_t=0$. 
    Define
    $$
        \hat{r}_x(i) \;\eqdef\; \frac{n}{T}\cdot \sum_{t\in[T]} X^{x,i}_t.
    $$

    \item \textup{\textsf{(Check that the number of non-star symbols is approximately the same in both strings):}} 
    If $| \hat{r}_s(n) - \hat{r}_{s'}(n) | > 2\Delta $, reject. Otherwise, set $m \eqdef \min\{\hat{r}_s(n), \hat{r}_{s'}(n)\}$. 
    
    \item \textup{\textsf{(Partition into rank-segments):}} 
    Partition each string $x\in \{s, s'\}$ into \textup{\textsf{rank-segments}}, where the $k\th$ rank-segment of $x$ is defined as:
    $$
        R^x_k \;\eqdef\; \{i\in[n] : (k-1)\cdot L \;<\; \hat{r}_x(i) \;\leq\;  k\cdot L \}.
    $$
    
    We say that $R^x_k$ is \textup{\textsf{short}} if $|R^x_k| \leq L/(\alpha_2\cdot \epsilon)$, where $\alpha_2>0$ is a sufficiently small constant \textup{(}to be set in the analysis\textup{)}; otherwise, we say $R^x_k$ is \textup{\textsf{long}}.

    \item \textup{\textsf{(Check that the residual strings of corresponding rank-segments match up to boundary slack):}}
    Repeat the following check $O(1/\epsilon)$ times.
    \begin{enumerate}
        \item Select $k\in[m/L]$ uniformly at random.

        \item If either $R^s_k$ or $R^{s'}_k$ is long, continue to the next iteration.

        \item Otherwise \textup{(}i.e., both $R^s_k$ and $R^{s'}_k$ are short\textup{)}, query $s$ at every position in $R^s_k$, and query $s'$ at every position in $R^{s'}_k$.

        \item We say that two strings \textup{\textsf{match up to boundary slack $k\in \mathbb{N}$}} if, after deleting at most $k$ symbols from the beginning of one of the two strings, and at most $k$ extra symbols from the end of one of them \textup{(}possibly the same one\textup{)}, the remaining strings are equal.
        
        Check if the residual strings of the restriction of $s$ to $R^s_k$ and the restriction of $s'$ to $R^{s'}_k$ match up to boundary slack $2\Delta$.
        If this check fails, then halt and reject.
        
    \end{enumerate}
    If none of the iterations rejected, then accept.

\end{enumerate}

\end{algorithm}

\paragraph{Query complexity.}

The above tester makes $T = \widetilde{O}\left( n^2/(\epsilon \cdot L)^2 \right)$ queries for obtaining the estimates (Step~2), and $O(1/\epsilon) \cdot O(L/\epsilon) = O(L/\epsilon^2)$ queries for reading the short segments (Step~5c). 
Setting $L = n^{2/3}$, we get a tester with query complexity $\widetilde{O}\left( n^{2/3}/\epsilon^2 \right)$.

\subsubsection{Correctness}\label{subsubsec:upper_bound:proof:warm_up1:correctness}

\paragraph{Successful estimates.} We say that \textsf{all estimates are successful} if for each string $x\in\{s,s'\}$ and all $i\in[n]$ (simultaneously) it holds that $|\hat{r}_x(i) - r_x(i)| \leq \Delta$.

\begin{subclaim}\label{claim:successful_estimates}
    Assuming $T = C\cdot \frac{n^2}{\Delta^2} \cdot \log(n)$ for a sufficiently large constant $C$, it holds that all estimates are successful with probability at least $0.99$.
\end{subclaim}

\begin{subproof}
    Consider any $x\in\{s,s'\}$ and $i\in[n]$. 
    Note that $X_1^{x,i},\ldots,X_T^{x,i}$ are independent Bernoulli random variables, where for each $t\in[T]$, it holds that $\Pr[X_t^{x,i} = 1] = r_x(i)/n$ (because the sampled positions are uniform in $[n]$, and exactly $r_x(i)$ positions $j$ among $[n]$ satisfy that $j \leq i$ and $x_j \neq *$).
    Hence, by the Chernoff bound, it holds that
    $
        \Pr\!\big[ | \hat r_x(i) - r_x(i) | > \Delta \big] \leq 2\cdot \exp(-{2T\cdot(\Delta/n)^2})
    $.
    By setting $T=C\cdot \frac{n^2}{\Delta^2}\cdot\log(n)$ for a sufficiently large constant $C$, we have that this probability is at most $o(1/n)$.
    Taking a union bound over the two input strings $x\in\{s,s'\}$ and over all $i\in[n]$, the claim follows.
\end{subproof}

\paragraph{Completeness:} 

Consider an input $(s,s')\in\textup{\texttt{ResStringEq}}$. 
Observe that if all estimates are successful, the tester accepts.
Indeed, since $\textrm{Res}(s)=\textrm{Res}(s')$ we have $r_s(n)=|\textrm{Res}(s)| = |\textrm{Res}(s')| = r_{s'}(n)$, and therefore if the estimates are successful the tester does not reject at Step~3.
Furthermore, if the estimates are successful, then for every $k\in[m/L]$ the true ranks of the first positions in $R^s_k$ and $R^{s'}_k$ both lie in the range $(k-1)\cdot L\pm\Delta$.
Hence, these two ranks differ by at most $2\Delta$. By a similar argument the true ranks of the last positions in $R^s_k$ and $R^{s'}_k$ also differ by at most $2\Delta$.
Since $\textrm{Res}(s)=\textrm{Res}(s')$, it follows that every pair of corresponding rank-segments passes the check for matching up to boundary slack (Step~5d), and the tester accepts. 
Completeness follows.

\paragraph{Soundness:} 

For $k\in[m/L]$, we call the pair of corresponding rank-segments $(R^s_k, R^{s'}_k)$ \textsf{bad} if both segments are short, and the residual strings obtained after restricting $s$ and $s'$ to these segments do not match up to boundary slack $2\Delta$.

\begin{subclaim}\label{claim:many_bad_blocks}
    Assume the input pair $(s,s')$ is $\epsilon$-far from \textup{\texttt{ResStringEq}}.
    If the estimates are successful and Algorithm~\ref{alg:basic_adaptive} has not rejected at Step~3, then at least $\Omega(\epsilon \cdot n/ L)$ of the corresponding rank-segments are bad.
\end{subclaim}

From the above claim it follows that each iteration in Step~5 rejects with probability at least $\Omega(\epsilon)$ (because the total number of corresponding rank-segments is $m/L \leq n/L$).  
Hence, w.h.p.\ at least one of the $O(1/\epsilon)$ iterations rejects, and soundness follows.

\begin{subproof}[\textup{\textsf{Proof of Claim~\ref{claim:many_bad_blocks}:}}]
    Suppose, towards a contradiction, that fewer than $\alpha_3 \cdot \epsilon \cdot n/ L$ corresponding pairs are bad, where $\alpha_3>0$ is a sufficiently small constant.
    We show that we can modify $(s,s')$ in at most $\epsilon\cdot n$ 
    positions so that it becomes a pair in \texttt{ResStringEq}, contradicting the claim's hypothesis.
    
	We modify $(s,s')$ as follows. 
    For every pair of corresponding rank-segments that is not bad and in which both segments are short, the residual strings in the two segments match up to boundary slack $2\Delta$. 
    We keep the non-star symbols in the matching parts of all such pairs, and replace every other non-star symbol in $s$ and $s'$ by a star.
	We bound the number of modified positions as follows:
    \begin{itemize}[leftmargin=0.8cm]

        \item First consider the pairs that are not bad and in which both segments are short.
        In each such pair, we modify at most $2\Delta$ non-star symbols at the beginning of one of the segments and at most $2\Delta$ non-star symbols at the end of one of the segments.
        Hence, we modify at most $4\Delta = O(\alpha_1\cdot \epsilon \cdot L)$ symbols.
        Since there are at most $n/L$ corresponding rank-segments, the total number of modifications in these pairs is at most $O(\alpha_1\cdot\epsilon\cdot n)$.

        \item Consider pairs in which at least one segment is long. Recall that the length of each long rank-segment is greater than $L/(\alpha_2\cdot\epsilon)$.
        Hence, each string can have at most $\alpha_2\cdot\epsilon\cdot n/L$ long rank-segments.
        Since the estimates are successful, every rank-segment has at most $L+2\Delta=O(L)$ non-star symbols. 
        Thus, replacing all non-star symbols in pairs in which at least one rank-segment is long contributes at most $O(\alpha_2\cdot\epsilon\cdot n)$ modifications.

        \item Next, consider the bad pairs. By our assumption, there are fewer than $\alpha_3\cdot\epsilon\cdot n/L$ bad pairs.
        Since each pair contains at most $O(L)$ non-star symbols, replacing all non-star symbols in the bad pairs contributes at most $O(\alpha_3\cdot\epsilon\cdot n)$ modifications.

        \item Finally, after the last corresponding rank-segment, one string may have an additional suffix. 
        Since Step~3 did not reject (by the claim's hypothesis) and the estimates are successful, this suffix has at most $O(\Delta) = o(\epsilon\cdot n)$ non-star symbols.

    \end{itemize}
    Taking $\alpha_1$, $\alpha_2$ and $\alpha_3$ sufficiently small, the total number of modifications is at most $\epsilon\cdot n$, establishing the claim. 
\end{subproof}

\subsubsection{Making the basic tester non-adaptive}\label{subsec:upper_bound:proof:warm_up2}

In this section we show how we can modify the basic tester from the previous section to be \emph{non-adaptive}, at the cost of increasing the query complexity to $\widetilde{O}(n^{4/5})$.
Note that the queries made in Step~5 of Algorithm~\ref{alg:basic_adaptive} are inherently adaptive, because the segments queried in this step are defined based on the estimates obtained in Step~2.
To get a non-adaptive tester, we modify Algorithm~\ref{alg:basic_adaptive} so that it pre-samples sufficiently many fixed blocks from both strings, such that w.h.p.\ both rank-segments in some bad pair are contained in sampled blocks.

\begin{definition}[blocks]\label{def:grid_blocks}
    For a string $x\in \{0,1,*\}^n$, the \textup{\textsf{covering of $x$ by blocks of length $b$ and spacing $\tau$}} is the set of all intervals of length $b$ starting at positions $1,1+\tau,1+2\tau,\ldots,1+\lfloor (n-1)/\tau \rfloor \cdot \tau$, where the last intervals are truncated at $n$.
    The intervals are referred to as \textup{\textsf{blocks}}.
\end{definition}

Note that any interval of positions in $x$ whose length is at most $b - \tau$ is necessarily contained in some block.
In the following, we use $b=2L/(\alpha_2\cdot\epsilon)$ and $\tau = L/(\alpha_2\cdot\epsilon)$, which ensures that every short rank-segment (as defined in Step~4 of Algorithm~\ref{alg:basic_adaptive}) is contained in a block.

\begin{algorithm}[a basic $\widetilde{O}(n^{4/5})$-query\footnote{
    For convenience, the query complexity of this construction is bounded only in expectation. 
    To get a worst-case query bound, alter the construction to halt and accept whenever the number of queries exceeds its expectation by a sufficiently large constant factor.
}
non-adaptive tester for \texttt{ResStringEq}]\label{alg:basic_non_adaptive}

On input $n\in \mathbb{N}$, $\epsilon > 0$, and oracle access to $(s,s') \in \{0,1,*\}^n \times \{0,1,*\}^n$, proceed as follows.

\begin{enumerate}[leftmargin=0.8cm]

    \item 
    Let $L \,\eqdef\, n^{3/5}/\epsilon^{2/5}$, $\Delta \,\eqdef\, \alpha_1 \cdot \epsilon \cdot L$ and $b \eqdef 2\cdot L/(\alpha_2 \cdot \epsilon)$, 
    where $\alpha_1$ and $\alpha_2$ are the same constants used in Steps~1 and~4 of Algorithm~\ref{alg:basic_adaptive}, respectively.

    \item \textup{\textsf{Query-Phase:}} 
    For each string $x\in\{s,s'\}$, independently:
    \begin{enumerate}

        \item \textup{\textsf{(Prepare estimation queries):}} 
        Take $\, T \,\eqdef\, O\!\left(\frac{n^2}{\Delta^2}\log n\right)$ independent uniform samples from $[n]$, and query $x$ at the sampled positions.

        \item \textup{\textsf{(Select \& query blocks):}} 
        Consider the covering of $x$ by blocks of length $b$ and spacing $\tau \eqdef b/2$ \textup{(}see Definition~\ref{def:grid_blocks}\textup{)}.
        Sample each block in the covering independently with probability 
        $p \eqdef O\left(\sqrt{\frac{L}{\epsilon^2 \cdot n}}\right)$, and query each selected block in its entirety.

    \end{enumerate}

    \item \textup{\textsf{Decision-Phase:}}
    \begin{enumerate}
        \item \textup{\textsf{(Partition into rank-segments)}} Using the estimation queries \textup{(}made in Step 2a\textup{)}, perform Steps 2--4 of Algorithm~\ref{alg:basic_adaptive}.

        \item \textup{\textsf{(Check if the residual strings of corresponding rank-segments match up to boundary slack):}}
        Go over all pairs of corresponding rank-segments $(R^{s}_k,R^{s'}_k)$ for $k\in [m/L]$.
        For each such pair, check
        if there exists a pair of blocks $B$ and $B'$ that were queried in $s$ and $s'$, respectively \textup{(}in Step 2b\textup{)}, such that $R^s_k \subseteq B$ and $R^{s'}_k \subseteq B'$.
        If so, check if the residual strings of the restriction of $s$ to $R^{s}_k$ and the restriction of $s'$ to $R^{s'}_k$ match up to boundary slack $2\Delta$ \textup{(}as in Step~5d of Algorithm~\ref{alg:basic_adaptive}\textup{)}.
        If the check fails, then halt and reject. 

        \item If no check performed in the last step rejected, then accept. 
    \end{enumerate}

\end{enumerate}

\end{algorithm}

\paragraph{Query complexity.}
The above tester makes $\widetilde{O}(n^{4/5}/\epsilon^{6/5})$ queries in expectation:
As before, the tester makes $\widetilde{O}(n^2/(\epsilon^2\cdot L^2))$ queries for obtaining the estimates (in Step~2a).
We next bound the number of queries made in Step~2b. 
Notice that there are $O(n/\tau)$ blocks in total.
Since each block is sampled with probability $p=O(\sqrt{L/(\epsilon^2\cdot n)})$, the expected number of sampled blocks is $O(p\cdot n/\tau)$.
Each sampled block has length $b=2\tau$, therefore the expected number of queries made in Step~2b is $O(p \cdot n\cdot b/\tau) = O(p\cdot n) = O(\sqrt{n\cdot L/\epsilon^2})$.
Thus, the expected query complexity of the tester is $\widetilde{O}\big( n^2/(\epsilon^2\cdot L^2)+\sqrt{n\cdot L}/\epsilon \big)$.
Balancing the two terms, we set $L= n^{3/5}/\epsilon^{2/5}$, giving us expected query complexity $\widetilde{O}(n^{4/5}/\epsilon^{6/5})$.
(As mentioned before, to get a \emph{worst-case} query bound, alter the construction to halt and accept if the number of queries exceeds its expectation by a sufficiently large constant factor.)

\paragraph{Correctness.}

Completeness follows exactly as in Section~\ref{subsubsec:upper_bound:proof:warm_up1:correctness}. 
Furthermore, note that Claim~\ref{claim:many_bad_blocks} still holds as well. 
That is, for every input $(s,s')$ that is $\epsilon$-far from \texttt{ResStringEq}, if the estimates are successful and the tester did not reject at Step~3a (i.e., when checking that the number of non-star symbols is approximately the same in both strings), then there are at least $\Omega(\epsilon \cdot n/L)$ bad pairs of corresponding rank-segments.
It remains to show that, with high probability, both rank-segments in some bad pair are contained in blocks sampled in Step~2b.

\begin{subclaim}\label{claim:birthday}
    Assume that the input pair $(s,s')$ is $\epsilon$-far from \textup{\texttt{ResStringEq}}, that the estimates are successful, and that the tester has not rejected at Step~3a.
    Then, with probability at least $0.99$, there exists a bad pair of corresponding rank-segments such that each rank-segment in the pair is contained in a block sampled in Step~2b. 
\end{subclaim}

\begin{subproof}
    Let us assign to each short rank-segment $R^x_k$ a block $B$ such that $R^x_k \subseteq B$.
    Recall that each segment in a bad pair is by definition short, so it has an assigned block.
    Notice that each block can be assigned to at most $O(1/\epsilon)$ rank-segments, because each rank-segment is of length at least $L-2\Delta$ whereas each block is of length at most $O(L/\epsilon)$.
    Thus, since there are at least $\Omega(\epsilon \cdot \frac{n}{L})$ bad pairs of rank-segments, there must be a set $S$ of at least $\Omega(\epsilon^2 \cdot \frac{n}{L})$ bad pairs such that the assigned blocks of different pairs in $S$ are all \emph{distinct}.
    
    For each bad pair in $S$, consider the event that the assigned blocks of both rank-segments in the pair are sampled. 
    For a single bad pair, the probability of this event is $p^2$.
    Since the assigned blocks of the pairs in $S$ are distinct, these events are mutually independent for different pairs in $S$.
    Hence, the probability that none of these events occurs among the $\Omega(\epsilon^2 \cdot \frac{n}{L})$ pairs in $S$ is at most:
    $$
        \left(1 - p^2\right)^{\Omega\left(\epsilon^2 \cdot \frac{n}{L}\right)} 
        \leq e^{-\Omega\left(p^2 \cdot \epsilon^2 \cdot \frac{n}{L}\right)}
    $$
    Setting $p = \frac{C}{\sqrt{\epsilon^2 \cdot \frac{n}{L}}}$ for a sufficiently large constant $C$, the claim follows.
\end{subproof}

\subsection{The actual non-adaptive tester}\label{subsec:upper_bound:proof:actual}

We extend the non-adaptive tester from the preceding section into a recursive tester that achieves the desired $O(n^{1/2+\delta})$ query complexity.
The tester we construct will actually have a stronger acceptance guarantee: 
It accepts inputs even if their residual strings match up to boundary slack $0.1 \cdot \epsilon \cdot n$ (while rejecting inputs that are $\epsilon$-far from having their residual strings match).
Recall that matching up to boundary slack is defined as follows: 

\begin{definition}[match up to boundary slack]\label{def:approximate_match}
    We say that two strings \textup{\textsf{match up to boundary slack $k\in \mathbb{N}$}} if, after deleting at most $k$ symbols from the beginning of one of the two strings, and at most $k$ extra symbols from the end of one of them \textup{(}possibly the same one\textup{)}, the remaining strings are equal.
\end{definition}

The tester is composed of two recursive procedures: a Query procedure, which is run separately on each input string and makes all queries, and a Decision procedure, which receives the outputs of the Query procedure and decides whether to accept or reject the input pair.
The Decision procedure also receives endpoint parameters $[p_s,q_s]$ and $[p_{s'},q_{s'}]$ specifying the first and last positions of two substrings $s$ and $s'$ that are currently being compared; 
in the first call these substrings are simply the full input strings, whereas in recursive calls they will be rank-segments inside sampled blocks.
Both the Query and Decision procedures receive a round parameter $r\in\mathbb{N}$ specifying the remaining number of recursion rounds.
The Query procedure also receives a failure parameter $\eta>0$ specifying the allowed failure probability.

\begin{algorithm}[an $O(n^{1/2+\delta})$-query non-adaptive tester for \texttt{ResStringEq}]\label{alg:actual_tester}
    On inputs $n\in \mathbb{N}$, $\epsilon>0$ and query access to $(s,s')\in\{0,1,*\}^n\times \{0,1,*\}^n$, proceed as follows:

    \begin{enumerate} 

        \item Let $\eta \,\eqdef\, 1/3$, and $r \,\eqdef\, O(\log(1/\delta))$.
        
        \item For each input string $x\in \{s,s'\}$, run the Query procedure \textup{(}Algorithm~\ref{alg:query}\textup{)} with inputs 
        $n,\epsilon,\eta,r$ and query access to $x$.

        \item \textup{\textsf{(Check that the number of non-star symbols is approximately the same in both strings):}} 
        Using the estimation queries in the output of the Query procedure, 
        check that $|\hat{r}_s(n) - \hat{r}_{s'}(n)| \leq 0.2\cdot \epsilon \cdot n + 2\Delta$, where $\hat{r}_s(n)$ and $\hat{r}_{s'}(n)$ are as defined in Step~2 of Algorithm~\ref{alg:basic_adaptive}, and $\Delta$ is as defined in Algorithm~\ref{alg:query}.
        If the check fails, then halt and reject.
    
        \item Run the Decision procedure \textup{(}Algorithm~\ref{alg:decision}\textup{)} with inputs $n,\epsilon,r$, the two outputs of the Query procedure, and endpoint parameters $[p_s,q_s]=[p_{s'},q_{s'}] \,\eqdef\, [1,n]$.
        Output the Decision procedure's decision.
    \end{enumerate}
    
\end{algorithm}

\begin{algorithm}[the \textsf{Query} procedure]\label{alg:query}
    On inputs $N\in\mathbb{N}$, $\epsilon>0$, $\eta>0$, $r\in \mathbb{N}$ and query access to a string $x \in \{0,1,*\}^N$, proceed as follows:

    \begin{itemize}[leftmargin=0.7cm]
        \item If $r=0$: Query the entire input string $x$ and output the answers.
        \item Otherwise: 
        \begin{enumerate}[leftmargin=0.6cm]

            \item \textup{\textsf{(Set parameters):}} 
            \begin{itemize}[leftmargin=0.5cm,itemsep=1ex]
                \item Let $\alpha_1, \alpha_2, \alpha_3>0$ be sufficiently small constants to be set in the analysis. 
                \item Let $L \,\eqdef\, L(N,r)$ be a parameter to be set in the analysis \textup{(}see Section~\ref{sec:query_complex_recursive}\textup{)}.
                \item Let $\epsilon' \,\eqdef\, \alpha_1 \cdot \epsilon$, and 
                $\Delta \eqdef \alpha_3 \cdot \epsilon' \cdot L$.
                \item Let $N' \,\eqdef\, b \,\eqdef\, 2L/(\alpha_2 \cdot \epsilon)$ and $\eta' \,\eqdef\, \eta/(2N)$.
            \end{itemize}
            
            \item \textup{\textsf{(Prepare estimation queries):}} 
            Take $T \,\eqdef\, O\!\left(\frac{N^2}{\Delta^2}\log \big(\frac{N}{\eta}\big) \right)$ independent uniform samples from $[N]$, and query $x$ at the sampled positions.
            
            \item \textup{\textsf{(Select blocks and continue recursively):}} 

            \begin{enumerate}
                \item Consider the covering of $x$ by blocks of length $b$ and spacing $b/2$ \textup{(}see Definition~\ref{def:grid_blocks}\textup{)}.
                Sample each block in the covering independently with probability $p \eqdef O\!\left(\sqrt{ \frac{L\cdot \log(N/\eta)}{\epsilon^3 \cdot N} }\right)$.

                \item For each selected block, recursively run the Query procedure on inputs $N',\epsilon',\eta',\allowbreak r-1$, and query access to the restriction of $x$ to the selected block. 
            \end{enumerate}

        \end{enumerate}
        Output the query positions and answers \textup{(}from Step~2\textup{)}, the selected blocks \textup{(}from Step~3\textup{)}, and the output of all recursive calls.
    \end{itemize}
    
\end{algorithm}

The Decision procedure has the following specification.
It receives the outputs obtained by running the Query procedure (Algorithm~\ref{alg:query}) on two strings $u,u' \in \{0,1,*\}^N$, with parameters $N,\epsilon,\eta,r$.
It also receives the values of $N$, $\epsilon$ and $r$ used in these two executions. 
In addition, it receives endpoints $[p_s,q_s]$ and $[p_{s'},q_{s'}]$ specifying the first and last positions of a substring $s$ in $u$ and a substring $s'$ in $u'$, respectively. 
Denoting $n \eqdef (|s| + |s'|)/2$, the substrings $s$ and $s'$ are promised to satisfy $n\geq \frac{\alpha_2}{4\alpha_1}\cdot \epsilon\cdot N = \Omega(\epsilon\cdot N)$, as well as $||\textrm{Res}(s)| - |\textrm{Res}(s')|| \leq 0.2\cdot \epsilon \cdot n + 4\Delta$, where $\alpha_1, \alpha_2$ and $\Delta$ are as defined in Algorithm~\ref{alg:query}. 
The Decision procedure should satisfy the following guarantees, where the probability is over the randomness of the Query procedure.
\begin{enumerate}

    \item  It \emph{accepts} w.p.\ at least $1-\eta$ if $\textrm{Res}(s)$ and $\textrm{Res}(s')$ match up to boundary slack $0.1\cdot \epsilon\cdot n$ (see Definition~\ref{def:approximate_match}).
    
    \item It \emph{rejects} w.p.\ at least $1-\eta$ if one must modify more than $\epsilon\cdot n$ positions in the pair $(s,s')$ so that their residual strings match.
\end{enumerate}

 \smallskip
In the following, positions in the output of the Query procedure are indexed relative to $u$ and $u'$, whereas the rank-segments are indexed relative to $s$ and $s'$. 
We translate between these two ways of indexing positions when needed.
    
\begin{algorithm}[the \textsf{Decision} procedure]\label{alg:decision}

On inputs $N\in \mathbb{N}$, $\epsilon>0$, $r\in \mathbb{N}$, the outputs of the Query procedure, and the endpoints $[p_s,q_s]$ and $[p_{s'},q_{s'}]$ specifying the positions of the substrings $s$ and $s'$ to be compared, proceed as follows.

For each $x\in \{s,s'\}$, let $n_x \eqdef q_x-p_x+1$, and let $n \eqdef (n_s + n_{s'})/2$.

\begin{itemize}[leftmargin=0.6cm]

\item If $r=0$: The outputs of the Query procedure contain the full strings $s$ and $s'$.
Directly check if the residual strings of $s$ and $s'$ match up to boundary slack $0.1\cdot \epsilon\cdot n$, and accept if and only if this check passes.

\item Otherwise: 

\begin{enumerate}[leftmargin=0.6cm]

    \item \textup{\textsf{(Set parameters):}} Set $L,\Delta,N',\epsilon'$ exactly as in Step~1 of the Query procedure \textup{(}Algorithm~\ref{alg:query}\textup{)}.
    
    \item \textup{\textsf{(Estimate the ranks):}}
    For each string $x\in\{s,s'\}$:
    For each $t\in[T]$, let $j_t\in[N]$ denote the position of the $t\th$ query in the output of the Query procedure corresponding to $x$. 
    For each $i\in[n_x]$ and $t\in[T]$, let $X^{x,i}_t$ equal $1$ if $p_x\leq j_t\leq p_x+i-1$ and $x_{j_t-p_x+1}\neq *$; otherwise, let $X^{x,i}_t=0$.
    Define:
    $$
        \hat{r}_x(i) \;\eqdef\; \frac{N}{T}\cdot \sum_{t\in[T]} X^{x,i}_t.
    $$
    
    \item \textup{\textsf{(Define rank-segments for every offset):}}
    For each $x\in \{s,s'\}$ and every offset $h\in \{0, 1,\ldots, 0.1\cdot \epsilon \cdot n\}$, define the $k\th$ rank-segment of $x$ with offset $h$ as:
    $$
        R^{x,h}_k \;\eqdef\; \{i\in[n_x] \,:\, (k-1)\cdot L \;<\; \hat{r}_x(i) - h \;\leq\;  k\cdot L \}
    $$

    \item \textup{\textsf{(Search for a successful offset):}} Consider each of the two possible orderings of the strings $(x,x')\in \{(s,s'),(s',s)\}$.
    For each ordering, go over all offsets $h\in\{0,1,\ldots,0.1\cdot\epsilon\cdot n\}$ \textup{(}the offset $h$ is applied to the first string $x$, whereas the second string $x'$ is left with offset $0$\textup{)}.
    For each such offset $h$, proceed as follows:
    \begin{enumerate}[leftmargin=0.9cm, itemsep=5pt]
        \item Let
        $
            m_{(x,x'),h} \;\eqdef\; \min\{\hat r_x(n_x)-h,\; \hat r_{x'}(n_{x'})\}.
        $

        \item \textup{\textsf{(Recursively check corresponding rank-segments):}}
        Go over all pairs of corresponding rank-segments $(R^{x,h}_k,R^{x',0}_k)$ for $k\in [m_{(x,x'),h}/L]$.
        For each such pair, check if there exists a pair of selected blocks $B$ and $B'$ in the respective outputs of the Query procedure, 
        such that $R^{x,h}_k \subseteq B$ and $R^{x',0}_k \subseteq B'$, 
        where $B$ and $B'$ are represented with respect to their positions in $x$ and $x'$.\footnote{
            That is, if $B$ occupies positions $[a,b]$ in the string on which the corresponding Query procedure was run, then, for the purpose of this containment check, we view $B$ as the set $[ a - p_x + 1, b - p_x + 1]$, and similarly for $B'$.
        }
        If so, recursively invoke the Decision procedure on inputs $N',\epsilon',r-1$, the outputs of the Query procedure in the recursive calls with $B$ and $B'$, and the positions of $R^{x,h}_k$ and $R^{x',0}_k$ inside $B$ and $B'$, respectively.

        \item If any recursive call rejects, then $h$ fails; continue to the next offset.
        Otherwise \textup{(}i.e., all recursive calls accepted\textup{)}, $h$ is successful; halt and accept.
    \end{enumerate}
    If every offset $h$ failed for both orderings of the strings, then reject.

\end{enumerate}
\end{itemize}
\end{algorithm}

\subsubsection{Query complexity}\label{sec:query_complex_recursive}

Let $Q_r(N, \epsilon, \eta)$ be the expected number of queries made by the Query procedure (Algorithm~\ref{alg:query}). 
Recall that there are $O(N/b)$ blocks.
Therefore, the expected number of blocks selected in Step~3a is $O(p\cdot N/b) = O\left(\sqrt{ \frac{L\cdot \log(N/\eta)}{\epsilon^3 \cdot N} } \cdot \frac{\epsilon \cdot N}{L}\right) = O\left(\sqrt{ \frac{N\cdot \log(N/\eta)}{\epsilon \cdot L} }\right)$.
Hence, 
$$
    Q_r(N, \epsilon, \eta) = O\left( \frac{N^2}{\epsilon^2\cdot L^2} 
    \cdot \log\left( \frac{N}{\eta} \right) + \sqrt{ \frac{N\cdot \log(N/\eta)}{\epsilon \cdot L} } \cdot Q_{r-1}\Big(O(L/\epsilon),\, \Omega(\epsilon),\, \Omega\!\left(\eta/N\right)\!\Big)\right) 
$$
where $Q_0(N, \epsilon, \eta) = N$.
Inductively, we have:
$$
    Q_r(N,\epsilon,\eta) \leq \left(1/\epsilon\right)^{O(r)} \cdot \log^{O(r)}\left(N/\eta\right)\cdot N^{\gamma_r} 
$$
where $\gamma_r$ is determined as follows.
We choose $L = L(N,r)$ so as to balance the terms $\frac{N^2}{L^2}$ and $\sqrt{\frac{N}{L}}\cdot L^{\gamma_{r-1}}$, which gives $L \eqdef N^{3/(3+2\gamma_{r-1})}$.
With this choice, $\gamma_r = 2-2\cdot 3/(3+2\gamma_{r-1})$. 
Together with $\gamma_0 = 1$, this yields $\gamma_r = \frac{1}{2-(3/4)^r}$.

Thus, for every constant $\delta$, we may set $r=O(\log(1/\delta))$, $\eta = 1/3$, and $N=n$ and get expected query complexity:
$$
    \left(1/\epsilon\right)^{O(\log(1/\delta))}
    \cdot
    \log^{O(\log(1/\delta))}(n)\cdot n^{1/2 + \delta}
$$
Since the above bound holds for every constant $\delta>0$, we may apply it with $\delta/2$ in place of $\delta$ and use $\log^{O(\log(1/\delta))}(n) = O(n^{\delta/2})$, giving us query complexity
$
    O(\left(1/\epsilon\right)^{O(\log(1/\delta))}\cdot n^{1/2 + \delta})
$.

\subsubsection{Correctness}\label{subsubsec:ub:actual:correctness}

\paragraph{Successful estimates.} We say that \textsf{all estimates are successful} if for each $x\in\{s,s'\}$ and all $i\in[n_x]$ (simultaneously) it holds that $|\hat{r}_x(i) - r_x(i)| \leq \Delta$.
Observe that by a straightforward extension of Claim~\ref{claim:successful_estimates}, we have: 

\begin{subclaim}\label{claim:successful_estimates_generalized}
    If in Step~2 of the Query procedure \textup{(}Algorithm~\ref{alg:query}\textup{)}, we set $T = C\cdot \frac{N^2}{\Delta^2} \cdot \log\big(N/\eta\big)$ for a sufficiently large constant $C$, then all estimates are successful with probability at least $1-\eta/2$.
\end{subclaim}

\paragraph{Completeness:}

Consider any pair of strings $(s,s')$ such that $\textrm{Res}(s)$ and $\textrm{Res}(s')$ match up to boundary slack $0.1\cdot \epsilon\cdot n$, where $n \eqdef (|s|+|s'|)/2$.
We first show that the Decision procedure accepts such a pair with probability at least $1-\eta$.
Since (by Claim~\ref{claim:successful_estimates_generalized}) all estimates are successful with probability at least $1-\eta/2$, we assume all estimates are successful and show that under this assumption the Decision procedure accepts with probability at least $1-\eta/2$.

Since $\textrm{Res}(s)$ and $\textrm{Res}(s')$ match up to boundary slack $0.1\cdot \epsilon\cdot n$, there exists an ordering of the strings $(x,x') \in \{(s,s'), (s',s)\}$ and an offset $h\in\{0,1,\ldots,0.1\cdot \epsilon \cdot n\}$ such that, after deleting the first $h$ symbols from $\textrm{Res}(x)$, the two residual strings agree until one of them ends.
This means that, assuming the estimates are successful, in the iteration of Step 4 that considers this ordering $(x,x')$ and offset $h$, the residual strings in each pair of corresponding rank-segments $(R^{x,h}_k, R^{x',0}_k)$ match up to boundary slack 
$2\Delta = 2\cdot \alpha_3 \cdot \epsilon' \cdot L \leq 0.1\cdot \epsilon' \cdot (|R^{x,h}_k| + |R^{x',0}_k|)/2$ (where we take $\alpha_3$ to be sufficiently small, and rely on the fact that $\max\{|R^{x,h}_k|, |R^{x',0}_k|\} \geq L-2\Delta$).
By induction, each recursive call rejects with probability at most $\eta' = \eta/(2N)$. 
By a union bound over the at most $n/L=o(N)$ corresponding rank-segments, the Decision procedure rejects with probability at most $o(N)\cdot \eta' \leq \eta/2$, as desired.

Now consider the actual tester (Algorithm~\ref{alg:actual_tester}).
Since $\textrm{Res}(s)$ and $\textrm{Res}(s')$ match up to boundary slack $0.1\cdot \epsilon\cdot n$, if the estimates are successful the tester will not reject at Step~3 (when checking that $|\hat{r}_s(n) - \hat{r}_{s'}(n)| \leq 0.2\cdot \epsilon \cdot n + 2\Delta$).
By the previous paragraph, it follows that the tester accepts with probability at least $1-\eta = 2/3$, as desired.

\paragraph{Soundness:}

We first show soundness of the Decision procedure (Algorithm~\ref{alg:decision}), and derive the soundness of the tester (Algorithm~\ref{alg:actual_tester}) at the end of the section.
Throughout the following soundness analysis of the Decision procedure, fix a pair of strings $(s,s')$ and let $n=(|s|+|s'|)/2$.
Assume that one must modify more than $\epsilon\cdot n$ symbols in $(s,s')$ to make the two residual strings agree.
We also assume that $s$ and $s'$ satisfy both promises required by the Decision procedure; that is, $n\geq \frac{\alpha_2}{4\alpha_1}\cdot \epsilon\cdot N = \Omega(\epsilon\cdot N)$ and $||\textrm{Res}(s)|-|\textrm{Res}(s')||\leq 0.2\cdot\epsilon\cdot n+4\Delta$.

We say that a rank-segment $R^{x,h}_k$ is \textup{\textsf{short}} if $|R^{x,h}_k| \leq L/(\alpha_2 \cdot \epsilon) = b/2$; otherwise, we say $R^{x,h}_k$ is \textup{\textsf{long}}. 
For each ordering $(x,x')$ and each offset $h$ considered by Step~4, we call a pair of corresponding rank-segments $(R^{x,h}_k, R^{x',0}_k)$ \textsf{bad} if (1) both rank-segments are short; and (2) one must modify more than $\epsilon'\cdot (|R^{x,h}_k|+|R^{x',0}_k|)/2$ non-star positions in the restrictions of $x$ and $x'$ to these rank-segments in order to make the two residual strings agree.

\begin{subclaim}\label{claim:many_bad_blocks_generalized}
    Assume the estimates are successful. 
    Then, for each ordering of the strings $(x,x')\in\{(s,s'),(s',s)\}$ and every offset $h\leq 0.1\cdot\epsilon\cdot n$ considered by Step~4 of the Decision procedure, at least $\Omega(\epsilon\cdot n/L)$ of the corresponding rank-segments $(R^{x,h}_k,R^{x',0}_k)$, for $k\in[m_{(x,x'),h}/L]$, are bad.
\end{subclaim}

\begin{subproof}    
Suppose, towards a contradiction, that for some ordering of the strings $(x,x')$ and some offset $h$, fewer than $\alpha_4\cdot\epsilon\cdot n/L$ corresponding pairs are bad, where $\alpha_4>0$ is a sufficiently small constant.
We show that we can modify at most $\epsilon\cdot n$ non-star symbols in $(s,s')$ so that the two residual strings agree, contradicting the hypothesis.

We modify $(s,s')$ as follows. 
For each pair of corresponding rank-segments $(R^{x,h}_k, R^{x',0}_k)$ that is not bad and in which both rank-segments in the pair are short,
we are guaranteed that the residual strings in the two segments can be made equal by modifying at most $\epsilon'\cdot (|R^{x,h}_k|+|R^{x',0}_k|)/2$ non-star positions. 
We make these modifications for all such pairs, and replace every non-star symbol outside these pairs by a star.
We bound the number of modified positions by considering the following cases.
(The following generalizes the argument in Claim~\ref{claim:many_bad_blocks}, where only the first and last cases are changed.)
\begin{itemize}[leftmargin=0.8cm]

    \item First, for each pair of corresponding rank-segments that is not bad and both rank-segments in the pair are short, we modify at most $\epsilon'\cdot (|R^{x,h}_k|+|R^{x',0}_k|)/2$ positions, where $\epsilon' = \alpha_1\cdot\epsilon$.
    Summing over all $k\in [m_{(x,x'), h}/L]$, we get that the total number of modifications is at most $\alpha_1\cdot\epsilon \cdot (|x|+|x'|)/2 = \alpha_1\cdot\epsilon\cdot n$. 

    \item Consider pairs in which at least one rank-segment is long.
    Recall that the length of each long rank-segment is greater than $L/(\alpha_2\cdot\epsilon)$.
    Hence, each string can have at most $O(\alpha_2\cdot\epsilon\cdot n/L)$ long rank-segments.
    Since the estimates are successful, every rank-segment has at most $L+2\Delta=O(L)$ non-star symbols. 
    Thus, replacing all non-star symbols in pairs in which at least one rank-segment is long contributes at most $O(\alpha_2\cdot\epsilon\cdot n)$ modifications.

    \item Next, consider the bad pairs. By our assumption, there are fewer than $\alpha_4\cdot\epsilon\cdot n/L$ bad pairs.
    Since each pair contains at most $O(L)$ non-star symbols, replacing all non-star symbols in the bad pairs contributes at most $O(\alpha_4\cdot\epsilon\cdot n)$ modifications.

    \item It remains to count the non-star symbols that lie before or after the range covered by corresponding rank-segments.
    The initial offset in the first residual string contributes at most $h+\Delta \leq 0.1\cdot\epsilon\cdot n + o(\epsilon\cdot n)$ non-star symbols.
    Furthermore, one of the two residual strings may have a suffix that lies after the range covered by the corresponding rank-segments.
    Recall that we are promised that $d \eqdef ||\textup{Res}(x)|-|\textup{Res}(x')||\leq 0.2\cdot\epsilon\cdot n + 4\Delta$.
    Together with the offset $h$ and the estimation errors, the suffix can be of length at most $d + h + 2\Delta \leq 0.3 \cdot\epsilon\cdot n + o(\epsilon\cdot n)$.
    Altogether, there are at most $0.4 \cdot\epsilon\cdot n + o(\epsilon\cdot n) \leq 0.5\cdot\epsilon\cdot n$ non-star positions outside the range of corresponding rank-segments.
\end{itemize}

Overall, the number of modified symbols is at most $O(\alpha_1 \cdot \epsilon \cdot n) + O(\alpha_2 \cdot \epsilon \cdot n) + O(\alpha_4 \cdot \epsilon \cdot n) + 0.5\cdot \epsilon\cdot n$.
Taking $\alpha_1$, $\alpha_2$ and $\alpha_4$ sufficiently small, this number is at most $\epsilon\cdot n$, as desired.
\end{subproof}

Next, similarly to Claim~\ref{claim:birthday}, we have: 

\begin{subclaim}\label{claim:birthday_generalized}
    Assume the estimates are successful. 
    Then, for every ordering of the strings $(x,x')$ and every offset $h$ considered by Step~4 of the Decision procedure, with probability at least $1-\eta/(2N)$, there exists a bad pair of corresponding rank-segments $(R^{x,h}_k,R^{x',0}_k)$ such that each rank-segment in the pair is contained in a sampled block.
\end{subclaim}

\begin{subproof}
    Recall that we are promised that $n\geq \frac{\alpha_2}{4\alpha_1}\cdot \epsilon\cdot N = \Omega(\epsilon\cdot N)$.
    By Claim~\ref{claim:many_bad_blocks_generalized}, in every iteration of Step~4
    there are $\Omega(\epsilon \cdot n/L) = \Omega(\epsilon^2 \cdot N/L)$ bad pairs of rank-segments.
    As in Claim~\ref{claim:birthday}, it follows that there is a set of at least $\Omega(\epsilon^3 \cdot N/L)$ bad pairs of rank-segments that are contained in blocks that are \emph{distinct} across different pairs (because each block can contain at most $O(1/\epsilon)$ rank-segments).

    Recall that each block is sampled with probability $p = C \cdot \sqrt{ \frac{L\cdot \log(N/\eta)}{\epsilon^3 \cdot N} }$, where $C$ is a sufficiently large constant.
    Thus, similarly to the proof of Claim~\ref{claim:birthday}, the probability that in a fixed iteration of Step~4 there is no bad pair for which each rank-segment is contained in a sampled block is at most:
    $$
        \left(1 - p^2\right)^{\Omega\left(\epsilon^3 \cdot N/L\right)} 
        \leq e^{-\Omega\left(p^2 \cdot \epsilon^3 \cdot N/L\right)} = e^{-\Omega(C^2 \cdot \log(N/\eta))} \leq \eta/(2N),
    $$
    as desired.
\end{subproof}

\paragraph{\textnormal{\textsf{Concluding the soundness of the Decision procedure}}.}
By Claim~\ref{claim:successful_estimates_generalized}, the probability that the estimates are unsuccessful is at most $\eta/2$.
Thus, we assume the estimates are successful and show that, under this assumption, the Decision procedure accepts $(s,s')$ with probability at most $\eta/2$.
Consider any iteration of Step~4, where the ordering of the strings is $(x,x')$ and the offset is $h$. 
The iteration can falsely accept $(s,s')$ either because there was no bad pair for which both rank-segments were contained in a sampled block, or because a recursive call on a bad pair falsely accepted.
By Claim~\ref{claim:birthday_generalized}, the probability of the former event is at most $\eta/(2N)$.

As for the latter event, consider any bad pair $(R^{x,h}_k,R^{x',0}_k)$ and let $y$ and $y'$ be the restriction of $x$ to $R^{x,h}_k$ and the restriction of $x'$ to $R^{x',0}_k$, respectively.
Note that since the estimates are successful, it holds that the number of non-star symbols in both $y$ and $y'$ is $L\pm 2\Delta$, which implies that the promises needed for the Decision procedure hold for the pair $(y,y')$:
Denoting $n' \eqdef (|y|+|y'|)/2$, it holds that $n'\geq \frac{\alpha_2}{4\alpha_1}\cdot \epsilon'\cdot N' = \Omega(\epsilon'\cdot N')$, because $n' \geq L - 2\Delta \geq L/2$ whereas $N' = b = 2L/(\alpha_2\cdot \epsilon)$ and $\epsilon' = \alpha_1\cdot \epsilon$. 
Furthermore, we have $||\textup{Res}(y)|-|\textup{Res}(y')||\leq 4\Delta = 4\cdot \alpha_3 \cdot \epsilon' \cdot L \leq 0.2\cdot \epsilon'\cdot n'$
(where we take $\alpha_3$ to be sufficiently small, and use $n' \geq L-2\Delta$).
Thus, both promises are satisfied.
Therefore, by induction, the probability that a recursive call on $(y, y')$ falsely accepts is at most $\eta' = \eta/(2N)$. 
It follows that each iteration falsely accepts with probability at most $\eta/(2N)+\eta/(2N) =\eta/N$.
By a union bound over the at most $2\cdot 0.1\cdot \epsilon \cdot n < n/2 \leq N/2$ iterations, the probability that the Decision procedure falsely accepts $(s,s')$ is at most $\eta/2$, as desired. 

\paragraph{\textnormal{\textsf{Concluding the soundness of the tester}}.}

Consider any pair $(s,s')\in \{0,1,*\}^n\times \{0,1,*\}^n$ that is $\epsilon$-far from \texttt{ResStringEq}.
Assume that the estimates are successful and that the tester (i.e., Algorithm~\ref{alg:actual_tester}) has not rejected at Step~3 (where it checks that $|\hat{r}_s(n)-\hat{r}_{s'}(n)|\leq 0.2\cdot\epsilon\cdot n+2\Delta$). 
Then, when the tester invokes the Decision procedure, both the required promises on the inputs to the procedure are satisfied.
Indeed, in the initial invocation we have $N=n$, and hence clearly 
$n\geq \frac{\alpha_2}{4\alpha_1}\cdot \epsilon\cdot N = \Omega(\epsilon\cdot N)$ (assuming $\alpha_2$ is sufficiently small).
Furthermore, the check at Step~3, together with the successful estimates, implies that $||\textup{Res}(s)|-|\textup{Res}(s')||\leq 0.2\cdot\epsilon\cdot n+4\Delta$.
By the previous paragraph, it follows that the tester rejects with probability at least $1-\eta = 2/3$, as desired.

\subsection{On the adaptive tester for $\texttt{ResStringEq}$}\label{subsec:adaptive_RSE}

By composing the basic $\widetilde{O}(n^{2/3})$-query adaptive tester (Algorithm~\ref{alg:basic_adaptive}) with the recursive $O(n^{1/2+\delta})$-query non-adaptive tester shown above \textup{(}Algorithm~\ref{alg:actual_tester}\textup{)}, we obtain an $O(n^{2/5+\delta})$-query \emph{adaptive} tester for \texttt{ResStringEq}, analogous to the tester of~\cite{FMS18} for the \texttt{Dyck} languages.
Recall that Algorithm~\ref{alg:actual_tester} accepts inputs even if they match \emph{up to a boundary slack}. 
Hence, we can replace Steps~5c and~5d of Algorithm~\ref{alg:basic_adaptive} with a call to Algorithm~\ref{alg:actual_tester}.
The recursive call uses the proximity parameter $\epsilon'$ defined in Algorithm~\ref{alg:query}, and the tester will use the definition of $\Delta$ in Algorithm~\ref{alg:query}.
The resulting tester uses $\widetilde{O}(n^2/(\epsilon^2\cdot L^2))$ queries for obtaining the estimates, and $O(\left(1/\epsilon\right)^{O(\log(1/\delta))} \cdot L^{1/2 + \delta})$ queries for the  boundary-slack matching checks. 
Setting $L = n^{4/5}$ gives us query complexity $O(\left(1/\epsilon\right)^{O(\log(1/\delta))} \cdot n^{2/5 + \delta})$.
The correctness analysis of this tester follows straightforwardly from the correctness analysis in Section~\ref{subsubsec:ub:actual:correctness}.
Thus, we obtain:

\begin{theorem}\label{thm:upper_bound:adaptive}
    For any constant $\delta>0$, there exists an \textup{(}adaptive\textup{)} tester for \textup{\texttt{ResStringEq}} with query complexity 
    $O(\left(1/\epsilon\right)^{O(\log(1/\delta))} \cdot n^{2/5 + \delta})$.
\end{theorem}

As discussed in Section~\ref{subsec:intro:our_results}, compared with the tester obtained by combining the reduction of~\cite{FMS18} with their tester for $\texttt{Dyck}_2$, the tester in Theorem~\ref{thm:upper_bound:adaptive} improves the dependence of the query complexity on $\epsilon$ from $\left(1/\epsilon\right)^{O(\textup{poly}(1/\delta))}$ to $\left(1/\epsilon\right)^{O(\log(1/\delta))}$.

%% file: I_ub_dyck.tex
\section{Extending the non-adaptive tester to the \texttt{Dyck} languages}\label{subsec:dyck_extension}

In this section, we extend the non-adaptive tester presented for \texttt{ResStringEq} in the previous section to the $\texttt{Dyck}$ languages. Specifically, we establish:

\begin{theorem}\label{thm:upper_bound_dyck}
    For every $m\in \mathbb{N}$, 
    and for any \textup{(}arbitrarily small\textup{)} constant $\delta>0$, there exists a \emph{non-adaptive} tester for $\textup{\texttt{Dyck}}_m$ with query complexity $O(\left(1/\epsilon\right)^{O(\log(1/\delta))} \cdot n^{1/2 + \delta})$.
\end{theorem}

The tester closely follows the construction of~\cite{FMS18}, except that we modify their construction to be non-adaptive using an approach very similar to the one we used for the non-adaptive \texttt{ResStringEq} tester.
We rely on several notions developed in~\cite{PRR03} and~\cite{FMS18}, which we briefly survey in the following subsection. 

\subsection{Preliminaries}

\paragraph{Notation.} 
Throughout, we let $\Sigma_m$ denote an alphabet of parentheses of $m$ types.
For a string $s$, we let $s_{\geq i}$ and $s_{>i}$ denote the suffixes of $s$ starting at positions $i$ and $i+1$, respectively.
Similarly, let $s_{\leq i}$ and $s_{<i}$ denote the prefixes of $s$ ending at positions $i$ and $i-1$, respectively.

\paragraph{A reduction to $\textup{\texttt{Dyck}}_m$-\textup{\texttt{consistency}}.}
We say that a string is $\texttt{Dyck}_m$-\textsf{consistent} if it is a consecutive substring of some string in $\texttt{Dyck}_m$.
We define the property $\texttt{Dyck}_m$-\texttt{consistency} to be the set of all $\texttt{Dyck}_m$-consistent strings.
The following reduction was implicit in the work of~\cite{PRR03}, and was established explicitly in~\cite[Lem.~2.4]{FMS18}. 

\begin{lemma}\label{lem:dyck_consistency_reduction}
    If there exists an $\epsilon$-tester for $\textup{\texttt{Dyck}}_m$-\textup{\texttt{consistency}} with query complexity $q$, then there exists a $\Theta(\epsilon)$-tester for $\textup{\texttt{Dyck}}_m$ with query complexity $O(q + \epsilon^{-2}\cdot\log(1/\epsilon))$.
    Furthermore, if the $\textup{\texttt{Dyck}}_m$-\textup{\texttt{consistency}} tester is non-adaptive, then the resulting $\textup{\texttt{Dyck}}_m$ tester is non-adaptive as well.
\end{lemma}

The furthermore-claim (in the above lemma) holds since the tester for $\texttt{Dyck}_m$ simply invokes the tester for $\texttt{Dyck}_m$-\texttt{consistency} as well as the tester for $\texttt{Dyck}_1$ of~\cite{AKNS01}, and the latter tester is non-adaptive. 
We will thus focus on showing a non-adaptive tester for $\texttt{Dyck}_m$-\texttt{consistency}.

\paragraph{Excess parentheses~{\cite[Def.~6]{PRR03}}.}
Intuitively, the excess parentheses of a substring $s$ over $\Sigma_m$ are the parentheses in $s$ that must match with parentheses outside $s$. Formally, let $\Sigma_1 = \{0,1\}$ be the alphabet over one parenthesis type, where $0$ denotes the opening parenthesis and $1$ denotes the closing parenthesis.
Let $s$ be a string over $\Sigma_m$, and let $\mu(s)$ be the string obtained from $s$ by ``removing the types of the parentheses''; that is, by mapping every opening parenthesis to $0$ and every closing parenthesis to $1$.
Let $k$ and $\ell$ be the smallest integers such that $0^k\mu(s)1^{\ell} \in \textup{\texttt{Dyck}}_1$.
Then $k$ is called the \textup{\textsf{excess number of closing parentheses}} in $s$, and $\ell$ is called the \textup{\textsf{excess number of opening parentheses}} in $s$. We call $k$ and $\ell$ the \textup{\textsf{excess numbers}} of $s$. 
We denote $e_1(s) \eqdef k$ and  $e_0(s) \eqdef \ell$.
Let $n_0(s)$ and $n_1(s)$ denote the number of opening and closing parentheses in $s$, respectively. 
The excess numbers can be computed as follows.
\begin{equation}\label{eq:excess}
\begin{aligned}
    & e_1(s) \,=\,  \max_{s'\text{ prefix of }s} \left( n_1(s') - n_0(s')\right) \\
    & e_0(s) \,=\,  \max_{s'\text{ suffix of }s} \left( n_0(s') - n_1(s')\right)
\end{aligned}
\end{equation}

\begin{claim}\label{claim:successful_estimates_excess}
    With $O\big(\frac{n^2}{\Delta^2} \cdot \log(n/\eta)\big)$ uniform queries to a string $x\in \Sigma^n_m$, one can obtain estimates of the excess numbers of every consecutive substring of $x$ \textup{(}simultaneously\textup{)} up to an additive deviation of $\Delta$, with success probability $1-\eta$.
\end{claim}

\begin{proof}
    As in Claim~\ref{claim:successful_estimates} and~\ref{claim:successful_estimates_generalized}, $O\big(\frac{n^2}{\Delta^2} \cdot \log(n/\eta)\big)$ uniform queries suffice to estimate $n_0(s)$ and $n_1(s)$ for every \emph{prefix} $s$ of $x$ up to an additive deviation of $\Delta/4$, where all the estimates are successful simultaneously with probability at least $1-\eta$.
    For every consecutive substring $s$ of $x$, subtracting the estimates for the prefix ending immediately before $s$ and the prefix ending at the last position of $s$ gives estimates of $n_0(s)$ and $n_1(s)$ up to an additive deviation of $\Delta/2$.
    Using these estimates as substitutes in Eq.~\eqref{eq:excess}, we get an estimate of the excess numbers of any consecutive substring of $x$, up to deviation $\Delta$.
\end{proof}

\paragraph{The matching graph~\cite[Sec.~3.4]{PRR03}.}
For each string $x$ such that $\mu(x)\in \texttt{Dyck}_1$, there exists a unique perfect matching between opening and closing parentheses in $x$, such that each opening parenthesis is matched to a closing parenthesis, and such that the matching is ``non-crossing''.
By \textsf{non-crossing} we mean that if $x_{j_1}$ is matched to $x_{k_1}$, and $x_{j_2}$ to $x_{k_2}$ where $j_1 < k_1$, and $j_1 < j_2 < k_2$, then either $k_1 < j_2$ or $k_1 > k_2$. 
We denote this unique perfect matching by $M(x)$.
This definition is also extended to strings
$x$ for which $\mu(x)\notin \texttt{Dyck}_1$ as follows: 
Let $\tilde{x} = 0^{e_1(x)}\mu(x)1^{e_0(x)} \in \texttt{Dyck}_1$. 
We let $M(x)$ be the restriction of $M(\tilde{x})$ to pairs of parentheses that are both in $x$.

We define the \textsf{matching graph} of $x$ as follows.
Let $b$ be a parameter we will specify later. 
We partition $x$ into $n/b$ consecutive blocks of length $b$.
The vertices of the matching graph are the blocks of $x$. 
Two blocks $i<j$ are connected by an edge $(i,j)$ if and only if the matching $M(x)$ matches between excess opening parentheses in block $i$ and excess closing parentheses in block $j$.
The weight $w(i, j)$ of the edge $(i, j)$ is the number of excess opening parentheses in block $i$ that are matched by $M(x)$ to excess closing parentheses in block $j$.
By the non-crossing property of the matching $M(x)$, it follows that: 

\begin{claim}[{\cite[Claim~9]{PRR03}}]\label{claim:planar}
    For every string $x$ over $\Sigma_m$, the matching graph of $x$ is planar, and therefore has at most $3n/b$ edges. 
\end{claim}

The next claim shows that we can obtain estimates of the weights $w(i,j)$ for every pair of blocks $i<j$ in the matching graph.

\begin{claim}[{\cite[Claim~10]{PRR03}}]\label{claim:estimated_weights}
    Given the \textup{(}exact\textup{)} excess numbers of all consecutive substrings of a string $x\in\Sigma^n_m$, one can compute the weight $w(i,j)$ of every pair of blocks $i<j$ in the matching graph of $x$.
    Furthermore, given \emph{estimates} of the excess numbers of all consecutive substrings of $x$ to within additive deviation $\Delta$, one can compute estimates $\hat{w}(i,j)$ for every $i<j$, such that $\hat{w}(i,j) \in w(i,j) \pm 2\Delta$.
\end{claim}

Loosely speaking, the next claim states that for each $i<j$ we can estimate the interval in block $i$ that contains the excess opening parentheses that should match with excess closing parentheses in block $j$ and vice versa.
A similar claim is implicit in~\cite{PRR03} and~\cite{FMS18}; however, for convenience reasons, here we define these estimated intervals slightly differently.
Therefore, we include a proof of the claim below.

\begin{claim}\label{claim:matching_intervals}
    Let $x\in\Sigma^n_m$, and partition $x$ into consecutive blocks $B_1,\ldots,B_{n/b}$ of length $b$.
    For every pair of blocks $i<j$, let $I_{i,j}\subseteq B_i$ be the smallest interval containing all excess opening parentheses in $B_i$ that are matched by $M(x)$ to excess closing parentheses in $B_j$.
    Likewise, let $I_{j,i}\subseteq B_j$ be the smallest interval containing all excess closing parentheses in $B_j$ that are matched by $M(x)$ to excess opening parentheses in $B_i$.  
    We refer to $I_{i,j}$ and $I_{j,i}$ as the \textup{\textsf{matching intervals}} of blocks $(i,j)$.
    Then, given the \textup{(}exact\textup{)} excess numbers of all consecutive substrings of $x$, one can compute the matching intervals of any pair of blocks $i<j$.
    Furthermore, suppose that we are given \emph{estimates} of the excess numbers of all consecutive substrings of $x$ to within additive deviation $\Delta$. 
    Then, for every pair of blocks $i<j$, it is possible to compute \textup{\textsf{approximate matching intervals}} $\hat{I}_{i,j}\subseteq I_{i,j}$ and $\hat{I}_{j,i}\subseteq I_{j,i}$ such that: 
    \begin{enumerate}[label=(\arabic*)]
        \item $I_{i,j}$ has at most $O(\Delta)$ excess opening parentheses outside $\hat{I}_{i,j}$, and likewise $I_{j,i}$ has at most $O(\Delta)$ excess closing parentheses outside $\hat{I}_{j,i}$. 
        
        \item $\hat{I}_{i,j}$ has at most $O(\Delta)$ excess opening parentheses that are not excess in $I_{i,j}$, and likewise $\hat{I}_{j,i}$ has at most $O(\Delta)$ excess closing parentheses that are not excess in $I_{j,i}$.
    \end{enumerate}
\end{claim}

 \medskip
Note that the excess opening parentheses of $\hat{I}_{i,j}$ that are not excess in $I_{i,j}$ must form a \emph{suffix} of the excess opening parentheses of $\hat{I}_{i,j}$. 
Similarly, the excess closing parentheses of $\hat{I}_{j,i}$ that are not excess in $I_{j,i}$ must form a prefix of the excess closing parentheses of $\hat{I}_{j,i}$.  

\begin{proof}
    Consider any pair of blocks $i<j$.
    Let $y$ and $y'$ denote the restrictions of $x$ to $B_i$ and $B_j$, respectively, and let $z$ denote the restriction of $x$ to $\bigcup_{k\in [i+1,j-1]}B_k$ (i.e., the interval strictly between blocks $i$ and $j$).
    Denoting $I_{i,j} = [p, q]$, observe that, by the non-crossing property of the matching $M(x)$, it holds that $q$ is the leftmost position in block $i$ such that $e_0(y_{> q}) = e_1(z)$, and $p$ is the rightmost position in block $i$ such that $e_0(y_{\geq p}) = e_1(z) + w(i,j)$.
    Symmetrically, denoting $I_{j,i} = [q', p']$, it holds that $q'$ is the rightmost position in block $j$ such that $e_1(y'_{< q'}) = e_0(z)$, and $p'$ is the leftmost position in block $j$ such that $e_1(y'_{\leq p'}) = e_0(z) + w(i,j)$.

    We next define $\hat{I}_{i,j}$; the case of $\hat{I}_{j,i}$ is symmetric.
    For each substring $s$ of $x$, let $\hat{e}_0(s) \in e_0(s) \pm \Delta$ and $\hat{e}_1(s) \in e_1(s) \pm \Delta$ be the estimated excess numbers of $s$ (as guaranteed by the claim's hypothesis).  
    Let $\hat{w}(i,j) \in w(i,j) \pm 2\Delta$ be the estimated weight of the pair $(i,j)$ as guaranteed by Claim~\ref{claim:estimated_weights}.
    If $\hat{w}(i,j)\leq 14\Delta$, we define $\hat{I}_{i,j}$ and $\hat{I}_{j,i}$ to be the empty intervals.
    Since $\hat{w}(i,j)\in w(i,j)\pm 2\Delta$, in this case $w(i,j)\leq 16\Delta$, and the claim immediately follows.
    We thus assume henceforth that $\hat{w}(i,j)>14\Delta$, and therefore that $w(i,j)>12\Delta$.
    Let us denote $\hat{I}_{i,j} = [\hat{p}, \hat{q}]$.
    We take $\hat{q}$ to be the leftmost position in block $i$ such that $\hat{e}_0(y_{> \hat{q}}) \leq \hat{e}_1(z) + 2\Delta$ and $\hat{e}_1(y_{> \hat{q}}) \leq \Delta$, and we take $\hat{p}$ to be the rightmost position such that $\hat{e}_0(y_{\geq \hat{p}}) \geq \hat{e}_1(z) + \hat{w}(i,j) - 4\Delta$.

    We need to show that $[\hat{p},\hat{q}] \subseteq [p,q]$, that there are at most $O(\Delta)$ excess opening parentheses in $[p,q]$ that lie outside $[\hat{p},\hat{q}]$, and that there are at most $O(\Delta)$ excess opening parentheses of $[\hat{p},\hat{q}]$ that are not excess in $[p,q]$.    
    Consider first the right endpoints $q$ and $\hat{q}$. 
    Notice that since $q$ satisfies $e_0(y_{> q}) = e_1(z)$, it must satisfy $\hat{e}_0(y_{> q}) \leq \hat{e}_1(z) + 2\Delta$. 
    Furthermore, notice that $q$ must satisfy $e_1(y_{> q}) = 0$, and therefore $\hat{e}_1(y_{> q}) \leq \Delta$. 
    Thus, by the minimality of $\hat{q}$ it follows that $\hat{q} \leq q$. 
    On the other hand, notice that there can be at most $4\Delta$ excess opening parentheses between $\hat{q}$ and $q$:
    Since $\hat{q}$ satisfies $\hat{e}_0(y_{> \hat{q}}) \leq \hat{e}_1(z) + 2\Delta$, it follows that $e_0(y_{> \hat{q}}) \leq e_1(z) + 4\Delta = e_0(y_{> q}) + 4\Delta$.
    Finally, notice that any excess opening parenthesis of $[\hat{p},\hat{q}]$ that is not excess in $[p,q]$ must match with an excess closing parenthesis in $y_{> \hat{q}}$, and by definition of $\hat{q}$ we have $\hat{e}_1(y_{> \hat{q}}) \leq \Delta$, which implies $e_1(y_{> \hat{q}}) \leq 2\Delta$.
    Hence, there can be at most $2\Delta$ excess opening parentheses of $[\hat{p},\hat{q}]$ that are not excess in $[p,q]$.  

    Similarly, consider the left endpoints $p$ and $\hat{p}$. 
    Since $p$ satisfies $e_0(y_{\geq p}) = e_1(z) + w(i,j)$, it must satisfy $\hat{e}_0(y_{\geq p}) \geq \hat{e}_1(z) + \hat{w}(i,j) - 4\Delta$. 
    Thus by the maximality of $\hat{p}$ it follows that $\hat{p} \geq p$. 
    On the other hand, there can be at most $8\Delta$ excess opening parentheses between $p$ and $\hat{p}$:
    Since $\hat{p}$ satisfies $\hat{e}_0(y_{\geq \hat{p}}) \geq \hat{e}_1(z) + \hat{w}(i,j) - 4\Delta$, it follows that $e_0(y_{\geq \hat{p}}) \geq e_1(z) +  w(i,j) - 8\Delta = e_0(y_{\geq p}) - 8\Delta$.
    Lastly, notice that $\hat{p}\leq\hat{q}$: Recall that $w(i,j)>12\Delta$, whereas we showed that there are at most $4\Delta$ excess parentheses between $\hat{q}$ and $q$ and at most $8\Delta$ excess parentheses between $p$ and $\hat{p}$.
\end{proof}

Finally, observe that if $x$ is $\textup{\texttt{Dyck}}_m$-consistent, then for every pair of blocks $i<j$, the excess opening parentheses in $\hat{I}_{i,j}$ and the excess closing parentheses in $\hat{I}_{j,i}$ ``match up to boundary slack $O(\Delta)$'', in the following sense:

\begin{definition}[match of excess parentheses up to boundary slack]\label{def:excess_parentheses_approx_match}
    Let $s$ and $s'$ be two non-overlapping substrings of a string $x\in\Sigma_m^n$, where $s$ appears before $s'$ in $x$. 
    Let $\mathcal{E}_0(s)$ denote the subsequence of excess opening parentheses in $s$, and let $\mathcal{E}_1(s')$ denote the subsequence of excess closing parentheses in $s'$.
    We say that $\mathcal{E}_0(s)$ and $\mathcal{E}_1(s')$ \textup{\textsf{match up to boundary slack $k\in\mathbb{N}$}} if, after deleting at most $k$ parentheses from the end of $\mathcal{E}_0(s)$, at most $k$ parentheses from the beginning of $\mathcal{E}_1(s')$, and at most $k$ parentheses either from the beginning of $\mathcal{E}_0(s)$ or from the end of $\mathcal{E}_1(s')$ \textup{(}but not both\textup{)}, the following holds:
    There exists a perfect non-crossing matching between the remaining parentheses in $\mathcal{E}_0(s)$ and the remaining parentheses in $\mathcal{E}_1(s')$, such that for every matched pair of parentheses, the two parentheses have the same type, and the number of parentheses in $x$ between them is even.
\end{definition}

To compare the above definition with Definition~\ref{def:approximate_match}, consider $\mathcal{E}_0(s)$ in reverse order.
Under this ordering, here we allow deleting up to $k$ symbols from the beginning of \emph{both} sequences, rather than only one of them.
Furthermore, here we have an additional requirement that the number of symbols between each matched pair is even.

\subsection{The excess parentheses matching procedure}

Before presenting the non-adaptive tester for $\texttt{Dyck}_m$-\texttt{consistency}, we describe a (non-adaptive) procedure for comparing the excess opening parentheses of one substring $s$ with the excess closing parentheses of another substring $s'$, and deciding (approximately) whether they match up to boundary slack $0.1\cdot\epsilon\cdot n$ (per Definition~\ref{def:excess_parentheses_approx_match}).
This procedure, which we call the \textsf{excess parentheses matching procedure}, is a modified version of the \texttt{ResStringEq} Decision procedure (i.e., Algorithm~\ref{alg:decision}), and it gets the outputs of the same Query procedure that was used in the \texttt{ResStringEq} tester (i.e., Algorithm~\ref{alg:query}).

At a high level, the procedure proceeds analogously to the \texttt{ResStringEq} Decision procedure.
It estimates the excess numbers (as guaranteed by Claim~\ref{claim:successful_estimates_excess}), and uses these estimates to rank the positions in $s$ from right to left according to the excess opening parentheses, and the positions in $s'$ from left to right according to the excess closing parentheses. 
It then partitions the strings into segments that each contain approximately $L$ excess parentheses, and recursively compares corresponding segments.

However, here we need to choose the segment boundaries more carefully.
Note that if we choose the boundaries arbitrarily, then a segment of $s$ may have opening parentheses that are excess in the segment but not excess in $s$, and symmetrically for excess closing parentheses in segments of $s'$.
The recursive call would then include these parentheses in the comparison, whereas we want to compare only the excess parentheses of $s$ and $s'$.
We therefore choose the boundaries more carefully so that each segment can have at most $O(\Delta)$ excess parentheses that are not excess in the original string.

Finally, in contrast to the definition of boundary slack that we used in the \texttt{ResStringEq} Decision procedure (i.e., Definition~\ref{def:approximate_match}), here the definition allows deleting at most $0.1 \cdot \epsilon \cdot n$ parentheses from the beginning of \emph{both} sequences before comparing them (where the excess opening parentheses are read from right to left).
To account for this, we skip the first $0.1\cdot\epsilon\cdot n$ excess parentheses in one string (according to the estimates), and iterate over all offsets between $0$ and $0.2 \cdot \epsilon \cdot n$ in the other string.

\begin{algorithm}[excess parentheses matching procedure]\label{alg:decision_dyck}

On input $N \in \mathbb{N}$, $\epsilon>0$, $r \in \mathbb{N}$, the outputs of the Query procedure on two strings of length $N$, and the endpoints of two substrings $s$ and $s'$ of a string $x$ such that $s$ appears before $s'$ in $x$, proceed as follows:  
Let $n = (|s| + |s'|)/2$. 

\begin{itemize}[leftmargin=0.6cm]

\item If $r=0$: The outputs of the Query procedure contain the full strings $s$ and $s'$.
Directly check if the excess opening parentheses of $s$ and the excess closing parentheses of $s'$ match up to boundary slack $0.1\cdot \epsilon\cdot n$ \textup{(}see Definition~\ref{def:excess_parentheses_approx_match}\textup{)}, and accept if and only if this check passes.

\item Otherwise: 

\begin{enumerate}[leftmargin=0.6cm]

    \item \textup{\textsf{(Set parameters):}} Set $L,\Delta,N',\epsilon'$ exactly as in Step~1 of the Query procedure \textup{(}Algorithm~\ref{alg:query}\textup{)}.
    
    \item \textup{\textsf{(Estimate the excess numbers):}}
    For every substring $y$ of either $s$ or $s'$, use the estimation queries in the outputs of the Query procedure to compute estimates $\hat{e}_0(y)$ and $\hat{e}_1(y)$ of the numbers of excess opening and closing parentheses in $y$, respectively, as guaranteed by Claim~\ref{claim:successful_estimates_excess}.

    \item \textup{\textsf{(Define excess-parentheses segments for every offset):}}
    For every $h$, define the excess-parentheses segments of $s$ and $s'$ with offset $h$, as follows.

    \begin{itemize}
        \item For $s$, choose endpoints along $s$, from right to left, where the $k\th$ endpoint $\ell_k$ satisfies: 
        $$
            \hat{e}_0(s_{>\ell_k}) \in h+(k-1)\cdot L \pm \Delta
            \qquad\text{and}\qquad
            \hat{e}_1(s_{>\ell_k}) \leq \Delta,\footnote{
                Note that if the estimates are successful then such an endpoint $\ell_k$ can always be found: the position $\ell$ whose \emph{exact} excess numbers satisfy $e_0(s_{> \ell}) =  h + (k-1)\cdot L$ and $e_1(s_{> \ell}) = 0$ will satisfy these constraints (but other positions may satisfy the constraints as well). 
            }
        $$
        and define the $k\th$ segment of $s$ as $E^{s,h}_k \eqdef [\ell_{k+1}+1,\ell_k]$.

        \item Symmetrically, for $s'$, choose endpoints along $s'$, from left to right, where the $k\th$ endpoint $\ell'_k$ satisfies: 
        $$
            \hat{e}_1(s'_{<\ell'_k}) \in h+(k-1)\cdot L \pm \Delta
            \qquad\text{and}\qquad
            \hat{e}_0(s'_{<\ell'_k}) \leq \Delta
        $$
        and define the $k\th$ segment of $s'$ as $E^{s',h}_k \eqdef [\ell'_k, \ell'_{k+1}-1]$.
    \end{itemize}

    \item \textup{\textsf{(Search for a successful offset):}} 
    Fix $h' \eqdef 0.1\cdot\epsilon\cdot n$, and go over all offsets $h\in\{0,1,\ldots,0.2\cdot\epsilon\cdot n\}$.
    For each such $h$ do as follows:
    \begin{enumerate}[leftmargin=0.9cm, itemsep=5pt]
        \item Let $m_h \eqdef \min\{\hat{e}_0(s)-h,\, \hat{e}_1(s')-h'\}$. 

        \item \textup{\textsf{(Recursively check corresponding excess-parentheses segments):}}
        Go over all pairs of corresponding segments $(E^{s,h}_k,E^{s',h'}_k)$ for $k\in [m_h/L]$.
        For each such pair, check if there exists a pair of selected blocks $B$ and $B'$ in the respective outputs of the Query procedure, such that $E^{s,h}_k \subseteq B$ and $E^{s',h'}_k \subseteq B'$, where $B$ and $B'$ are represented with respect to their positions in $s$ and $s'$.
        If so, recursively invoke the excess parentheses matching procedure on inputs $N',\epsilon',r-1$, the outputs of the Query procedure in the recursive calls with $B$ and $B'$, and the positions of $E^{s,h}_k$ and $E^{s',h'}_k$ inside the global string $x$.

        \item If any recursive call rejects, then $h$ fails; continue to the next offset.
        Otherwise \textup{(}i.e., all recursive calls accepted\textup{)}, $h$ is successful; halt and accept.
    \end{enumerate}
    If every offset $h$ failed then reject.

\end{enumerate}
\end{itemize}
\end{algorithm}

 \medskip
The following claim states the guarantees of the excess parentheses matching procedure, which are analogous to the guarantees of the \texttt{ResStringEq} Decision procedure (Algorithm~\ref{alg:decision}). 
The proof is analogous to the analysis in Section~\ref{subsubsec:ub:actual:correctness} and~\cite[Sec.~4.2]{FMS18}, and we omit the details.

\begin{subclaim}[completeness and soundness of the excess parentheses matching procedure]\label{claim:dyck_decision_soundness}
    Let $x$ be a string over $\Sigma_m$, and let $s$ and $s'$ be two substrings of $x$. 
    Denote $n=(|s|+|s'|)/2$.
    Suppose that Algorithm~\ref{alg:decision_dyck} is invoked with the outputs obtained by running the Query procedure \textup{(}Algorithm~\ref{alg:query}\textup{)} on two substrings of $x$ of length $N$ containing $s$ and $s'$, respectively, where both invocations use proximity parameter $\epsilon$ and failure parameter $\eta$.
    Furthermore, suppose that $s$ and $s'$ are promised to satisfy $|e_0(s) - e_1(s')| \leq 0.2\cdot \epsilon \cdot n$, and that $n=\Omega(\epsilon\cdot N)$.
    Then Algorithm~\ref{alg:decision_dyck} satisfies the following guarantees \textup{(}where the probability is over the randomness of the Query procedure\textup{)}.
    \begin{enumerate}
        \item It \emph{accepts} w.p.\ at least $1-\eta$ if the excess opening parentheses of $s$ and the excess closing parentheses of $s'$ match up to boundary slack $0.1\cdot \epsilon\cdot n$ \textup{(}see Definition~\ref{def:excess_parentheses_approx_match}\textup{)}.
        
        \item If it \emph{accepts} w.p.\ greater than $\eta$, then there exists a non-crossing matching between the excess opening parentheses of $s$ and the excess closing parentheses of $s'$ such that (a) for every matched pair of parentheses, the two parentheses have the same type, and the number of parentheses between them in $x$ is even; (b) at most $\epsilon\cdot n$ of the excess parentheses in $s$ and $s'$ are unmatched.
    \end{enumerate}

\end{subclaim}

\subsection{The tester}

We finally present the non-adaptive tester for $\texttt{Dyck}_m$-\texttt{consistency} (recall that Lemma~\ref{lem:dyck_consistency_reduction} establishes a reduction from testing $\texttt{Dyck}_m$ to testing $\texttt{Dyck}_m$-\texttt{consistency}).
The tester that we will present uses recursive calls.
Therefore, aside from the usual inputs (i.e., $n\in\mathbb{N},\epsilon>0$ and query access to $x\in \Sigma_m^n$), it receives two additional parameters: (1) a round parameter $r\in \mathbb{N}$ indicating the remaining number of recursion rounds, initialized to $3$; and (2) a failure parameter $\eta > 0$ indicating the allowed failure probability, initialized to $1/3$.

\begin{algorithm}[an $O(n^{1/2+\delta})$-query non-adaptive tester for $\texttt{Dyck}_m$-\texttt{consistency}]\label{alg:dyck_tester}
    
On input $n\in \mathbb{N}$, $\epsilon>0$ and query access to $x\in \Sigma_m^n$, as well as parameters $r \in \mathbb{N}$ \textup{(}initialized to $3$\textup{)} and $\eta>0$ \textup{(}initialized to $1/3$\textup{)}, proceed as follows: 

\begin{itemize}[leftmargin=0.55cm]

\item If $r = 0$: Query the entire input string $x$, and accept if and only if $x$ is $\textup{\texttt{Dyck}}_m$-consistent.

\item Otherwise:

\begin{enumerate}[leftmargin=0.55cm]

    \item Let $b \,\eqdef\, n^{3/4}$, $\epsilon' \,\eqdef\, \alpha_1 \cdot \epsilon$,\, $\eta' \,\eqdef\, \eta/n$, $L \,\eqdef\, \Omega(\epsilon' \cdot b)$, and $\Delta \,\eqdef\, \alpha_2 \cdot \epsilon' \cdot L$, where $\alpha_1, \alpha_2>0$ are sufficiently small constants.
    
    \item \textup{\textsf{Query-Phase:}} 
    \begin{enumerate}
        \item Take $T = O\!\left(\frac{n^2}{\Delta^2}\log (n/\eta)\right)$ independent uniform samples from $[n]$, and query $x$ at the sampled positions.
        
        \item Partition $x$ into consecutive blocks of length $b$.
        Sample each block independently with probability $p = O\left(\sqrt{\frac{b\cdot\log(1/\eta)}{\epsilon^2 \cdot n}}\right)$.

        \item For each selected block $B$, run the Query procedure \textup{(}Algorithm~\ref{alg:query}\textup{)} with query access to the restriction of $x$ to block $B$, proximity parameter $\epsilon'$, failure parameter $\eta'$, 
        and round parameter $O(\log(1/\delta))$.
    \end{enumerate}

    \item \textup{\textsf{Decision-Phase:}} 
    \begin{enumerate}
        \item Using the estimation queries made in Step~2a, compute the estimates of the excess numbers of every consecutive substring of $x$, with additive deviation at most $\Delta$ and failure probability at most $\eta/2$, as guaranteed by Claim~\ref{claim:successful_estimates_excess}.
    
        \item For each pair of blocks $i<j$, compute the approximate matching intervals $\hat{I}_{i,j}$ and $\hat{I}_{j,i}$ as guaranteed by Claim~\ref{claim:matching_intervals}. 

        \item For each pair of blocks $i<j$ such that $\max\{|\hat{I}_{i,j}|, |\hat{I}_{j,i}|\} \geq L$, if both blocks $i$ and $j$ were selected in Step~2b, then run the excess parentheses matching procedure \textup{(}Algorithm~\ref{alg:decision_dyck}\textup{)} 
        with the same parameters used for the Query procedure in Step~2c, 
        the outputs of the Query procedure when run on the restrictions of $x$ to blocks $i$ and $j$, as well as the positions of the intervals $\hat{I}_{i,j}$ and $\hat{I}_{j,i}$ inside $x$.
        If any call rejects, then halt and reject. 
        
        \item Select $O(\epsilon^{-1} \cdot \log(1/\eta))$ uniformly random blocks.
        For each selected block $B$, recursively run the $\textup{\texttt{Dyck}}_m$-\textup{\texttt{consistency}} tester \textup{(}Algorithm~\ref{alg:dyck_tester}\textup{)} on the restriction of $x$ to $B$, using proximity parameter $\epsilon'$, 
        failure parameter $\eta'$ 
        and round parameter $r-1$.
        Accept if and only if all recursive calls accept.
    \end{enumerate}
    
\end{enumerate}

\end{itemize}

\end{algorithm}

\subsubsection{Query complexity} 
Let us  bound the query complexity of Algorithm~\ref{alg:dyck_tester}.
We first bound the expected number of queries made in a \emph{single} recursion round.
Notice that $\Delta = \Omega(\epsilon^2\cdot b)$, and therefore in Step~2a Algorithm~\ref{alg:dyck_tester} makes $O\left(\frac{n^2\cdot\log(n/\eta)}{\epsilon^4\cdot b^2}\right)$ queries.
In Step~2b the number of blocks sampled in expectation is $p\cdot \frac{n}{b} = O\Big(\sqrt{\frac{n\cdot \log(1/\eta)}{\epsilon^2\cdot b}}\Big)$.
Since each block is of length $b$, the expected number of queries made in Step~2c is $O\Big(\sqrt{\frac{n\cdot \log(1/\eta)}{\epsilon^2\cdot b}}\Big)\cdot Q(b, \epsilon', \eta')$, where $Q(b, \epsilon', \eta') = Q(b, \Omega(\epsilon), \Omega(\eta/n))$ is the expected query complexity of the Query procedure (Algorithm~\ref{alg:query}) invoked in Step 2c.
As established in Section~\ref{sec:query_complex_recursive}, $Q(b, \Omega(\epsilon), \Omega(\eta/n)) = \left(1/\epsilon\right)^{O(\log(1/\delta))}\cdot \log^{O(\log(1/\delta))}\left(n/\eta\right)\cdot b^{1/2 + \delta}$.
Since we set $b=n^{3/4}$, we get that the expected number of queries made in a {single recursion round} is at most
$
    O\Big(\frac{n^2\cdot\log(n/\eta)}{\epsilon^4\cdot b^2} + \sqrt{\frac{n\cdot \log(1/\eta)}{\epsilon^2\cdot b}}\cdot Q(b, \Omega(\epsilon), \Omega(\eta/n))\Big) 
    \leq \left(1/\epsilon\right)^{O(\log(1/\delta))}\cdot \log^{O(\log(1/\delta))}\left(n/\eta\right)\cdot n^{1/2 + \delta}.
$

Now, in each recursion round, Algorithm~\ref{alg:dyck_tester} makes $O(\epsilon^{-1} \cdot \log(1/\eta))$ recursive calls.
In the recursive calls, the input length shrinks from $n$ to $b = n^{3/4}$. 
After $3$ recursion rounds, the input length becomes $n^{(3/4)^3} < n^{1/2}$, and the input is read entirely. 
Note that throughout all 3 recursion rounds, the proximity parameter remains $\Omega(\epsilon)$ (where $\epsilon$ is the initial proximity parameter), and the failure parameter $\eta$ remains $1/\textrm{poly}(n)$. 
Thus, altogether, the expected query complexity of Algorithm~\ref{alg:dyck_tester} is $\left(1/\epsilon\right)^{O(\log(1/\delta))}\cdot \log^{O(\log(1/\delta))}(n)\cdot n^{1/2 + \delta}$.
Since this bound holds for every constant $\delta>0$, we may apply it with $\delta/2$ in place of $\delta$ and use $\log^{O(\log(1/\delta))}(n) = O(n^{\delta/2})$, giving us query complexity
$
    O(\left(1/\epsilon\right)^{O(\log(1/\delta))}\cdot n^{1/2 + \delta})
$.

\subsubsection{Correctness}

The correctness analysis of Algorithm~\ref{alg:dyck_tester} is very similar to that in~\cite{FMS18}. 
We will detail only the places where some modifications to the analysis in~\cite{FMS18} are needed, and indicate where the analysis proceeds exactly as in~\cite{FMS18}.

The completeness of Algorithm~\ref{alg:dyck_tester} follows straightforwardly from the completeness of the excess parentheses matching procedure, similarly to~\cite{FMS18}.
We therefore focus on the soundness analysis.
Let $x\in\Sigma_m^n$, and consider a partitioning of $x$ into consecutive blocks of length $b$. As in~\cite{FMS18}, we say that a parenthesis in $x$ is \textup{\textsf{locally excess}} if it is excess in its block but is not excess in $x$.
Equivalently, the set of locally excess parentheses consists of the excess parentheses that lie in the matching intervals $I_{i,j}$ and $I_{j,i}$, for all $i<j$ (see Claim~\ref{claim:matching_intervals}).
We next prove the following claim, which is analogous to~\cite[Lem.~3.6]{FMS18}.

\begin{subclaim}\label{claim:locally_excess_matching}
    Suppose that Algorithm~\ref{alg:dyck_tester} is invoked on a string $x\in\Sigma_m^n$ with proximity parameter $\epsilon$ and failure parameter $\eta$.
    If the algorithm accepts $x$ with probability greater than $\eta$, then there exists a non-crossing matching on the locally excess parentheses of $x$ such that
    (a) for every matched pair of parentheses, the two parentheses have the same type, and the number of parentheses in $x$ between them is even;
    (b) at most $0.1\cdot\epsilon\cdot n$ of the locally excess parentheses of $x$ are unmatched.
\end{subclaim}

From the above claim, the soundness of Algorithm~\ref{alg:dyck_tester} follows by the same argument as in~\cite[Sec.~3.2]{FMS18}.

\begin{subproof}[\textsf{\textup{\large Proof of Claim~\ref{claim:locally_excess_matching}:}}]
    Since the estimates of the excess numbers computed in Step~3a of Algorithm~\ref{alg:dyck_tester} are successful with probability at least $1-\eta/2$, and $x$ is accepted with probability greater than $\eta$, there exists a fixing of the estimates such that the estimates are successful and, under this fixing, $x$ is accepted with probability greater than $\eta/2$.
    Fix such estimates, and consider the resulting approximate matching intervals $(\hat{I}_{i,j},\hat{I}_{j,i})$, for every $i<j$.

    We call a pair of blocks $(i,j)$ \textsf{good} if there exists a non-crossing matching between the excess opening parentheses in the restriction of $x$ to $\hat{I}_{i,j}$ and the excess closing parentheses in the restriction of $x$ to $\hat{I}_{j,i}$ such that 
    (a) for every matched pair of parentheses, the two parentheses have the same type, and the number of parentheses in $x$ between them is even;
    and 
    (b) at most $\epsilon' \cdot b$ of these excess parentheses are unmatched. Otherwise, we call the pair \textsf{bad}. 

    Note that in every bad pair both approximate matching intervals contain at least $\Omega(\epsilon'\cdot b) = L$ excess parentheses: 
    Since the estimates are successful, the numbers of excess parentheses in the two intervals can differ by at most $O(\Delta)$ (by the guarantees of Claim~\ref{claim:matching_intervals}). 
    Thus, both intervals must contain at least $(\epsilon'\cdot b - O(\Delta))/2 = \Omega(\epsilon'\cdot b)$ parentheses or the pair would be trivially good.

    We next claim that there are at most $\epsilon'\cdot n/b$ bad pairs.
    Indeed, suppose to the contrary that there are more than $\epsilon'\cdot n/b = \Omega(\epsilon\cdot n/b)$ bad pairs. 
    Note that because each interval of a bad pair contains at least $\Omega(\epsilon'\cdot b)= \Omega(\epsilon\cdot b)$ parentheses, it holds that each block can participate in at most $O(1/\epsilon)$ bad pairs.
    Thus, there is a set $S$ of at least $\Omega\left(\epsilon^2\cdot n/b\right)$ bad pairs that are pair-wise \emph{disjoint} (i.e., each block appears in at most one pair in $S$).
    Setting the sampling probability in Step~2b (of Algorithm~\ref{alg:dyck_tester}) to $p = C\cdot\sqrt{\frac{b\cdot\log(1/\eta)}{\epsilon^2 \cdot n}}$, where $C$ is a sufficiently large constant, we get that the probability that no pair in $S$ has both of its blocks sampled is at most
    $$
        \left(1-p^2\right)^{\Omega\left(\epsilon^2\cdot n/b\right)}
        \leq e^{-\Omega\left(p^2\cdot\epsilon^2\cdot n/b\right)}
        = e^{-\Omega\left(C^2\cdot{\log(1/\eta)}\right)}
        \leq \eta/4.
    $$
    By Claim~\ref{claim:dyck_decision_soundness}, the excess parentheses matching procedure accepts a bad pair with probability at most $\eta' = \eta/n \leq \eta/4$.
    It follows that the tester will accept with probability at most $\eta/4+\eta/4 = \eta/2$, contradicting our choice of the fixed estimates. 
    Hence, there are at most $\epsilon'\cdot n/b$ bad pairs. Since every pair of blocks can contain at most $O(b)$ excess parentheses, the bad pairs contain at most $O(\epsilon'\cdot n)$ locally excess parentheses in total. 

    We combine the matchings guaranteed for all good pairs (by the definition of good) to create the desired matching on the locally excess parentheses of $x$. 
    Recall that by Claim~\ref{claim:matching_intervals} it holds that for every $i<j$, the approximate matching intervals $\hat{I}_{i,j}$ and $\hat{I}_{j,i}$ are contained in the respective (exact) matching intervals $I_{i,j}$ and $I_{j,i}$.
    Thus, the approximate intervals of different pairs are \emph{disjoint}, and hence combining the different matchings creates a legal matching (i.e., no parenthesis is matched to more than one parenthesis).
    Furthermore, since the matching graph is planar, the resulting matching is non-crossing.

    It remains to bound the number of locally excess parentheses that are not included in the resulting matching. 
    First, for each good pair $i<j$, at most $\epsilon'\cdot b$ of the excess parentheses in $\hat{I}_{i,j}$ and $\hat{I}_{j,i}$ are unmatched. 
    Furthermore, by Claim~\ref{claim:matching_intervals}, the (exact) matching intervals ${I}_{i,j}$ and ${I}_{j,i}$ contain at most $O(\Delta) = O(\epsilon'\cdot b)$ excess parentheses outside $\hat{I}_{i,j}$ and $\hat{I}_{j,i}$.
    Hence, for each good pair $i<j$, at most $O(\epsilon'\cdot b)$ of the excess parentheses in ${I}_{i,j}$ and ${I}_{j,i}$ are unmatched.  
    By Claim~\ref{claim:planar}, the matching graph has at most $O(n/b)$ edges.
    Therefore, altogether, the number of unmatched locally excess parentheses belonging to good pairs is at most $O(\epsilon'\cdot n)$.
    Together with the at most $O(\epsilon'\cdot n)$ locally excess parentheses belonging to bad pairs, the total number of unmatched locally excess parentheses is at most $O(\epsilon'\cdot n)$.
    Recalling that $\epsilon' = \alpha_1 \cdot \epsilon$, we set $\alpha_1$ sufficiently small so that the claim follows. 
\end{subproof}

\subsection{On the adaptive tester for the \texttt{Dyck} languages}\label{subsec:adaptive_dyck}

Recall that~\cite{FMS18} showed an \emph{adaptive} tester for the \texttt{Dyck} languages with query complexity $O(\left(1/\epsilon\right)^{\textup{\textrm{poly}}(1/\delta)} \cdot n^{2/5 + \delta})$.
Their tester uses a procedure called ``Substring $\epsilon$-matching'', which has essentially the same guarantees as our ``excess parentheses matching'' procedure (Algorithm~\ref{alg:decision_dyck}) when the latter procedure is applied to the output of our Query procedure (Algorithm~\ref{alg:query}).
The query complexity of the ``Substring $\epsilon$-matching'' procedure is $O(\left(1/\epsilon\right)^{\textup{\textrm{poly}}(1/\delta)} \cdot n^{1/2 + \delta})$.
Through a more careful analysis and parameter choice, the query complexity of our procedure has an improved dependence on $\epsilon$; namely, our procedure's query complexity is $O(\left(1/\epsilon\right)^{O(\log(1/\delta))} \cdot n^{1/2 + \delta})$.\footnote{
    Specifically, the improvement comes from the fact that in the procedure of~\cite{FMS18}, the recursive call uses proximity parameter $\Theta(\epsilon^2)$, whereas our analysis allows us to use proximity parameter $\Theta(\epsilon)$.
} 
By replacing the ``Substring $\epsilon$-matching'' procedure in the adaptive tester of~\cite{FMS18} with our procedure, we obtain:

\begin{theorem}\label{thm:upper_bound_dyck:adaptive}
    For every $m\in \mathbb{N}$, and for any \textup{(}arbitrarily small\textup{)} constant $\delta>0$, there exists an \textup{(}adaptive\textup{)} tester for $\textup{\texttt{Dyck}}_m$ with query complexity
    $O(\left(1/\epsilon\right)^{O(\log(1/\delta))} \cdot n^{2/5 + \delta})$.
\end{theorem}

%% file: J_acknowledgments.tex
\section*{Acknowledgments}
\addcontentsline{toc}{section}{Acknowledgments}

I am grateful to my advisor, Oded Goldreich, for his guidance throughout the development of this work, and for very helpful feedback on its presentation.
I am grateful to Gil Aharoni for valuable suggestions on the presentation.
The results in this work were obtained with assistance from GPT 5.5 and 6.0.

%% file: K_equal_length_reduction.tex
\section{A reduction to the equal-length case in \texttt{ResStringEq}}\label{apdx:equal_length}

In Definition~\ref{def:rse}, the property \texttt{ResStringEq} is defined over pairs of strings $(s,s')$ of equal length. 
An alternative definition would allow $s$ and $s'$ to differ in length.
The purpose of this short appendix is to show that $\epsilon$-testing this more general definition reduces to $\epsilon/2$-testing the equal-length case, by padding the shorter string with star (i.e., `$*$') symbols.
Clearly, the padding preserves membership in \textup{\texttt{ResStringEq}}.
On the other hand, consider any modifications to the symbols of the padded pair that make the resulting residual strings equal.
Notice that we may assume that no star symbol is modified:
if a star symbol in one string is changed into a non-star symbol, we can instead replace the symbol matched with it in the other string by a star symbol.
Thus, the padding symbols are not modified, implying the same modifications can be made in the original pair to make its residual strings equal.
Since padding at most doubles the input length, an $\epsilon$-far pair remains $\epsilon/2$-far.

%% file: L_sd_fact.tex
\section{Proof of Fact~\ref{fact:stat_dist_adaptive}}\label{apdx:proof_of_stat_dist_fact}

This section provides a proof of Fact~\ref{fact:stat_dist_adaptive} (restated below), which generalizes Fact~\ref{fact:stat_dist}.

\begin{fact}[Fact~\ref{fact:stat_dist_adaptive}, restated]\label{fact:stat_dist_adaptive_restated}
    Let $X$ and $X'$ be two discrete random variables supported over a domain $D$, and let $\mathcal{E}$ and $\mathcal{E}'$ be two events over the probability spaces of $X$ and $X'$, respectively. If for every $x\in D$ it holds that $\Pr[X=x, \ovbar{\mathcal{E}}] = \Pr[X' = x, \ovbar{\mathcal{E}}']$, then $\Pr\left[\mathcal{E}\right] = \Pr\left[\mathcal{E}'\right]$ and $\Delta(X,X') \leq \Pr\left[\mathcal{E}\right]$.
\end{fact}

\begin{proof}
    First, notice that $\Pr[\ovbar{\mathcal{E}}] = \sum_{x\in D}{\Pr[X = x, \ovbar{\mathcal{E}}]} = \sum_{x\in D}{\Pr[X' = x, \ovbar{\mathcal{E}}']} =\Pr[\ovbar{\mathcal{E}}']$, and thus $\Pr[\mathcal{E}] = \Pr[\mathcal{E}']$.
    Now, it holds that:
    \begin{align*}
        \Delta(X, X') &= \frac{1}{2}\cdot \sum_{x\in D} | \Pr[X=x] - \Pr[X'=x] | \\
                      &= \frac{1}{2}\cdot \sum_{x\in D} | \Pr[X=x, \mathcal{E}] + \Pr[X=x, \ovbar{\mathcal{E}}]  - (\Pr[X'=x, \mathcal{E}'] + \Pr[X'=x, \ovbar{\mathcal{E}}']) | \\
                      &= \frac{1}{2}\cdot \sum_{x\in D} | \Pr[X=x, \mathcal{E}] - \Pr[X'=x, \mathcal{E}'] | \\
                      &\leq \frac{1}{2}\cdot \sum_{x\in D} \left(\, \Pr[X=x, \mathcal{E}] + \Pr[X'=x, \mathcal{E}'] \,\right) \\
                      &= \frac{1}{2}\cdot \big( \Pr[\mathcal{E}] + \Pr[\mathcal{E}'] \,\big) 
    \end{align*}
    where the inequality is by the triangle inequality. 
    Since $\Pr[\mathcal{E}] = \Pr[\mathcal{E}']$, the claim follows.
\end{proof}